\documentclass[11pt,a4paper]{article}

\usepackage{amssymb,amsthm,amsmath,mathtools}
\usepackage{thmtools,thm-restate}
\usepackage[round]{natbib}%
\usepackage{pifont}
\newcommand{\cmark}{\ding{51}}
\newcommand{\xmark}{\ding{55}}
\usepackage[usenames,svgnames,xcdraw,table]{xcolor}
\definecolor{DarkGreen}{rgb}{0.1,0.5,0.1}
\usepackage[backref=page]{hyperref}
\hypersetup{
	colorlinks=true,
	linkcolor=red,%
	urlcolor=DarkGreen,%
	citecolor=blue%
}
\renewcommand*{\backref}[1]{}
\renewcommand*{\backrefalt}[4]{%
    \ifcase #1 (Not cited.)%
    \or        (Cited on page~#2)%
    \else      (Cited on pages~#2)%
    \fi}
\usepackage[capitalise,noabbrev]{cleveref}
\Crefname{property}{Property}{Properties}
\Crefname{example}{Example}{Examples}
\Crefname{table}{Table}{Tables}

\usepackage{cancel}
\usepackage{tabularx}
\usepackage{nicefrac}
\usepackage{enumerate}
\usepackage{etoolbox}%
\usepackage[margin=0.99in]{geometry}
\usepackage[T1]{fontenc}
\usepackage[tt=false]{libertine}

\usepackage{url}
\usepackage[utf8]{inputenc}
\usepackage[small]{caption}
\usepackage{graphicx}
\usepackage{booktabs}
\usepackage{mathrsfs,amsfonts,dsfont,authblk}
\usepackage{array,multirow,graphicx,bigdelim}

\usepackage{algorithm}
\usepackage[noend]{algpseudocode}%

\usepackage{tikz}
\usepackage{tikz-cd}
\usetikzlibrary{fit,calc,shapes,decorations.text,arrows,decorations.markings,decorations.pathmorphing,shapes.geometric,positioning,decorations.pathreplacing}
\tikzset{snake it/.style={decorate, decoration=snake}}

\colorlet{mygray}{gray!40}

\makeatletter
\let\oldnl\nl%
\newcommand{\nonl}{\renewcommand{\nl}{\let\nl\oldnl}}%
\makeatother

  {\list{}{\leftmargin=#1\rightmargin=#1}\item[]}%
  {\endlist}

\usepackage[labelfont={normalfont,bf},textfont=it]{caption}
\usepackage{subcaption}

\newtheorem{lemma}{Lemma}%

\newtheorem{proposition}{Proposition}%
\theoremstyle{definition}

\newtheorem{examplex}{Example}
\newenvironment{example}{\pushQED{\qed}\examplex}{\popQED\endexamplex}
\theoremstyle{remark}
\newtheorem{remark}{Remark}

\usepackage{flushend}
\usepackage[utf8]{inputenc}
\Crefname{claim}{Claim}{Claims}

\allowdisplaybreaks
\renewcommand{\>}{\succ}

\newcommand{\DA}{\texttt{DA}}
\newcommand{\NonProp}{\textup{\textrm{Non-Prop}}}

\renewcommand{\O}{\mathcal{O}}
\newcommand{\Prop}{\textup{\textrm{Prop}}}

\usepackage{comment}

\begin{document}

\title{Two for One \& One for All:\\ Two-Sided Manipulation in Matching Markets}
\date{}

\author[1]{Hadi Hosseini}
\author[2]{Fatima Umar}
\author[3]{Rohit Vaish}
\affil[1]{Pennsylvania State University\\
	{\small\texttt{hadi@psu.edu}}}
\affil[2]{Rochester Institute of Technology\\
	{\small\texttt{fu1476@rit.edu}}}
\affil[3]{Indian Institute of Technology Delhi\\
	{\small\texttt{rvaish@iitd.ac.in}}}

\maketitle

\begin{abstract}
Strategic behavior in two-sided matching markets has been traditionally studied in a ``one-sided'' manipulation setting where the agent who misreports is also the intended beneficiary. Our work investigates ``two-sided'' manipulation of the deferred acceptance algorithm where the misreporting agent and the manipulator (or beneficiary) are on different sides. Specifically, we generalize the recently proposed \emph{accomplice manipulation} model (where a man misreports on behalf of a woman) along two complementary dimensions: (a) the \emph{two for one} model, with a pair of misreporting agents (man and woman) and a single beneficiary (the misreporting woman), and (b) the \emph{one for all} model, with one misreporting agent (man) and a coalition of beneficiaries (all women).

Our main contribution is to develop polynomial-time algorithms for finding an optimal manipulation in both settings. We obtain these results despite the fact that an optimal \emph{one for all} strategy fails to be \emph{inconspicuous}, while it is unclear whether an optimal \emph{two for one} strategy satisfies the inconspicuousness property. We also study the conditions under which stability of the resulting matching is preserved. Experimentally, we show that two-sided manipulations are more frequently available and offer better quality matches than their one-sided counterparts.
\end{abstract}

\section{Introduction}

The \emph{deferred acceptance} algorithm~\citep{GS62college} 
is one of the biggest success stories of matching theory and market design. It has profoundly impacted numerous practical applications including \emph{school choice}~\citep{APR05new,APR+05boston} and \emph{entry-level labor markets}~\citep{RP99redesign} %
and has inspired a long line of work in economics and computer science~\citep{GI89stable,RS92two,R08deferred,M13algorithmics}.

The success of the deferred acceptance (or \DA{}) algorithm has been driven by its \emph{stability} property, which prevents pairs of agents %
from preferring each other over their assigned partners. Stability eliminates the incentives for agents to participate in secondary markets or `scrambles'~\citep{KPR13matching}, and has been a key predictor of
the long-term sustenance of many real-world matching markets~\citep{R02economist}.%

Unfortunately, \emph{any} stable matching algorithm is known to be vulnerable to strategic misreporting of preferences by the agents~\citep{R82economics}. For the \DA{} algorithm, in particular, it is known that truth-telling is a dominant strategy for the proposing side---colloquially, the \emph{men}---implying that any strategic behavior is confined to the proposed-to side---the \emph{women}~\citep{DF81machiavelli,R82economics}. 

\paragraph{One-sided vs. two-sided manipulation.} Given the strong practical appeal of \DA{} algorithm, significant research effort has been devoted towards understanding its incentive properties. Much of this work has focused on ``one-sided'' manipulation wherein the agent who misreports %
 is also the intended beneficiary; %
 that is, the misreporting agent and the beneficiary are on the \emph{same} side. For this \emph{self manipulation} problem, the structural and computational questions have been extensively studied~\citep{DF81machiavelli,GS85ms,GS85some,TS01gale,VG17manipulating}.

By contrast, there are many real-world settings that, in essence, resemble ``two-sided'' manipulations where the misreporting agent and the beneficiary are on \emph{different} sides. For example, in the student-proposing school choice, schools could influence the preferences of students they find ``undesirable'' (such as those from low-income backgrounds) by using indirect measures such as fee hike%
~\citep{HKN16improving}. In ridesharing platforms~\citep{banerjee2019ride}, %
a driver may influence the preferences of certain riders by strategically moving to a farther distance. Similarly, in a gig economy, freelancers' preferences over tasks may be affected by an employer's restrictive requirements~(e.g., demanding additional IDs or enforcing zoning).

Motivated by these examples, recent works have studied the \emph{accomplice manipulation} model wherein a man misreports his preferences in order to help a specific woman~\citep{BH19partners,HUV21accomplice}. It has been shown via simulations that accomplice manipulation strategies are \emph{more frequently available} than self manipulation and result in \emph{better matches} for the woman. Further, an optimal misreport for the accomplice %
is known to be \emph{inconspicuous} (i.e., the manipulated list can be derived from his true list by promoting exactly one woman), \emph{efficiently computable}, and \emph{stability-preserving} (i.e., the manipulated \DA{} matching is stable with respect to the true preferences).

\paragraph{Towards coalitional two-sided manipulation.} 
The aforementioned advantages of two-sided manipulation call for a deeper investigation into the topic. Our work takes a step in this direction by focusing on \emph{coalitional aspects} of the two-sided manipulation problem.

Our starting point is the accomplice manipulation model~\citep{HUV21accomplice}, which involves one misreporting agent and one beneficiary (i.e., a \emph{one for one} setting). We consider two coalitional generalizations of this model: %
(i)~The \emph{two for one} model, with a coalition of two misreporting agents (a man and a woman) and a single beneficiary (the woman), and (ii) the \emph{one for all} model, with one misreporting agent (man) and a coalition of beneficiaries (all women).

Coalitional manipulation of the \DA{} algorithm has received considerable theoretical interest over the years~\citep{DF81machiavelli,GS85ms,
GS85some,DGS87further,KM09successful,KM10cheating}, and recently its practical relevance has also been discussed. Indeed, in college admissions in China, universities have been known to form ``leagues'' for conducting independent recruitment exams allowing them to jointly manipulate admission results~\citep{SDT21coalitional}. Prior work on coalitional manipulation has focused exclusively on one-sided manipulation. %

\begin{table}[t]
\centering
\footnotesize
\begin{tabular}{|cc|ccc|c|}
\hline
\multirow{2}{*}{\textbf{Who misreports?}} & \multirow{2}{*}{\textbf{Who benefits?}}  & \multicolumn{3}{c|}{\textbf{Results for optimal manipulation}} &  \multirow{2}{*}{\textbf{Reference}} \\
&& Inconspicuous? & Poly-time? & Stability-preserving? &\\
\hline
Woman $w$ & Woman $w$ & \cmark & \cmark & \cmark & \citet{VG17manipulating} \\
Man $m$ & Woman $w$ & \cmark & \cmark & \cmark & \citet{HUV21accomplice} \\
Man $m$ and woman $w$ & Woman $w$ & Open & \cmark & \xmark & \textbf{This paper}\\
Man $m$ & All women & \xmark & \cmark & \cmark & \textbf{This paper}\\
\hline
\end{tabular}
\caption{Summary of previously known (top two rows) and new results (bottom two rows).%
}
\label{tab:Summary}
\end{table}

\paragraph{Our contributions.}
We study two coalitional generalizations of accomplice manipulation %
and make the following theoretical and experimental contributions (see \Cref{tab:Summary}):

\begin{itemize}

    \item \textbf{Two for one}: We show that when the accomplice and the beneficiary can jointly misreport, an optimal \emph{pair manipulation} strategy can be strictly better than either of the optimal individual (i.e., self or accomplice) manipulation strategies~(\Cref{eg:pair_manipulation}). 
    In contrast to both self and accomplice manipulation, an optimal pair manipulation may not be stability-preserving~(\Cref{rem:Unstable_Pair}) and it is unclear whether it is inconspicuous. Nevertheless, we provide a polynomial-time algorithm for computing an optimal pair manipulation~(\Cref{thm:PairManipulation}).

    \item \textbf{One for all}: We observe that optimal  manipulation by the accomplice for helping all women  %
    could fail to be inconspicuous~(\Cref{eg:inconspicuous_is_suboptimal}). By contrast, when helping a single woman, an optimal strategy for the accomplice is known to be inconspicuous~\citep{HUV21accomplice}.
    
    Despite losing this structural benefit, we develop a polynomial-time algorithm for computing an optimal one-for-all strategy~(\Cref{cor:PolyAlgorithm_OptimalStrategy}), and show that such a strategy is stability-preserving (\Cref{cor:Min_NoRegret_PushUp_StabilityPreserving}). In fact, we show a stronger result that a \emph{minimum} optimal misreport~(i.e., one that pushes up as few women as possible) can be efficiently computed and provide tight bounds on the size of the promoted set~(\Cref{sec:Smallest-size_One-for-all}).
    
    \item \textbf{Experiments}: We show via simulations on uniformly random preferences that two-sided strategies are more frequently available~(\Cref{fig:two_sided_frequencies}) and result in better quality matches~(\Cref{sec:Additional_Experiments}) than one-sided strategies.
\end{itemize}

\paragraph{Related Work.}
The literature on strategic aspects of stable matching procedures has classically focused on \emph{truncation} strategies where the misreported list is a prefix of the true list~\citep{DF81machiavelli,R82economics,RR99truncation}. %
Our work, on the other hand, focuses on \emph{permutation manipulation} where the manipulated list is a reordering of the true list.

\citet{TS01gale} initiated the study of permutation manipulation by a \emph{single} woman and provided a polynomial-time algorithm for finding an optimal misreport. \citet{VG17manipulating} showed that an optimal strategy for the woman is inconspicuous and stability-preserving. %
Permutation manipulation by a \emph{group} of agents has been studied for 
a coalition of men~\citep{H06cheating,H07cheating} and for a coalition of women~\citep{KM09successful,KM10cheating,SDT21coalitional}. %
In particular, \citet{SDT21coalitional} provided an algorithm for finding a strategy for a coalition of women that is Pareto optimal among all stability-preserving strategies, and showed that such a strategy is inconspicuous.

In the \emph{two-sided} setting, \citet{BH19partners} introduced the accomplice manipulation model and observed that it can be more beneficial for a woman than optimal self manipulation. Subsequently, \citet{HUV21accomplice} studied \emph{with-regret} and \emph{no-regret} accomplice manipulation, depending on whether the accomplice's match worsens or stays the same. They showed that an optimal \emph{no-regret} manipulation is stability-preserving while its \emph{with-regret} counterpart is not, and that optimal strategies under both models are inconspicuous and therefore efficiently computable. Our work will focus exclusively on no-regret strategies.

\section{Preliminaries}

\paragraph{Problem instance.} An instance of the \emph{stable marriage problem}~\citep{GS62college} is given by a tuple $\langle M, W, \> \rangle$, where $M$ is a set of $n$ men, $W$ is a set of $n$ women, and $\>$ is a \textit{preference profile} which specifies the preference lists of the agents. The preference list of a man $m \in M$, denoted by $\>_m$, is a strict total order over all women in $W$. The list $\>_w$ of a woman $w \in W$ is defined analogously. We will write $w_1 \succeq_m w_2$ to denote ``either $w_1 \>_m w_2$ or $w_1 = w_2$'', and write $\>_{-m}$ to denote the %
profile without the list of man $m$; thus, $\> = \{\>_{-m}, \>_m\}$.

\paragraph{Stable matching.} A \emph{matching} is a function $\mu: M \cup W \rightarrow M \cup W$ such that $\mu(m) \in W$ for all $m \in M$, $\mu(w) \in M$ for all $w \in W$, and $\mu(m) = w$ if and only if $\mu(w) = m$. Given a matching $\mu$, a \emph{blocking pair} with respect to the preference profile $\>$ is a man-woman pair $(m, w)$ who prefer each other over their assigned partners, i.e., $w \>_m \mu(m)$ and $m \>_w \mu(w)$. A matching is said to be \emph{stable} if it does not have any blocking pair. We will write $S_{\>}$ to denote the set of all stable matchings with respect to $\>$. Note that in the worst case, the size of $S_{\>}$ can be exponential in $n$~\citep{K97stable}.%

For any pair of matchings $\mu,\mu'$, we will write $\mu \succeq_M \mu'$ to denote that all men weakly prefer $\mu$ over $\mu'$, i.e., $\mu(m) \succeq_m \mu'(m)$ for all $m \in M$ (analogously $\mu \succeq_W \mu'$ for women).

\paragraph{Deferred acceptance algorithm.}
The deferred acceptance~(\DA{}) algorithm is a well-known procedure for finding stable matchings~\citep{GS62college}. Given as input a preference profile, the algorithm alternates between a \emph{proposal} phase, where each currently unmatched man proposes to his favorite woman among those who haven't rejected him yet, and a \emph{rejection} phase, where each woman tentatively accepts her favorite proposal and rejects the rest. The algorithm terminates when no further proposals can be made. %

\citet{GS62college} showed that given any preference profile~$\>$ as input, the matching computed by the \DA{} algorithm, which we will denote by $\DA(\>)$, is stable. Furthermore, this matching is \emph{men-optimal} as it assigns to each man his favorite partner among all stable matchings in $S_{\>}$. Subsequently, it was observed that the same matching is also \emph{women-pessimal}~\citep{MW71stable}.

\begin{restatable}[\protect\citealp{GS62college,MW71stable}]{proposition}{DAstableMenOptimalWomenPessimal}
Let $\>$ be a %
profile and let ${\mu \coloneqq \DA(\>)}$. Then, $\mu \in S_{\>}$. Furthermore, for any $\mu' \in S_{\>}$, $\mu \succeq_M \mu'$ and $\mu' \succeq_W \mu$.
\label{prop:DA_stable_MenOptimal_WomenPessimal}
\end{restatable}

\paragraph{Self manipulation.}
Given a profile $\>$ and the matching $\mu \coloneqq \DA(\>)$, we say that woman $w$ can \emph{self manipulate} if there exists a list $\>'_w$ (which is a \emph{permutation} of $w$'s true list $\>_w$) such that $\mu'(w) \>_w \mu(w)$, where $\mu' \coloneqq \DA(\>_{-w},\>'_w)$. 
An \emph{optimal} self manipulation $\>'_w$ is one for which there is no other list $\>''_w$ such that $\mu''(w) \>_w \mu'(w)$, where $\mu'' \coloneqq \DA(\>_{-w},\>''_w)$.

\paragraph{Accomplice manipulation.} A different model of strategic behavior is \emph{accomplice manipulation}, wherein a woman $w$, instead of misreporting herself, asks a man $m$ to misreport his preference in order to improve $w$'s match~\citep{BH19partners,HUV21accomplice}. Formally, given a profile $\>$ and a fixed man $m$, we say that woman $w$ can \emph{manipulate via accomplice} $m$ if there exists a list $\>'_m$ for man $m$ (which is a permutation of his true list $\>_m$) such that $\mu'(w) \>_w \mu(w)$, where $\mu \coloneqq \DA(\>)$ and $\mu' \coloneqq \DA(\>_{-m}, \>'_m)$. 
An \emph{optimal} accomplice manipulation $\>'_m$ is one for which there is no other list $\>''_m$ such that $\mu''(w) \>_w \mu'(w)$, where $\mu'' \coloneqq \DA(\>_{-m},\>''_m)$. 
We will call woman $w$ a `beneficiary' and man $m$ an `accomplice'.%

\paragraph{No-regret assumption.}
In this paper, we will consider two generalizations of accomplice manipulation: (a) the ``two for one'' problem with two misreporting agents (man $m$ and woman $w$) and a single beneficiary (woman $w$), and (a) the ``one for all'' problem with a single misreporting agent (man $m$) and a coalition of beneficiaries (all women). In both cases, we will assume \emph{no-regret} manipulation for the man which means that $m$'s match does not worsen upon misreporting, i.e.,  $\mu'(m) \succeq_m \mu(m)$.

Interestingly, for both generalizations mentioned above, the no-regret assumption implies that the accomplice's match stays the same, i.e., $\mu(m) = \mu'(m)$. Indeed, in the one-for-all problem, it follows from the strategyproofness of \DA{} algorithm for the proposing side that $\mu(m) \succeq_m \mu'(m)$. Along with the no-regret assumption, this implies $\mu(m) = \mu'(m)$. For the two-for-one problem where both man $m$ and woman $w$ can misreport, it is known that $m$ cannot be strictly better off unless $w$ is strictly worse off~\citep[Corollary 4]{H07cheating}. %
To prevent the beneficiary $w$ from being worse off, we must ensure that man $m$'s match does not improve, implying once again that $\mu(m) = \mu'(m)$.

\paragraph{Inconspicuous manipulation.} A misreported list (or \emph{strategy}) $\>'_m$ for an accomplice $m$ is said to be \emph{inconspicuous} if it can be derived from his true list $\>_m$ by promoting at most %
one woman and making no other changes. Similarly, when the misreporting agent is a woman $w$, inconspicuousness involves promoting at most %
one man in her true list $\>_w$.

\paragraph{Push up and push down operations}
For any man $m \in M$, let $\>^L_m$ and $\>^R_m$ denote the parts of $m$'s list above and below his \DA{} partner, respectively. That is, $\>_m = (\>^L_m,\mu(m),\>^R_m)$. We say that man $m$ \emph{pushes up} a set $X \subseteq W$ if the new list is $\>^{X\uparrow}_m \coloneqq (\>^L_m \cup X,\mu(m),\>^R_m \setminus X)$. Similarly, \emph{pushing down} a set $Y \subseteq W$ results in $\>^{Y\downarrow}_m \coloneqq (\>^L_m \setminus Y,\mu(m),\>^R_m \cup Y)$. Note that the exact positions at which agents in $X$ (or $Y$) are placed above (or below) $\mu(m)$ is not important, as long as the sets are appropriately pushed above~(or below) $\mu(m)$. \citet{H06cheating} has shown that the \DA{} outcome remains unchanged if each man $m$ arbitrarily permutes the part of his list above and below his \DA{}-partner $\mu(m)$.

\begin{restatable}[\protect\citealp{H06cheating}]{proposition}{PermutingFalsifiedLists}
Let $\>$ be a profile and let $\mu \coloneqq \DA(\>)$. For any man $m \in M$ with true list $\>_m = (\>^L_m,\mu(m),\>^R_m)$, 
let ${ \>'_m \coloneqq (\pi^L(\>^L_m),\mu(m),\pi^R(\>^R_m)) }$, where $\pi^L$ and $\pi^R$ are arbitrary permutations. Let %
$\mu' \coloneqq \DA(\>_{-m},\>'_m)$. Then, $\mu' = \mu$.
\label{prop:Permuting_Falsified_Lists}
\end{restatable}

All omitted proofs are in the Appendix.

\section{Two for One: Helping a Single Woman Through Pair Manipulation} \label{sec:two-for-one}

In this section, we will consider the ``two for one'' generalization of accomplice manipulation where the accomplice and the strategic woman can jointly misreport in order to benefit the latter. %
To see how such a generalization can be useful, let us start with an example showing that \emph{pair} manipulation can be strictly more beneficial for the woman compared to either self or accomplice manipulation~(\Cref{eg:pair_manipulation}).

\begin{example}[\textbf{Pair manipulation can be strictly better than accomplice or self manipulation}] \label{eg:pair_manipulation}
Consider the following preference profile where the \DA{} outcome is underlined. The notation ``${m_1}: w_5 \ w_3 \ w_4 \ w_2$ \ $w_1$'' denotes that $m_1$'s top choice is $w_5$, second choice is $w_3$, and so on.

\begin{table}[H]
    \centering
    \begin{tabularx}{0.7\linewidth}{XXXXXXXXXXXXXXX}
            $\boldsymbol{\textcolor{blue}{{m_1}}}\colon$ & $w_5$ & $\underline{w_3^*}$ & $w_4$ & $w_2$ & $w_1$ && $\boldsymbol{\textcolor{blue}{{w_1}}}\colon$ & $m_3^*$ & $\underline{m_4}$ & $m_5$ & $m_1$ & $m_2$\\
            ${m_2}\colon$ & $\underline{w_4^*}$ & $w_1$ & $w_5$ & $w_3$ & $w_2$ && ${w_2}\colon$ & $m_1$ & $\underline{m_5^*}$ & $m_3$ & $m_2$ & $m_4$\\
            ${m_3}\colon$ & $\underline{w_5}$ & $w_4$ & $w_1^*$ & $w_2$ & $w_3$ && ${w_3}\colon$ & $m_5$ & $m_4$ & $m_3$ & $m_2$ & $\underline{m_1^*}$\\
            ${m_4}\colon$ & $\underline{w_1}$ & $w_4$ & $w_5^*$ & $w_2$ & $w_3$ && ${w_4}\colon$ & $\underline{m_2^*}$ & $m_5$ & $m_3$ & $m_1$ & $m_4$\\
            ${m_5}\colon$ & $\underline{w_2^*}$ & $w_4$ & $w_3$ & $w_1$ & $w_5$ && ${w_5}\colon$ & $m_5$ & $m_2$ & $m_4^*$ & $\underline{m_3}$ & $m_1$
        \end{tabularx}
\end{table}

Suppose the manipulating pair is $(m_1, w_1)$. Since $m_4$ is the only man who proposes to $w_1$ during the execution of \DA{} algorithm on this profile, it follows %
that there is no beneficial self manipulation strategy for $w_1$. 

To find an optimal accomplice manipulation for $m_1$, it suffices to focus on inconspicuous strategies~\citep{HUV21accomplice}. It is straightforward to verify that promoting any woman below $w_3$ in $m_1$'s list does not result in a better match for $w_1$. In fact, none of the other men can give $w_1$ a better partner via no-regret manipulation.

The \DA{} matching when $m_1$ and $w_1$ jointly misreport with %
$\>'_{m_1} \coloneqq w_1 \> w_5 \> w_3 \> w_4 \> w_2$ and $\>'_{w_1} \coloneqq m_3 \> m_5 \> m_1 \> m_4 \> m_2$ is marked by ``$*$''. Note that the strategic woman $w_1$ is now able to match with her top choice without worsening the match of the accomplice $m_1$. \qed
\end{example}

Since pair manipulation can be strictly more beneficial than either self or accomplice manipulation, it is natural to ask whether an \emph{optimal} pair manipulation can be efficiently computed. 
Our main result in this section is that an optimal pair manipulation can be computed in polynomial time.

\subsection{Computing an Optimal Joint Strategy} \label{sec:Computing_Pair}

A natural approach for finding an optimal joint strategy is to combine (or ``concatenate'') an optimal self manipulation for the woman and an optimal accomplice manipulation for the man.
However, as we saw in \Cref{eg:pair_manipulation}, an optimal pair manipulation may exist despite there being no beneficial accomplice nor self manipulations. Further, \Cref{eg:concatenate_pair} in \Cref{subsec:concatenate} shows that the woman's match could actually \emph{worsen} by naively combining the respective individual strategies.
Thus, pair manipulation appears to be ``more than just the sum of its parts''.

Another natural approach is to combine \emph{inconspicuous} (but not necessarily individually optimal) strategies of the accomplice and the strategic woman. However, as we discuss in \Cref{Proof_PairManipulation} 
in the Appendix, there are some subtleties that arise from this approach that become difficult to resolve. We leave the question of determining whether optimal pair manipulation is inconspicuous as an open problem.

Nevertheless, the idea of looking for \emph{structure} in the individual strategies turns out to be useful. To see why, consider a manipulating pair $(m,w)$. Let $\>^*_m$ and $\>^*_w$ denote the respective lists of $m$ and $w$ under an optimal pair manipulation, and let $\>^* \coloneqq (\>_{-\{m,w\}},\>^*_m,\>^*_w)$ denote the corresponding preference profile.

Since it is easier to think about single-agent misreports, let us break down the transition from the true profile $\>$ to the pair manipulation profile $\>^*$ in two steps: First, swap $w$'s list in $\>$ to obtain the intermediate profile $\>^w \coloneqq \{\>_{-w},\>^*_w\}$, and then swap $m$'s list in $\>^w$ to get $\>^*$; see \Cref{fig:PairManipulation}. We will show that this two-step approach allows us to impose additional structure on the individual strategies $\>^*_m$ and $\>^*_w$.

\begin{figure}[t]%
\centering
\begin{tikzpicture}
    \node[] at (0, 0) (l) {$\>$};
    \node[] at (8, 0)   (r) {$\>^*$};
    \node[] at (4,0)   (b) {$\>^w= \{\>_{-w},\>^*_w\}$};
    \draw[-latex,line width=1pt] (l)--(b) node [midway,above=0.5pt,sloped,fill=white] {{$\>_w \rightarrow \, \>^*_w$}};
    \draw[-latex,line width=1pt] (b)--(r) node [midway,above=0.5pt,sloped,fill=white] {{$\>_m \rightarrow \, \>^*_m$}};
\end{tikzpicture}
\caption{Preference profiles under pair manipulation.}
\label{fig:PairManipulation}
\end{figure}

Let us start by analyzing woman $w$'s strategy $\>^*_w$. Consider the transition $\> \rightarrow \, \>^w$ in %
\Cref{fig:PairManipulation}, where $w$ is the only misreporting agent. Let $S_w$  denote the set of preference lists that can be obtained from $w$'s true list $\>_w$ by moving some pair of men to the top two positions, i.e., 
$$S_w \coloneqq \{ (m_i, \, m_j, \, {\>_w \setminus \{m_i, m_j\}}) : m_i, m_j \in M \}.$$ 
In \Cref{prop:WomanListInconspicuous}, we show that for an \emph{arbitrary} misreport by $w$, there exists a list in $S_w$ that creates the same matching for \emph{all} agents. Thus, it follows that $\>^*_w \, \in S_w$.
Observe that the set $S_w$ is of polynomial size $\O(n^2)$ and can be efficiently enumerated.%

\begin{restatable}{lemma}{WomanListInconspicuous}
Let $\>$ be a profile and let $\>'_w$ be any misreport %
for a fixed woman $w$. Then, there exists a list ${\>''_w \, \in S_w}$ that achieves the same matching, i.e., $\mu'' = \mu'$, where $\mu' \coloneqq \DA(\>_{-w},\>'_w)$ and $\mu'' \coloneqq \DA(\>_{-w},\>''_w)$.
\label{prop:WomanListInconspicuous}
\end{restatable}

Next consider the transition $\>^w \rightarrow \, \>^*$ in %
\Cref{fig:PairManipulation}. For this step, man $m$ is the only misreporting agent. 
We define $\widehat{\>}_m$ as the list obtained by promoting $m$'s original match, namely $\mu(m)$, to the top of his original list $\>_m$, and define $S_m \coloneqq \{\widehat{\>}_m\} \cup \{\widehat{\>}_m^{w'\uparrow} : w' \neq \mu(m)\}$ 
as the set consisting of the list $\widehat{\>}_m$ as well as all preference lists that are obtained by individually pushing up each woman other than $\mu(m)$ to the top position in the list $\widehat{\>}_m$. %
Further, we say that an arbitrary misreport $\>'_m$ is \emph{feasible} if $m$ matches with $\mu(m)$ under $\>' \coloneqq \{\>_{-\{m,w\}}, \>'_m, \>^*_w\}$. In \Cref{lem:AccompliceMisreports_PairManipulation}, we show that for an \emph{arbitrary} feasible misreport by the accomplice,
there exists another feasible list in $S_m$ that results in the same partner for $w$. Thus, we can assume that $\>^*_m \, \in S_m$. Again, observe that the set $S_m$ is of polynomial size $\O(n)$.

\begin{restatable}{lemma}{AccompliceMisreportsPairManipulation}
Let $\>$ be a profile and let $\>'_w$ and $\>'_m$ be any misreports for a fixed pair $(m, w)$ such that $\mu'(m) = \mu(m)$, where $\mu \coloneqq \DA(\>)$ and $\mu' \coloneqq (\>_{-\{m,w\}}, \>'_m, \>'_w)$. 
Then, there exists a list $\>''_m \ \in S_m$ such that $\mu''(m) = \mu(m)$ and $\mu''(w) = \mu'(w)$, %
where $\mu'' \coloneqq \DA(\>_{-\{m,w\}}, \>''_m, \>'_w)$.
\label{lem:AccompliceMisreports_PairManipulation}
\end{restatable}

Although the lists in sets $S_m$ and $S_w$ are not necessarily inconspicuous versions of the true lists $\>_m$ and $\>_w$, respectively, we have been able to identify nominally-sized sets of misreports $S_m$ and $S_w$ that are sufficient to check, leading to a simple algorithm for finding an optimal pair manipulation strategy: Enumerate the sets $S_m$ and $S_w$, evaluate the \DA{} outcome for each possible $\>'_m \in S_m, \>'_w \in S_w$ pair, and return the strategy that gives the best match for the woman $w$ without causing regret for the accomplice $m$; see Algorithm~\ref{alg:PairManipulation}.%

\begin{algorithm}[t]
\small
\caption{Computing an optimal pair manipulation}
\label{alg:PairManipulation}
\begin{algorithmic}[1]
    \Require{Profile $\succ$, accomplice $m$, beneficiary $w$}
    \Ensure{Optimal pair manipulations $\>^*_m$ and $\>^*_w$}
    
    \State Initialize $(\mu^*, \>^*_m, \>^*_w) \gets (\mu \coloneqq \DA(\>), \>_m, \>_w)$
    \State Compute $S_w$ %
    \For{each $\>'_w \in S_w$}
        \State Compute $S_m$ %
        \For{each $\>'_m \in S_m$} 
            \State{$\mu' \gets \DA(\>_{-\{m,w\}}, \>'_m, \>'_w)$}
            \If{$\mu'(w) \>_{w} \mu^*(w)$ and $\mu'(m) = \mu(m)$}
                \State{Update $(\mu^*, \>^*_m, \>^*_w) \gets (\mu', \>'_m, \>'_w)$}
            \EndIf
        \EndFor
    \EndFor
    \State \Return{$\>^*_m$ and $\>^*_w$}
\end{algorithmic}
\end{algorithm}

\begin{restatable}{theorem}{PairManipulation}
An optimal pair manipulation can be computed in $\O(n^5)$ time.
\label{thm:PairManipulation}
\end{restatable}

\begin{remark}
In \Cref{eg:pair_manipulation_unstable} in the Appendix, we show that optimal pair manipulation could fail to be \emph{stability-preserving}. Thus, an \emph{unrestricted} pair manipulation (i.e., when the manipulated matching is not required to be stable with respect to true preferences) can be strictly better than an optimal \emph{stability-preserving} pair manipulation. 
\label{rem:Unstable_Pair}
\end{remark}

\section{One for All: Helping All Women Through a Single Accomplice} \label{sec:one-for-all}

Let us now consider a different generalization of accomplice manipulation which we call ``one for all'' manipulation where a single accomplice (man $m$) misreports in order to improve the outcome for \emph{all} women in $W$. Recall that due to the no-regret assumption, the manipulated match of the accomplice $m$ is the same as his true match. As we are interested in improving a \emph{group} of agents, it will be helpful to define the notions of \emph{Pareto improvement} and \emph{Pareto optimal strategies}.

\paragraph{Pareto optimal and optimal strategies.}
Let $\>$ be the true preference profile and let $\mu \coloneqq \DA(\>)$. We say that a strategy $\>'_m$ of the accomplice \emph{Pareto improves} another strategy $\>''_m$ if $\mu' \succeq_W \mu''$ and $\mu'(w) \>_w \mu''(w)$ for some $w \in W$, where $\mu' \coloneqq \DA(\>_{-m}, \>'_m)$ and $\mu'' \coloneqq \DA(\>_{-m}, \>''_m)$. A strategy $\>'_m$ is \emph{Pareto optimal} if $\mu' \succeq_W \mu$ and there is no other strategy $\>''_m$ that Pareto improves $\>'_m$. Further, a strategy $\>'_m$ is \emph{optimal} if $\mu' \succeq_W \mu$ and for any other strategy $\>''_m$, we have $\mu' \succeq_W \mu''$. Thus, given a Pareto optimal strategy, any other strategy that improves some woman must make some other woman worse off, while the outcome under an optimal strategy simply cannot be improved for \emph{any} woman. 

Similarly, we say that a matching $\mu' \coloneqq \DA(\>_{-m},\>'_m)$ is ``Pareto optimal'' (respectively, ``optimal'') if the corresponding strategy $\>'_m$ is Pareto optimal (respectively, optimal).

Note that an optimal strategy is also Pareto optimal. The finiteness of the strategy space implies that %
a Pareto optimal strategy is guaranteed to exist. Whether an optimal strategy also always exists is not immediately clear; however, if an optimal strategy exists, then the set of Pareto optimal matchings---the Pareto frontier---must be a singleton, consisting only of the optimal matching. There can be multiple optimal strategies, but all such strategies must induce the same optimal matching.%

Let us now proceed to analyzing the structure of (Pareto) optimal strategies. When an accomplice manipulates on behalf of a single beneficiary woman (i.e., ``one for one''), it is known that there always exists an optimal strategy that is inconspicuous~~\citep{HUV21accomplice}. By contrast, when an accomplice misreports on behalf of multiple women (i.e., ``one for all''), an inconspicuous strategy may no longer be optimal~(\Cref{eg:inconspicuous_is_suboptimal}).

\begin{example}[\textbf{Inconspicuous strategy can be suboptimal}] 
Consider the following preference profile where the \DA{} outcome is underlined. %

\begin{table}[H]
    \centering
    \begin{tabularx}{0.7\linewidth}{XXXXXXXXXXXXXXX}
            $\boldsymbol{\textcolor{blue}{{m_1}}}\colon$ & $\underline{w_1^*}$ & $w_2$ & $w_3$ & $w_4$ & $w_5$ && ${w_1}\colon$ & $\underline{m_1^*}$ & $m_2$ & $m_3$ & $m_4$ & $m_5$\\
            ${m_2}\colon$ & $\underline{w_2}$ & $w_3^*$ & $w_4$ & $w_5$ & $w_1$ && ${w_2}\colon$ & $m_3^*$ & $m_4$ & $m_5$ & $m_1$ & $\underline{m_2}$\\
            ${m_3}\colon$ & $\underline{w_3}$ & $w_2^*$ & $w_1$ & $w_5$ & $w_4$ && ${w_3}\colon$ & $m_2^*$ & $m_5$ & $m_4$ & $m_1$ & $\underline{m_3}$\\
            ${m_4}\colon$ & $\underline{w_4}$ & $w_5^*$ & $w_1$ & $w_2$ & $w_3$ && ${w_4}\colon$ & $m_5^*$ & $m_3$ & $m_2$ & $m_1$ & $\underline{m_4}$\\
            ${m_5}\colon$ & $\underline{w_5}$ & $w_4^*$ & $w_3$ & $w_2$ & $w_1$ && ${w_5}\colon$ & $m_4^*$ & $m_2$ & $m_3$ & $\underline{m_5}$ & $m_1$
        \end{tabularx}
\end{table}

Suppose the accomplice is $m_1$ and all women are beneficiaries. The \DA{} matching after $m_1$ submits the optimal no-regret manipulated list $\>'_{m_1} \coloneqq w_2 \> w_4 \> w_1 \> w_3 \> w_5$ is marked by ``$*$''. Notice that $\succ'_{m_1}$ is derived from $\succ_{m_1}$ by pushing up $w_2$ and $w_4$ and therefore is not inconspicuous. The manipulation results in the women-optimal matching, where  %
all women are matched with their top choices.

There is no inconspicuous strategy that $m_1$ (or any other man for that matter) can report to produce the same matching; indeed, if $m_1$ were to push up only $w_2$, then only $w_2$ and $w_3$ would improve, and if he were to push up only $w_4$, then only $w_4$ and $w_5$ would improve. This observation highlights the conflict between optimality and inconspicuousness when the set of beneficiaries consists of \emph{all} women.\qedhere
\label{eg:inconspicuous_is_suboptimal}
\end{example}

Our main result in this section is that an \emph{optimal} strategy for the accomplice is guaranteed to exist and can be computed in polynomial time~(\Cref{thm:OptimalStrategy_MultipleWomen}). In fact, we will show a stronger result: Among all optimal strategies, we can efficiently compute one that promotes the \emph{smallest} number of women in the accomplice's list~(\Cref{thm:No-regret_Minimal_Is_Minimum} in \Cref{sec:Smallest-size_One-for-all}).

\subsection{Computing an Optimal Strategy}

Recall from \Cref{prop:Permuting_Falsified_Lists} that even if each man $m$ arbitrarily permutes the part of his list above and below his \DA{}-partner $\mu(m)$, the \DA{} outcome remains unchanged. 
This result shows that any strategy of the accomplice $m$, without loss of generality, can be expressed in terms of only \emph{push up} and \emph{push down} operations, where a set of women is pushed above the \DA{} partner $\mu(m)$, and another disjoint set of women is pushed below $\mu(m)$.

We will now provide a structural simplification: Any matching obtained by a combination of push up and push down operations in the accomplice's list can be weakly improved for \emph{all} women by the push up operation alone.

\begin{restatable}%
{proposition}{CombiningPushUpPushDownAllWomen}
Let $\>$ be a %
profile. For any fixed man $m$ and any subsets $X \subseteq W$ and $Y \subseteq W$ of women who are ranked below and above $\mu(m)$, respectively, let $\>' \coloneqq \{ \succ_{-m}, \succ_m^{X\uparrow} \}$ denote the profile after pushing up the set $X$ and let $\>'' \coloneqq \{ \succ_{-m}, \succ_m^{X\uparrow, Y\downarrow} \}$ denote the profile after pushing up $X$ and pushing down $Y$ in the true preference list $\>_m$ of man $m$. Then, $\mu' \succeq_W \mu''$, where $\mu' \coloneqq \DA(\>')$ and $\mu'' \coloneqq \DA(\>'')$.
\label{lem:CombiningPushUpPushDown_AllWomen}
\end{restatable}

Having established that push up operations alone do not suffice, %
let us now examine \emph{which} subset of women the accomplice should push up. Given a profile $\>$ and an accomplice $m$, define the \emph{no-regret set} $W^\textsc{NR} \coloneqq \{w \in W \, : \, \>' \coloneqq \{ \>_{-m}, \>_m^{w \uparrow} \} \text{ is a no-regret profile}\}$ as the set of all women who do not cause $m$ to incur regret when pushed up individually, and its complement \emph{with-regret set} $W^\textsc{R} \coloneqq W \setminus W^\textsc{NR}$.%

We will first show that pushing up any subset of no-regret women does not cause regret for the accomplice~(\Cref{lem:NoRegretSubset}).

\begin{restatable}{lemma}{NoRegretSubset}
Let $\>$ be a profile and let $\mu \coloneqq \DA(\>)$. For any subset $Y \subseteq W^\textsc{NR}$, let $\>^Y \coloneqq \{ \succ_{-m}, \succ_m^{Y\uparrow} \}$ denote the preference profile after pushing up the set $Y$ in the true preference list $\>_m$ of man $m$, and let $\mu^Y \coloneqq \DA(\>^Y)$. Then, $m$ does not incur regret under $\>^Y$, i.e., $\mu^Y(m) = \mu(m)$.%
\label{lem:NoRegretSubset}
\end{restatable}

In contrast to \Cref{lem:NoRegretSubset}, any subset $Y \subseteq W$ that contains at least one woman from the with-regret set (i.e., $Y \cap W^\textsc{R} \neq \emptyset$) causes regret for the man $m$~(\Cref{lem:Combining_WithRegret_NoRegret_PushUp}).

\begin{restatable}{lemma}{CombiningWithRegretNoRegretPushUp}
Let $w' \in W^\textsc{R}$ and let $Y \subseteq W$ be such that $w' \in Y$. Then, $m$ incurs regret under $\>^Y \coloneqq \{ \succ_{-m}, \succ_m^{Y\uparrow} \}$, i.e., $\mu(m) \>_m \mu^Y(m)$, where $\mu^Y \coloneqq \DA(\>^Y)$.
\label{lem:Combining_WithRegret_NoRegret_PushUp}
\end{restatable}

Together, \Cref{lem:NoRegretSubset,lem:Combining_WithRegret_NoRegret_PushUp} imply that a push up operation is no-regret if and only if the pushed-up set is a subset of $W^\textsc{NR}$. Thus, an optimal (or Pareto optimal) strategy should promote some subset of $W^\textsc{NR}$. This observation, however, does not automatically provide an efficient algorithm for computing the desired strategy because brute force enumeration of subsets of $W^\textsc{NR}$ could take exponential time. Also, in case an optimal strategy does not exist, %
the Pareto frontier of strategies can be %
exponential in size, %
again ruling out exhaustive search.

Our main result of this section (\Cref{thm:OptimalStrategy_MultipleWomen}) alleviates both of the above concerns. We show that not only does an optimal strategy always exist, but also that pushing up the entire no-regret set $W^\textsc{NR}$ achieves such an outcome. %

\begin{restatable}{theorem}{OptimalStrategyMultipleWomen}
An optimal one-for-all strategy for the accomplice is to push up the no-regret set $W^\textsc{NR}$ in his true list.
\label{thm:OptimalStrategy_MultipleWomen}
\end{restatable}

Our proof of \Cref{thm:OptimalStrategy_MultipleWomen} leverages the following known result which says that the matching resulting from a no-regret push up operation is weakly preferred by all women.

\begin{restatable}[\protect\citealp{HUV21accomplice}]{proposition}{PushUp}
Let $\>$ be a preference profile and let ${\mu \coloneqq \DA(\>)}$. For any man $m$, let $\>' \coloneqq \{ \succ_{-m}, \succ_m^{X\uparrow} \}$ and $\mu' \coloneqq \DA(\>')$. If $m$ does not incur regret, then $\mu' \in S_{\>}$ and thus $\mu' \succeq_W \mu$ and $\mu \succeq_M \mu'$.
\label{prop:PushUp}
\end{restatable}

\begin{proof} (of \Cref{thm:OptimalStrategy_MultipleWomen})
From \Cref{prop:Permuting_Falsified_Lists}, we know that any accomplice manipulation can be simulated via push up and push down operations. \Cref{lem:CombiningPushUpPushDown_AllWomen} shows that any combination of push up and push down operations can be weakly improved for all women by push up only. 
From \Cref{lem:Combining_WithRegret_NoRegret_PushUp}, we know that the desired push up set, say $Y \subseteq W$, should not contain any woman from the with-regret set $W^R$. Therefore, $Y \subseteq W^\textsc{NR}$. From \Cref{lem:NoRegretSubset}, we know that pushing up $Y$ satisfies no-regret assumption. If $Y \neq W^\textsc{NR}$ (thus, $Y \subset W^\textsc{NR}$), then \Cref{prop:PushUp} shows that $\>^Y \coloneqq \{ \succ_{-m}, \succ_m^{Y\uparrow} \}$ can be weakly improved for all women by additionally pushing up the women in $W^\textsc{NR} \setminus Y$. Thus, pushing up all women in $W^\textsc{NR}$ gives an optimal no-regret accomplice manipulation strategy for helping all women, as desired.
\end{proof}

\Cref{thm:OptimalStrategy_MultipleWomen} readily gives a polynomial-time algorithm for computing an optimal strategy~(\Cref{cor:PolyAlgorithm_OptimalStrategy}).

\begin{restatable}{corollary}{PolyAlgorithmOptimalStrategy}
An optimal one-for-all strategy for the accomplice can be computed in $\O(n^3)$ time.
\label{cor:PolyAlgorithm_OptimalStrategy}
\end{restatable}

Although pushing up the entire no-regret set $W^\textsc{NR}$ is optimal~(\Cref{thm:OptimalStrategy_MultipleWomen}), the accomplice may want to displace as few women as possible in order to remain close to his true preference list. In \Cref{sec:Smallest-size_One-for-all}, we provide a polynomial-time algorithm for computing a \emph{minimum} optimal strategy~(i.e., one that promotes the smallest number of women). We also show that the size of the promoted set is at most $\lfloor \frac{n-1}{2} \rfloor$ and that this bound is tight (\Cref{sec:Tight_Bounds_Appendix}).

\section{Experimental Results} \label{sec:experiments}

Let us now experimentally compare the two-sided and one-sided models in terms of the \emph{fraction of instances} where each model improves upon truthful reporting. To this end, we will consider an experimental setup where the preferences of $n$ men and $n$ women are drawn uniformly at random.\footnote{The assumption about uniformly random preferences is quite common in the literature on strategic aspects of stable matchings; see, for example,~\citep{TS01gale,KPR13matching,IM15incentives,A15susceptibility,AKL17unbalanced}.} For each value of $n \in \{4,6,\dots,20\}$, we independently sample $1\,000$ preference profiles.

For the \emph{two-for-one} part, we compute the fraction of instances where some man $m$ can jointly misreport with a fixed woman $w$ to improve her match, and compare it with the analogous fraction where only one of $m$~(accomplice) or $w$~(self) can misreport; see \Cref{fig:two_sided_frequencies}~(left).
Similarly, for the \emph{one-for-all} part, we compute the fraction of instances where some man can misreport to help all women~(i.e., weakly improve all and strictly improve some compared to their true matches), and compare it with the analogous fraction where a woman helps all women; see \Cref{fig:two_sided_frequencies}~(right).

\begin{figure}[t]
    \centering
    \begin{minipage}{0.49\linewidth}
        \centering
        \includegraphics[width=\textwidth]{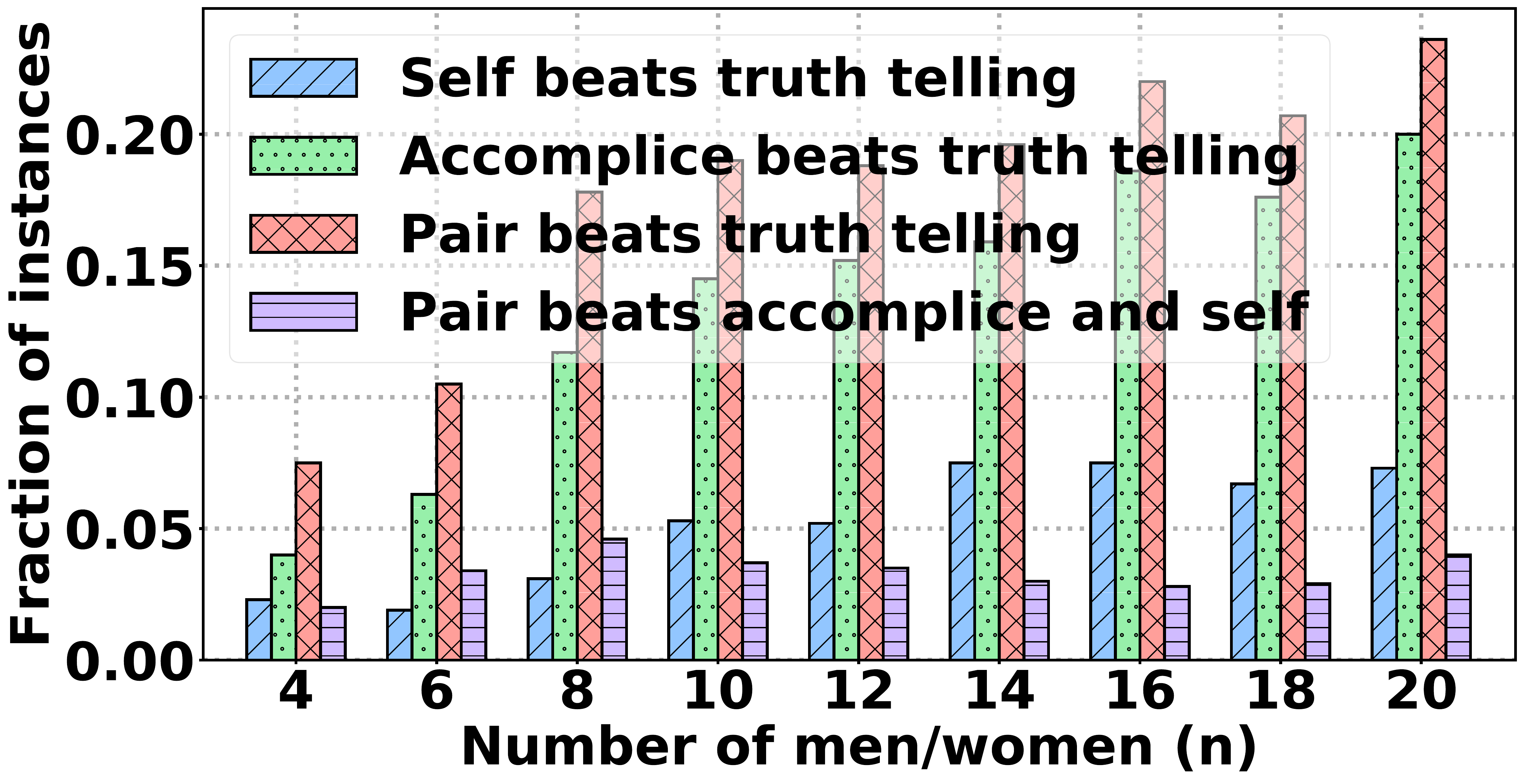}
    \end{minipage}
    \hfill
    \begin{minipage}{0.49\linewidth}
        \centering
        \includegraphics[width=\textwidth]{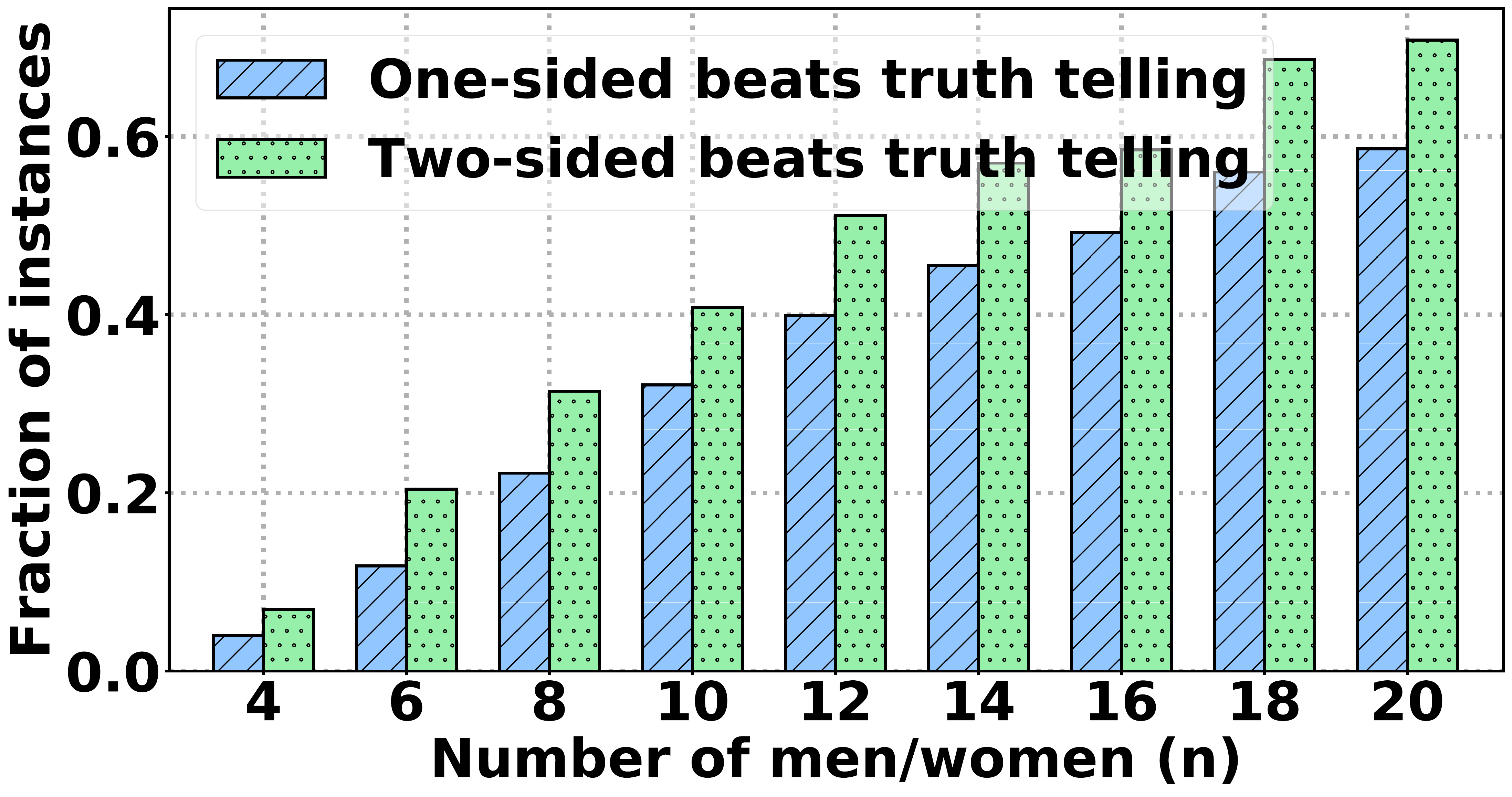}
    \end{minipage}
    \caption{Comparing one-sided and two-sided strategies for helping a \emph{single} woman (left) %
    and %
    for helping \emph{all} women (right) in terms of the fraction of instances where a (Pareto) improvement is possible over truthful reporting.
    }
     \label{fig:two_sided_frequencies}
\end{figure}

\Cref{fig:two_sided_frequencies} shows that two-sided strategies are \emph{more frequently available} than one-sided; specifically, in roughly $2\%$ more instances %
under the two-for-one setting and roughly $10\%$ more instances under the one-for-all. In \Cref{sec:Additional_Experiments}, we show that two-sided manipulation outperforms one-sided in terms of the \emph{extent of improvement} for the beneficiary/beneficiaries, i.e., the difference between the ranks of old and new matched partner(s).

\section{Concluding Remarks}

We studied two coalitional generalizations of two-sided manipulation of the \DA{} algorithm. Moving from single-agent to coalitional manipulation impacted the structure of optimal strategies in the form of loss of inconspicuousness, but we showed that efficient computation can still be achieved.

Going forward, it will be interesting to consider manipulation by \emph{arbitrary coalitions} of men and women. Another relevant direction could be to interpret two-sided manipulation as a \emph{bribery}~ problem~\citep{FHH09hard,BBH+21bribery} wherein there is a cost associated with each pairwise swap in an agent's true list. 
Finally, extensions of our work to more general preference models (e.g., partial orders), %
as well as experimental evaluation on \emph{non-uniform} distributions or real-world preference data, will also be of interest.

\section*{Acknowledgments}
Hadi Hosseini acknowledges support from NSF IIS grants \#2052488 and \#2107173.

\bibliographystyle{named}%
\bibliography{References}

\clearpage
\appendix
\begin{center}
    \textbf{\huge{Appendix}}
\end{center}

\section{Useful Results from Prior Work}

For ease of reference, we recall the notation of push up and push down operations below.

\paragraph{Push up and push down operations}
For any man $m \in M$, let $\>^L_m$ and $\>^R_m$ denote the parts of $m$'s list above and below his \DA{} partner, respectively. That is, $\>_m = (\>^L_m,\mu(m),\>^R_m)$. We say that man $m$ \emph{pushes up} a set $X \subseteq W$ if the new list is $\>^{X\uparrow}_m \coloneqq (\>^L_m \cup X,\mu(m),\>^R_m \setminus X)$. Similarly, \emph{pushing down} a set $Y \subseteq W$ results in $\>^{Y\downarrow}_m \coloneqq (\>^L_m \setminus Y,\mu(m),\>^R_m \cup Y)$. Note that the exact positions at which agents in $X$ (or $Y$) are placed above (or below) $\mu(m)$ is not important, as long as the sets are appropriately pushed above~(or below) $\mu(m)$. \citet{H06cheating} has shown that the \DA{} outcome remains unchanged if each man $m$ arbitrarily permutes the part of his list above and below his \DA{}-partner $\mu(m)$ (\Cref{prop:Permuting_Falsified_Lists}).\newline

We now present several useful known results.

\begin{restatable}[\protect\citealp{DF81machiavelli,H06cheating}]{proposition}{PushDown}
Let $\>$ be a preference profile and let $\mu \coloneqq \DA(\>)$. For any fixed man $m$ and any subset $X \subseteq W$ of women who are ranked above $\mu(m)$ in $\>_m$, let ${\>' \coloneqq \{\>_{-m},\>_m^{X \downarrow}\} }$ denote the preference profile after pushing down the women in $X$ in the true preference list $\>_m$ of man $m$, and let $\mu' \coloneqq \DA(\>')$. Then, $\mu' \succeq_M \mu$ and $\mu'(m) = \mu(m)$.
\label{prop:PushDown}
\end{restatable}

\begin{proposition}[\protect\citealp{GF13sisterhood}] \label{prop:women-optimality-of-woman-manipulations}
Let $W^* \subseteq W$ be a set of manipulating women. If no woman in $W^*$ is strictly worse off after the misreport, then all women (including those not in $W^*$) are weakly better off and all men are weakly worse off after the manipulation.
\end{proposition}

\begin{restatable}[\protect\citealp{HUV21accomplice}]{proposition}{PushDownWorseForWomen}
Let $\>$ be a preference profile and let $\mu \coloneqq \DA(\>)$. For any fixed man $m$ and any subset $X \subseteq W$ of women who are ranked above $\mu(m)$ in $\>_m$, let ${\>' \coloneqq \{\>_{-m},\>_m^{X \downarrow}\} }$ denote the preference profile after pushing down the women in $X$ in the true preference list $\>_m$ of man $m$, and let $\mu' \coloneqq \DA(\>')$. Then, $\mu \succeq_W \mu'$.
\label{lem:PushDown_Worse_For_Women}
\end{restatable}

\begin{restatable}[\protect\citealp{HUV21accomplice}]{proposition}{AccompliceMatchAfterWithRegretManipulation}
Let $\>$ be a preference profile and let ${\mu \coloneqq \DA(\>)}$. For any fixed man $m$ and any subset $X \subseteq W$ of women who are ranked below $\mu(m)$ in $\>_m$, let $\>^X \coloneqq \{ \>_{-m}, \>_m^{X \uparrow} \}$ denote the preference profile after pushing up the women in $X$ in the true preference list $\>_m$ of man $m$, and let ${\mu^X \coloneqq \DA(\>^X)}$. If $m$ incurs regret (i.e., if $\mu(m) \>_{m} \mu^X(m)$), then $\mu^X(m) \in X$.
\label{lem:accomplice-match-after-regret-manipulation}
\end{restatable}

\begin{restatable}[\protect\citealp{HUV21accomplice}]{proposition}{SingleAgentPushUpNoRegret}
Let $\succ$ be a preference profile and let $\mu \coloneqq \DA(\>)$. For any fixed man $m$ and any subset $X \subseteq W$ of women who are ranked below $\mu(m)$ in $\>_m$, let $\>' \coloneqq \{ \succ_{-m}, \succ_m^{X\uparrow} \}$ denote the preference profile after pushing up the women in $X$ in the true preference list $\>_m$ of man $m$, and let $\mu' \coloneqq \DA(\>')$. If $m$ does not incur regret (i.e., if $\mu'(m) = \mu(m)$), then for every $w_x \in X$, $m$ does not incur regret under the matching $\mu^x \coloneqq \DA(\>^x)$, where $\>^x \coloneqq \{\>_{-m},\>_m^{w_x\uparrow}\}$.
\label{lem:SingleAgentPushUpNoRegret}
\end{restatable}

\begin{restatable}[\protect\citealp{HUV21accomplice}]{proposition}{WeakPushUp}
Let $\succ$ be a preference profile and let $\mu \coloneqq \DA(\>)$. For any fixed man $m$ and any subset $X \subseteq W$ of women who are ranked below $\mu(m)$ in $\>_m$, let $\>' \coloneqq \{ \succ_{-m}, \succ_m^{X\uparrow} \}$ denote the preference profile after pushing up the women in $X$ in the true preference list $\>_m$ of man $m$, and let $\mu' \coloneqq \DA(\>')$. If $m$ does not incur regret (i.e., if $\mu'(m) = \mu(m)$) and $\mu(w') \>_{w'} m$ for all $w' \in X$, where $\mu' \coloneqq \DA(\>')$, then $\mu' = \mu$.
\label{prop:WeakPushUp}
\end{restatable}

\begin{restatable}[\protect\citealp{HUV21accomplice}]{proposition}{StrictPushUp}
Let $\succ$ be a preference profile and let $\mu \coloneqq \DA(\succ)$. For any man $m$ and any subset $X \subseteq W$ of women who are ranked below $\mu(m)$ in $\>_m$, let $\>' \coloneqq \{ \succ_{-m}, \succ_m^{X\uparrow} \}$ denote the preference profile after pushing up the women in $X$ in the true preference list $\>_m$ of man $m$, and let $\mu' \coloneqq \DA(\>')$. If $m$ does not incur regret (i.e., if $\mu'(m) = \mu(m)$) and $\mu' \neq \mu$, then there exist at least two distinct women $w',w'' \in W$ such that $\mu'(w') \>_{w'} \mu(w')$ and $\mu'(w'') \>_{w''} \mu(w'')$, and at least two distinct men $m',m'' \in M$ such that $\mu(m') \>_{m'} \mu'(m')$ and $\mu(m'') \>_{m''} \mu'(m'')$.
\label{prop:StrictPushUp}
\end{restatable}

\begin{restatable}[\protect\citealp{HUV21accomplice}]{proposition}{PushUpProposals}
Let $\>$ be a preference profile and let ${ \mu \coloneqq \DA(\>) }$. For any fixed man $m$ and any subset $X \subseteq W$ of women who are ranked below $\mu(m)$ in $\>_m$, let $\>' \coloneqq \{ \succ_{-m}, \succ_m^{X\uparrow} \}$ denote the preference profile after pushing up the women in $X$ in the true preference list $\>_m$ of man $m$, and let $\mu' \coloneqq \DA(\>')$. If $m$ does not incur regret (i.e., if $\mu'(m) = \mu(m)$), then the set of proposals that occur under $\>$ is contained within the set of proposals that occur under $\>'$.
\label{prop:PushUpProposals}
\end{restatable}

\begin{restatable}[\protect\citealp{HUV21accomplice}]{proposition}{PushUpProposalsContainment}
Let $X \coloneqq \{w_a,w_b,w_c,\dots\}$ be an arbitrary finite set of women that the accomplice $m$ can push up without incurring regret. Then, any proposal that occurs under $\>^X$ also occurs under at least one of the profiles $\>^{a}, \>^{b}, \>^{c}, \dots$, where $\>^x \coloneqq \{\>_{-m}, \>^{w_x \uparrow}\}$ denotes the preference profile after pushing up $w_x$ in the true preference list $\>_m$ of man $m$ for every $w_x \in X$.
\label{prop:PushUpPropsalsContainment}
\end{restatable}

\begin{restatable}[\protect\citealp{HUV21accomplice}]{proposition}{PushUpOneWoman}
For some fixed man $m$, let ${X \subseteq W}$ be a set of women that $m$ can push up (with or without incurring regret). Then, the match for some fixed woman $w$ that is obtained by pushing up all women in $X$ can also be obtained by pushing up exactly one woman in $X$.
\label{prop:PushUpOneWoman}
\end{restatable}

\section{Omitted Material from Section~\ref{sec:two-for-one}}

\subsection{Concatenating Optimal Self and Optimal Accomplice Strategies Can Be Harmful} \label{subsec:concatenate}

As discussed in \Cref{sec:Computing_Pair}, a natural approach for finding an optimal joint strategy is to combine (or ``concatenate'') an optimal self manipulation for the woman and an optimal accomplice manipulation for the man. Unfortunately, as we show in \Cref{eg:concatenate_pair} below, the woman's match may \emph{worsen} by naively combining the respective individual strategies.\footnote{Our example's message of ``coordination is key'' is reminiscent of a similar observation by \citet[Sec 3.1]{SDT21coalitional} in the context of one-sided manipulation by a coalition of women. They showed that naively combining optimal self-manipulation strategies for the individuals may actually be detrimental for the coalition.} 

\begin{example}[\textbf{Optimal self + Optimal accomplice $\neq$ Optimal pair manipulation}]
\label{eg:concatenate_pair}
Consider the following preference profile where the \DA{} outcome is underlined.

\begin{table}[H]
    \centering
    \begin{tabularx}{0.7\linewidth}{XXXXXXXXXXXXXXX}
            $\boldsymbol{\textcolor{blue}{{m_1}}}\colon$ & $\underline{w_4^a}$ & $w_3^s$ & $w_1$ & $w_2^*$ & $w_5$ && $\boldsymbol{\textcolor{blue}{{w_1}}}\colon$ & $m_3^s$ & $m_1$ & $m_5^a$ & $\underline{m_4}$ & $m_2^*$\\
            ${m_2}\colon$ & $w_1^*$ & $w_3$ & $\underline{w_2^s}$ & $w_5^a$ & $w_4$ && ${w_2}\colon$ & $m_4^a$ & $m_3$ & $m_1^*$ & $\underline{m_2^s}$ & $m_5$\\
            ${m_3}\colon$ & $\underline{w_3^a}^{,*}$ & $w_1^s$ & $w_5$ & $w_2$ & $w_4$ && ${w_3}\colon$ & $m_1^s$ & $\underline{m_3^a}^{,*}$ & $m_2$ & $m_4$ & $m_5$\\
            ${m_4}\colon$ & $\underline{w_1}$ & $w_5^{s,*}$ & $w_2^a$ & $w_3$ & $w_4$ && ${w_4}\colon$ & $m_2$ & $m_5^{s,*}$ & $\underline{m_1^a}$ & $m_3$ & $m_4$\\
            ${m_5}\colon$ & $w_3$ & $\underline{w_5}$ & $w_1^a$ & $w_4^{s,*}$ & $w_2$ && ${w_5}\colon$ & $m_2^a$ & $m_4^{s,*}$ & $\underline{m_5}$ & $m_1$ & $m_3$
        \end{tabularx}
\end{table}

Suppose the manipulating pair is $(m_1, w_1)$. The \DA{} matching after $m_1$ submits the optimal no-regret accomplice manipulated list $\>'_{m_1} \coloneqq w_2 \> w_4 \> w_3 \> w_1 \> w_5$ is marked by ``$a$'', and the one where $w_1$ submits the optimal self manipulated list $\>'_{w_1} \coloneqq m_3 \> m_2 \> m_1 \> m_5 \> m_4$ is marked by ``$s$''. Additionally, the \DA{} matching after both submit $\>'_{m_1}$ and $\>'_{w_1}$ simultaneously is marked by ``$*$''. 

Notice that the accomplice $m_1$ and the strategic woman $w_1$ are both strictly worse off under joint manipulation, and truthful reporting by the pair $(m_1,w_1)$ is strictly better for the woman compared to the concatenation strategy.\qed
\end{example}

\subsection{Proof of Lemma~\ref{prop:WomanListInconspicuous}}%

We will start by presenting a result that shows that any misreport by a woman that leaves her matched with the same man as before has no effect on the final matching~(\Cref{lem:Woman_Match_Changes}).

\begin{restatable}{lemma}{WomanMatchChanges}
Let $\>$ be a preference profile and let $\mu \coloneqq \DA(\>)$. For any woman $w$, let $\>' \coloneqq \{\>_{-w}, \>'_{w}\}$ denote the preference profile after $w$ misreports her preferences and let $\mu' \coloneqq \DA(\>')$. If $\mu'(w) = \mu(w)$, then $\mu' = \mu$.
\label{lem:Woman_Match_Changes}
\end{restatable}

\begin{proof}
Suppose, for contradiction, that $\mu'(w) = \mu(w)$ and $\mu'(w') \neq \mu(w')$ for some $w' \in W \setminus \{w\}$. Since $w$'s partner does not change after misreporting, it follows from \Cref{prop:women-optimality-of-woman-manipulations} that $\mu'(w') \succeq_{w'} \mu(w')$. Further, since $\mu'(w') \neq \mu(w')$ and the preferences are strict, we have that $\mu'(w') \>_{w'} \mu(w')$ for some $w'\in W\setminus \{w\}$. Since $w'$ is truthful, she must receive a proposal during the execution of \DA{} algorithm on $\>'$ that she does not under $\>$; thus $P_{\>' \setminus \>} \neq \emptyset$.
Without loss of generality, suppose $(m', w')$ is the \emph{first} proposal in $P_{\>' \setminus \>}$ to occur during the \DA{} execution on $\>'$. 

Since men propose in decreasing order of their preference and it is assumed that $(m', w') \notin P_{\>}$ and $\>'_{m'} = \ \>_{m'}$, it must be that $m'$ was rejected by $w'' \coloneqq \mu(m')$ under $\>'$ before proposing to $w'$. Then, $w''$ must have received a proposal from some man, say $m''$, such that $m'' \>'_{w''} m'$. Since $m'$ is matched with $w''$ under $\mu$ but not under $\mu'$, we infer that $w''$ must be truthful; indeed, the strategic woman retains her $\mu$-partner under $\mu'$. Thus, in particular, $m'' \>'_{w''} m' \implies m'' \>_{w''} m'$.

We had assumed $(m',w')$ to be the first proposal during the $\DA$ execution on $\>'$ to not belong to $P_{\>}$. Since $(m'',w'')$ occurs before $(m',w')$ under $\>'$, we must have that $(m'',w'')$ also occurs under $\>$. Thus, under the execution of \DA{} algorithm on $\>$, $w''$ receives a better proposal from $m''$ but ends up getting matched with $m'$. This contradicts the correctness of \DA{} algorithm. Hence, we must have that $\mu'=\mu$.
\end{proof}

\begin{remark}[\textbf{Strategyproofness and non-bossiness}]
In addition to its role in the proof of \Cref{prop:WomanListInconspicuous}, we note that \Cref{lem:Woman_Match_Changes} offers an important insight: The \DA{} algorithm is \emph{non-bossy} for women. A mechanism is said to be non-bossy if no agent can misreport in a way that changes the allocation of other agents without changing his or her own match~\citep{SS81strategy,T16non}. In two-sided matching, non-bossiness---similar to strategyproofness---can be defined for each side of the market separately.

From \Cref{lem:Woman_Match_Changes}, we know that a woman cannot influence the outcome of any other agent without changing her own match. By contrast, a man can influence the outcome of another agent (a woman) without changing his own match, as evidenced by no-regret accomplice manipulation model~\citep{BH19partners,HUV21accomplice}. 
Combining these observations with the known results for strategyproofness of the \DA{} algorithm~\citep{DF81machiavelli,R82economics}, we obtain that: %
\begin{itemize}
    \item For men, the \DA{} algorithm is strategyproof but fails to be non-bossy, while
    \item For women, the \DA{} algorithm is not strategyproof but satisfies non-bossiness.\qed
\end{itemize}
\end{remark}

We now present \Cref{lem:SwappingLemma}---which is due to \citet{VG17manipulating}---that outlines the effect that swapping any pair of adjacent men in a woman's preference list has on her final partner. It will be helpful to define the notation $\Prop(w,\>)$ that denotes the set of all men who propose to $w$ during the run of \DA{} algorithm on $\>$. Further, $\Prop(w,\>,i)$ denotes the $i^{\text{th}}$ favorite man of $w$ (according to $w$'s list $\>$, namely $\>_w$) in the set $\Prop(w,\>)$. Thus, ${\Prop(w,\>,1) = \mu(w)}$ where ${\mu \coloneqq \DA(\>)}$. The set of non-proposing men is denoted by ${\NonProp(w,\>)}$, i.e., $\NonProp(w,\>) \coloneqq M \setminus \Prop(w,\>)$.

\begin{restatable}[\protect\citealp{VG17manipulating}]{proposition}{SwappingLemma}
 \label{lem:SwappingLemma}
 Let $\>$ and $\>'$ be two preference profiles differing only in the preferences of a fixed woman $w$, and let $\mu = \DA(\>)$ and $\mu' = \DA(\>')$. Let $\>'_w$ be derived from $\>_w$ by swapping the positions of an adjacent pair of men $(m_i,m_j)$ and making no other changes. Then,
 \begin{enumerate}%
 	\item if $m_i \in \NonProp(w,\>)$ or $m_j \in \NonProp(w,\>)$, then $\mu'(w) = \mu(w)$,
 	\item if $m_i,m_j \notin \{\Prop(w,\>,1), \Prop(w,\>,2)\}$, then $\mu'(w) = \mu(w)$.
 	\item if $m_i = \Prop(w,\>,2)$ and $m_j = \Prop(w,\>,3)$, then $\mu'(w) \in \{\mu(w),m_j\}$.
	\item if $m_i = \Prop(w,\>,1)$ and $m_j = \Prop(w,\>,2)$, then $\Prop(w,\>',2) \in \{m_i, m_j\}$.
 \end{enumerate}
 \end{restatable}

In words, the matched partner remains the same when at least one of the swapped agents is a non-proposer (case 1), or when neither of the swapped agents is the first or the second-best proposer (case 2). When swapping the second and third-best proposers, the matched partner can either stay the same or worsen (with respect to true preferences). Further, when the first and the second-best proposers are swapped, there can be at most one new proposal that is better than the old partner (with respect to the true preference list).

\begin{remark}
Note that whenever \Cref{lem:SwappingLemma} guarantees that $\mu'(w) = \mu(w)$ (e.g., in cases 1 and 2), then by \Cref{lem:Woman_Match_Changes}, we can, in fact, infer that $\mu'=\mu$. Furthermore, since men's preferences stay unchanged between $\>$ and $\>'$, the entire set of proposals during the execution of \DA{} algorithm also remains unchanged between $\>$ and $\>'$.
\label{rem:Stronger_Swapping_Lemma}
\end{remark}

Before proving our next result (\Cref{prop:WomanListInconspicuous}), we recall that the set $S_w$ denotes the set of preference lists that can be obtained from $w$'s true list by moving exactly two men to the top. Although the lists in $S_w$ are \emph{not necessarily inconspicuous}, we note that the proof of this result relies heavily on that of Theorem 4 in \citet{VG17manipulating}, which shows that the partner a woman achieves via an \emph{optimal} misreport can also be achieved using an inconspicuous list.

\begin{figure*}[h]
\centering
  \begin{tabular}{@{}c@{}}
\begin{tikzpicture}
\tikzset{edge1/.style = {->}}

\node[rectangle,draw=none] (10) at (0.2,4.3)   {$\>'_w$};
\node[rectangle,draw=none] (11) at (0.2,3.8)   {$\vdots$};
\node[rectangle,draw]         (12) at (0.2,3)     
{{$p$}};
\node[rectangle,draw=none] (13) at (0.2,2.3)   {$\vdots$};
\node[rectangle,draw=none] (14) at (0.2,1.5)   
{{$q$}};
\node[rectangle,draw=none] (15) at (0.2,1)      {$\vdots$};
\node[rectangle,draw=none] (16) at (0.2,0.58)  {$\vdots$};

\draw[edge1](14) to[looseness=1.8,out=35,in=35] (13);

\node[rectangle,draw=none] (20) at (2,4.3)   {$\>^{(1)}_w$};
\node[rectangle,draw=none] (21) at (2,3.8)   {$\vdots$};
\node[rectangle,draw]         (22) at (2,3)     {{$p$}};
\node[rectangle,draw=none] (23) at (2,2.4)   {{$q$}};
\node[rectangle,draw=none] (24) at (2,1.9)   {$\vdots$};
\node[rectangle,draw=none] (25) at (2,1.17)  {$\vdots$};
\node[rectangle,draw=none] (26) at (2,0.75)  {$\vdots$};

\draw [decorate,decoration={brace,amplitude=5pt}] (2.25,3.3) -- (2.25,2.3) node [black,midway,xshift=0cm,yshift=0.3cm] (b) {};

\draw[edge1](b) to[looseness=1.15,out=0,in=0] (21);

\node[rectangle,draw=none] (30) at (4.1,4.3)   {$\>^{(2)}_w$};
\node[rectangle,draw]         (32) at (4.1,3.7)     {{$p$}};
\node[rectangle,draw=none] (33) at (4.1,3.1)   {{$q$}};
\node[rectangle,draw=none] (31) at (4.1,2.6)   {$\vdots$};
\node[rectangle,draw=none] (34) at (4.1,1.9)   {$\vdots$};
\node[rectangle,draw=none] (35) at (4.1,1.17)  {$\vdots$};
\node[rectangle,draw=none] (36) at (4.1,0.75)  {$\vdots$};

\node[rectangle,draw=none] (40) at (7.5,4.3)   {$\>^{(k)}_w$};
\node[rectangle,draw]         (32) at (7.5,3.7)     {{$p$}};
\node[rectangle,draw=none] (33) at (7.5,3.1)   {{$q$}};
\node[rectangle,draw=none] (31) at (7.5,2.6)   {$\vdots$};
\node[rectangle,draw=none] (44) at (7.5,1.9)   {$\vdots$};
\node[rectangle,draw=none] (45) at (7.5,1.2)  { $m_j$};
\node[rectangle,draw=none] (46) at (7.5,0.8)    { $m_i$};
\node[rectangle,draw=none] (47) at (7.5,0.5)  {$\vdots$};

\draw[edge1](45) to[looseness=2,out=0,in=0] (46);
\draw[edge1](46) to[looseness=2,out=180,in=180] (45);

\node[rectangle,draw=none] (50) at (9.4,4.3)   {$\>^{(k+1)}_w$};
\node[rectangle,draw]         (32) at (9.4,3.7)     {{$p$}};
\node[rectangle,draw=none] (33) at (9.4,3.1)   {{$q$}};
\node[rectangle,draw=none] (31) at (9.4,2.6)   {$\vdots$};
\node[rectangle,draw=none] (54) at (9.4,1.9)   {$\vdots$};
\node[rectangle,draw=none] (55) at (9.4,1.2)  
{{$m_i$}};
\node[rectangle,draw=none] (56) at (9.4,0.8)    
{{$m_j$}};
\node[rectangle,draw=none] (57) at (9.4,0.5)  {$\vdots$};

\node[rectangle,draw=none] (60) at (11.5,4.3)   
{$\>^{(\ell)}_w=\>''_w$};
\node[rectangle,draw]         (32) at (11.5,3.7)     {{$p$}};
\node[rectangle,draw=none] (33) at (11.5,3.1)   {{$q$}};
\node[rectangle,draw=none] (31) at (11.5,2.6)   {$\vdots$};
\node[rectangle,draw=none] (64) at (11.5,1.9)   {$\vdots$};
\node[rectangle,draw=none] (65) at (11.5,1.17)  {$\vdots$};
\node[rectangle,draw=none] (66) at (11.5,0.75)  {$\vdots$};

\draw [decorate,decoration={brace,amplitude=7pt}] (11.7,2.8) -- (11.7,0.3) node [black,midway,xshift=1.8cm,yshift=0cm, text width=2.3cm] { resembles $\>_w$\\ except for $\{p,q\}$};

\path (30) -- node[auto=false]{\ldots} (40);
\path (50) -- node[auto=false]{\ldots} (60);

\end{tikzpicture}
  \end{tabular}
  \caption{The sequence of preference lists constructed in the proof of \Cref{prop:WomanListInconspicuous} (based on a similar construction in Theorem 4 in \citet{VG17manipulating}). Here, $p \coloneqq \Prop(w,\>',1) = \mu'(w)$ and $q \coloneqq \Prop(w,\>',2)$. %
  }
\label{fig:Inconspicuous-manipulation-one-agent}
\end{figure*}
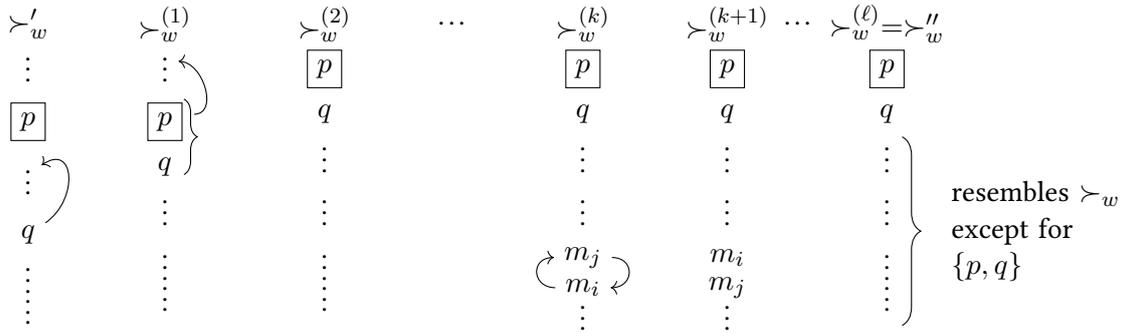

\WomanListInconspicuous*

\begin{proof}
Our proof of \Cref{prop:WomanListInconspicuous} closely follows that of Theorem 4 in \citet{VG17manipulating}, which showed that the partner a woman can achieve by an \emph{optimal} misreport can also be achieved using an inconspicuous list. 
For the sake of completeness, we recall parts of their argument below with the necessary modifications to establish \Cref{prop:WomanListInconspicuous}.

Due to \Cref{lem:Woman_Match_Changes}, it suffices to show that there exists some list $\>''_w \, \in S_w$ such that $\mu''(w) = \mu'(w)$, where $\mu'' \coloneqq \DA(\>_{-w}, \>''_w)$. 

For the profile $\>' \coloneqq \{\>_{-w},\>'_w\}$, let $p \coloneqq \Prop(w, \>', 1)$ and $q \coloneqq \Prop(w, \>', 2)$ denote the favorite and second-favorite proposers of $w$ (with respect to the list $\>'_w$) during the execution of \DA{} algorithm on $\>'$. Thus, $p = \mu'(w)$.

Starting from the true list $\>_w$, let $\>''_w$ be the preference list %
obtained by promoting $p$ %
to the top and then placing the man $q$ 
in the position immediately below $p$. Notice that $\>''_w \, \in S_w$.

In the remainder of the proof, we will show that the misreported list $\>'_w$ can be transformed into $\>''_w$ by a sequence of \emph{massaging} operations (involving swaps of adjacent pairs of men) such that the man $p = \mu'(w)$ continues to propose to $w$ at each intermediate step of the transformation. Further, we will argue that $p$ is the \emph{best} proposal that $w$ receives with respect to each intermediate list (which ensures that $p$ remains the \DA{} partner of $w$ at each intermediate step). Note that proving this would imply that $\mu''(w) = \mu'(w)$.

We will first describe the transformation of $\>'_w$ into $\>''_w$ followed by arguing that doing so maintains the property that $p$ will propose to $w$, and importantly, that $p$ will continue to be the \DA{} partner of $w$ for each intermediate matching.

\begin{enumerate}%
    \item \emph{Constructing $\>''_w$}: Starting with $\>'_w$, we construct a sequence of preference lists $\>^{(1)}_w, \>^{(2)}_w, \dots, \>^{(\ell)}_w$ that terminates with $\>''_w$, i.e., $\>''_w \, = \, \>^{(\ell)}_w$ (see Figure~\ref{fig:Inconspicuous-manipulation-one-agent}):
	\begin{enumerate}
	    \item \emph{Moving $q$ next to $p$}: We create the list $\>^{(1)}_w$ from $\>'_w$ by moving $q$ to the position right below $p$. Since all men who are ranked between $p$ and $q$ in the list $\>'_w$ must be non-proposers under $\>'$, it follows from case (1) of \Cref{lem:SwappingLemma} that $\mu'(w) = \mu^{(1)}(w)$, where $\mu^{(1)} \coloneqq \DA(\>_{-w}, \>^{(1)}_w)$. Further, by \Cref{rem:Stronger_Swapping_Lemma}, the set of proposals under the \DA{} algorithm on the profiles $\>$ and $\>^{(1)}$ is exactly the same, implying that $\Prop(w,\>^{(1)},1) = p$ and ${\Prop(w,\>^{(1)},2) = q}$.
	    Note that any man above $p$ in the list $\>^{(1)}_w$ must still be a non-proposer.
	    
	    \item \emph{Promoting $p$ and $q$}: In this step, starting from the list $\>^{(1)}_w$, we promote $p$ to the top rank position and $q$ to the second rank position. %
	    Since $p$ is $w$'s favorite proposer according to $\>^{(1)}_w$, %
	    no man above $p$ could have proposed to her under $\>'$ or $\>^{(1)}$. Thus, we are only required to swap a proposer with a non-proposer at each step, %
	    which, by case (1) of \Cref{lem:SwappingLemma}, does not affect $w$'s match, and thus, by \Cref{rem:Stronger_Swapping_Lemma}, also does not affect the run of \DA{} algorithm.
	    We call the resulting preference list $\>^{(2)}_w$. Clearly, ${\Prop(w,\>^{(2)},1) = p}$ and ${\Prop(w,\>^{(2)},2) = q}$, where ${\>^{(2)} \coloneqq \{\>_{-w},\>^{(2)}_w\}}$.

	    \item \emph{Fixing the part of the list below $q$}: The final step in our construction involves a sequence of preference lists $\{\>^{(3)}_w, \>^{(4)}_w, \dots\}$. The list $\>^{(k+1)}_w$ is derived from $\>^{(k)}_w$ by \emph{swapping} a pair of adjacent men $(m_i,m_j)$ in $\>^{(k)}_w$ such that:
	    \begin{enumerate}[(i)]
	        \item $\Prop(w,\>^{(k)},2) \>^{(k)}_w m_i$ and $\Prop(w,\>^{(k)},2) \>^{(k)}_w m_j$, and 
	        \item $m_i \, \>_w \, m_j$ and $m_j \, \>^{(k)}_w \, m_i$.
	    \end{enumerate}
	    That is, each new list in the sequence is derived from the previous list by swapping a pair of adjacent men who are (i) both positioned below the second-favorite proposer according to the previous list, and (ii) are incorrectly ordered with respect to the true preference list $\>_w$. No other changes are made. 
	    Notice that this sequence of preference lists must be finite since there can only be a finite number of pairs of men who are incorrectly ordered with respect to the true list. Let $\>^{(\ell)}_w$ be the final list in this sequence, and let ${\>^{(k)} \coloneqq \{\>_{-w},\>^{(k)}_w\}}$ and ${\mu^{(k)} = \DA(\>^{(k)})}$ denote the preference profile and the corresponding \DA{} matching at each step. This finishes the construction of the sequence of preference lists.
    \end{enumerate}

    \item \emph{Correctness}: It follows from the above construction that the list $\>^{(\ell)}_w$ can be obtained from the true list $\>_w$ by promoting $p$ to the top and then placing $q$ to the position right below $p$, while making no other changes. Indeed, steps 1(a) and 1(b) secure the positions of $p$ and $q$ at the top of $\>^{(\ell)}_w$,
    while step 1(c) incrementally corrects for pairs that are out of order with respect to $\>_w$ and eventually terminates with a list $\>^{(\ell)}_w$ that is identical to $\>_w$ with the exception of the positions of men $p$ and $q$.
    
    Now, we must show that $\Prop(w,\>^{(\ell)},1) = \mu^{(\ell)}(w) = p$. We do so using induction. The base case consists of showing that $\mu^{(3)}(w) = p$. Let $(m_i,m_j)$ be the pair of adjacent men in $\>^{(2)}_w$ that are swapped in $\>^{(3)}_w$. By construction, we know that $m_i, m_j \notin \{{\Prop(w,\>^{(2)},1),} \, {\Prop(w,\>^{(2)},2)}\}$. Therefore, from cases (1) and (2) of \Cref{lem:SwappingLemma}, we have that $\mu^{(3)}(w) = p$. Additionally, from \Cref{rem:Stronger_Swapping_Lemma}, the set of proposals under the \DA{} algorithm remains unchanged.

	Our induction hypothesis is that $\mu^{(k)}(w) = p$ for all $3 < k \leq K$. Given this, we show that $\mu^{(K+1)}(w) = p$. %
	
	As before, let $(m_i,m_j)$ be the pair of adjacent men that are swapped in $\>^{(K)}_w$ to obtain $\>^{(K+1)}_w$. By construction, we again have that $m_i, m_j \notin \{{\Prop(w,\>^{(K)},1),} \, {\Prop(w,\>^{(K)},2)}\}$. Therefore, from cases (1) and (2) of \Cref{lem:SwappingLemma}, we have that $\mu^{(K+1)}(w) = \mu^{(K)}(w)$. Finally, using the induction hypothesis, we get that $\mu^{(K+1)}(w) = p$.
\end{enumerate}
Hence, by induction, we have that $\mu^{(\ell)}(w) = p = \mu'(w)$. %
Since $\>^{(\ell)} \, = \, \>''_w$ by construction, and $\>''_w \, \in S_w$, we substantiate \Cref{prop:WomanListInconspicuous}.
\end{proof}

\begin{remark}[\textbf{Key difference from the construction of \citet{VG17manipulating}}]
One might wonder why the proof of \Cref{prop:WomanListInconspicuous} does not follow the construction of \citet{VG17manipulating} \emph{exactly}. Specifically, one might ask why it is not enough to simply promote $q$ to a position immediately below $p$ (as done by \citet{VG17manipulating} in the context of self manipulation), without further promoting these agents to the top two positions.

Note that \citet{VG17manipulating} invoke \Cref{lem:SwappingLemma} in the context of an \emph{optimal} self manipulation. This ensures that the agents above $p$ in the true list are guaranteed to be non-proposers.

By contrast, our construction in \Cref{prop:WomanListInconspicuous} works for an \emph{arbitrary} misreport by $w$. Thus, it is possible that there is a proposer $p'$ that is below $p$ in the misreported list $\>'_w$ but is above $p$ in the true list $\>_w$. If $p$ is not promoted to the top position, then in order to ensure that the agents in $\>'_w \setminus \{p,q\}$ are ordered exactly as in $\>_w \setminus \{p,q\}$, we might need to exchange $p$ and $p'$. However, this could change the \DA{} match of $w$ (case 4 of \Cref{lem:SwappingLemma}). For this reason, we need to promote $p$ to the top position.

For a similar reason, we also need to promote $q$ to immediately after $p$, thus preventing the swap of the second and third-best proposers and again avoiding a change in the \DA{} partner (case 3 of \Cref{lem:SwappingLemma}).
\end{remark}

\citet{VG17manipulating} showed that optimal self manipulation is inconspicuous, i.e., only one man's position needs to be changed from the true list. One may ask whether \emph{any} misreport---optimal, suboptimal, or with-regret---is also inconspicuous. \Cref{eg:with-regret-woman-manipulation-not-inconspicuous} gives a negative answer to this question. More specifically, we show that a with-regret misreport by the woman is not guaranteed to be inconspicuous. We leave the question of whether supoptimal manipulations by women are inconspicuous as an open problem.

Our interest in with-regret misreports by women becomes apparent in the proof of \Cref{thm:PairManipulation}, where we discuss the problems that arise in  establishing inconspicuousness of pair manipulation.

\begin{example}[\textbf{With-regret misreport by a woman may not be inconspicuous}] \label{eg:with-regret-woman-manipulation-not-inconspicuous}
Consider the following preference profile where the \DA{} outcome is underlined.

\begin{table}[H]
    \centering
    \begin{tabularx}{0.7\linewidth}{XXXXXXXXXXXXXXX}
            $m_1\colon$ & $\underline{w_3}$ & $w_1$ & $w_2^*$ & $w_4$ & $w_5$ && $\boldsymbol{\textcolor{blue}{{w_1}}}\colon$ & $m_1$ & $m_2$ & $\underline{m_3}$ & $m_4^*$ & $m_5$\\
            ${m_2}\colon$ & $\underline{w_4^*}$ & $w_5$ & $w_3$ & $w_1$ & $w_2$ && ${w_2}\colon$ & $m_3$ & $m_5$ & $m_1^*$ & $m_2$ & $\underline{m_4}$\\
            ${m_3}\colon$ & $\underline{w_1}$ & $w_3^*$ & $w_5$ & $w_4$ & $w_2$ && ${w_3}\colon$ & $m_5$ & $m_2$ & $m_3^*$ & $m_4$ & $\underline{m_1}$\\
            ${m_4}\colon$ & $\underline{w_2}$ & $w_1^*$ & $w_3$ & $w_4$ & $w_5$ && ${w_4}\colon$ & $\underline{m_2^*}$ & $m_3$ & $m_1$ & $m_5$ & $m_4$\\
            ${m_5}\colon$ & $w_1$ & $\underline{w_5^*}$ & $w_3$ & $w_2$ & $w_4$ && ${w_5}\colon$ & $m_1$ & $\underline{m_5^*}$ & $m_3$ & $m_2$ & $m_4$
        \end{tabularx}
\end{table}

The \DA{} matching after $w_1$ submits $\>'_{w_1} \coloneqq m_4 \> m_5 \> m_1 \> m_2 \> m_3$ is marked by ``$*$''. Notice that the matching is with-regret for $w_1$ as she is assigned $m_4$, who she prefers less than her \DA{} match $m_3$. A polynomial number of checks can be made to verify that there is no inconspicuous list that can achieve the same matching.

Although $\>'_{w_1}$ is not an inconspicuous version of $w_1$'s true list $\>_{w_1}$, it can be created from $\>_{w_1}$ just by promoting $m_4 \coloneqq \Prop(w, \>', 1)$, where $\>' \coloneqq \{\>_{-w_1}, \>'_{w_1}\}$, to the topmost position and $m_5 \coloneqq \Prop(w, \>' \coloneqq \{\>_{-w_1}, 2)$ to the second-most position, which is the same construction that was described in the proof of \Cref{prop:WomanListInconspicuous}.
\qed
\end{example}

\subsection{Proof of Lemma~\ref{lem:AccompliceMisreports_PairManipulation}}

Recall that $\widehat{\>}_m$ is the preference list obtained by promoting man $m$'s original match $\mu(m)$ to the top of his original list $\>_m$, and $S_m \coloneqq \{\widehat{\>}_m^{w'\uparrow} : w' \neq \mu(m)\} \cup \{\widehat{\>}_m\}$ is the set of preference lists obtained by individually pushing up each woman other than $\mu(m)$ to the top of $\widehat{\>}_m$ as well as the list $\widehat{\>}_m$ itself.

\AccompliceMisreportsPairManipulation*

\begin{proof}
To construct the desired list $\>''_m$, we will make use of the following three preference profiles:
\begin{itemize}
    \item $\>^w \coloneqq \{\>_{-\{m,w\}}, \>'_{w}, \>_m\}$, %
    \item $\widehat{\>} \coloneqq \{\>_{-\{m,w\}},  \>'_{w},\widehat{\>}_m\}$, and
    \item $\>' \coloneqq \{\>_{-\{m,w\}}, \>'_{w},\>'_m\}$.
\end{itemize}
Notice that the three profiles differ only in the preferences of man $m$. We will denote the corresponding \DA{} matchings for these profiles by $\mu^w \coloneqq \DA(\>^w)$, $\widehat{\mu} \coloneqq \DA(\widehat{\>})$ and $\mu' \coloneqq \DA(\>')$.

We will start by observing that $\mu^w(m) \succeq_{m} \mu(m)$. Indeed, if $\mu(m) \>_{m} \mu^w(m)$, then man $m$ can switch from the profile $\>^w$ (where his list is the same as his true list $\>_m$) to the profile $\>'$ (where his list is $\>'_m$) and improve his \DA{} match with respect to $\>^w$; recall that it is given that $\mu'(m) = \mu(m)$. This, however, contradicts the strategyproofness of \DA{} algorithm for the proposing side. Thus, we must have $\mu^w(m) \succeq_{m} \mu(m)$.

The rest of the proof is divided into two cases, depending on whether $\mu^w(m)$ is the same as $\mu(m)$.

\noindent
\newline
\textbf{Case I}: When $\mu^w(m) \neq \mu(m)$.\\

Recall from \Cref{prop:Permuting_Falsified_Lists} that any misreport by man $m$ can be simulated via push up and push down operations in the true list. %
Therefore, we can assume without loss of generality that the list $\>'_m$ is derived from $\>_m$ by pushing up a set $Y' \subseteq W$ of women who are ranked below $\mu^w(m)$ and pushing down a set $X \subseteq W$ of women who are ranked above $\mu^w(m)$.

Further, let us write man $m$'s list under the profile $\>^w$ as $\>_m \coloneqq (\>^L_m,\mu^w(m),\>^R_m)$ where $\>^L_m$ and $\>^R_m$ denote the parts of $m$'s list $\>_m$ above and below his \DA{} partner $\mu^w(m)$, respectively. Based on this notation, we let $\bar{X} \coloneqq \, {\>^L_m \setminus X}$ denote the set of women other than $X$ who are ranked above $\mu^w(m)$ in $\>_m$, and analogously write $\bar{Y} \coloneqq \, {\>^R_m \setminus Y'}$ for the set of women other than $Y'$ who are ranked below $\mu(m)$ in $\>_m$. By \Cref{prop:Permuting_Falsified_Lists}, permuting the order of agents above (or below) the \DA{} partner in man $m$'s list does not affect the \DA{} matching. Therefore, we can represent the lists $\>_m$ and $\>'_m$ as follows:
$$\>_m : X \> \bar{X} \> \boxed{\mu^w(m)} \> Y' \> \bar{Y}, \text{ and }$$
$$\>'_m : \bar{X} \> Y' \> \mu^w(m) \> X \> \bar{Y},$$
where the boxed entry is the \DA{} partner of $m$ under $\>_m$.

By the assumption in Case I, we have $\mu^w(m) \>_m \mu(m)$ and thus $\mu(m) \in Y' \cup \bar{Y}$. We are also given that $\mu'(m) = \mu(m)$. Therefore, the transition from the profile $\>^w$ to the profile $\>'$, where only $m$ misreports, must be \emph{with-regret} for $m$.

Consider another profile $(\>^w)^{X \downarrow}$ obtained from $\>^w$ by pushing down the set $X$ in man $m$'s list; that is,
$$\>_m^{X \downarrow} : \bar{X} \> \boxed{\mu^w(m)} \> X \> Y' \> \bar{Y}.$$
Notice that by \Cref{prop:PushDown}, man $m$ continues to be matched with $\mu^w(m)$ under $(\>^w)^{X \downarrow}$ and therefore, by \Cref{prop:Permuting_Falsified_Lists}, we can assume that the sets $X$, $Y$, and $Y'$ can be ordered as shown above.

Observe that the list $\>'_m$ can be obtained from $\>_m^{X \downarrow}$ by pushing up the set $Y'$. Since $\mu(m) \in Y' \cup \bar{Y}$, the transition $(\>^w)^{X \downarrow} \rightarrow \, \>'$ constitutes a with-regret push up operation for man $m$. Then, by \Cref{lem:accomplice-match-after-regret-manipulation}, man $m$ must be matched with some agent in $Y'$ under $\>'$. We already know that $\mu'(m) = \mu(m)$, and therefore we have that $\mu(m) \in Y'$. Let $Y \coloneqq Y' \setminus \{\mu(m)\}$. The preference lists after this simplification are shown in \Cref{fig:AccompliceMisreports_PairManipulation_Case-I}. Notice that due to \Cref{prop:Permuting_Falsified_Lists}, we can assume that $\mu(m)$ is placed immediately below $\mu^w(m)$ in $\>_m$, and, as a result, immediately above it in $\>'_m$.

\begin{figure*}[th]
\centering
\begin{tikzcd}[cramped, column sep=tiny]
    \succ_m: & X & \bar{X} & \boxed{\mu^w(m)} & \mu(m) & Y & \bar{Y}\\
    \succ'_m: & \bar{X} & Y & \boxed{\mu(m)} & \mu^w(m) & X & \bar{Y}\\
    \succ'_m (\text{eqv.}): & \bar{X} & Y & \boxed{\mu(m)} & X & \mu^w(m) & \bar{Y}\\
    \widehat{\succ}_m: & \boxed{\mu(m)} & \mu^w(m) & X & \bar{X} &  Y  & \bar{Y}\\
\end{tikzcd}
\caption{An illustration of man $m$'s preference lists under the profiles $\succ^w$, $\succ'$, $\succ' (\text{eqv.})$, and $\widehat{\succ}$ in the proof of \Cref{lem:AccompliceMisreports_PairManipulation} for Case I (when $\mu^w(m) \neq \mu(m)$). The boxed agents correspond to \DA{} partners of man $m$ under the respective profiles. Recall that it follows from \Cref{prop:Permuting_Falsified_Lists} that the order of women below these partners is not crucial.}
\label{fig:AccompliceMisreports_PairManipulation_Case-I}
\end{figure*}

Let us now consider the list $\widehat{\>}_m$ obtained by promoting man $m$'s original match $\mu(m)$ to the top of his original list $\>_m$, i.e.,
$$\widehat{\>}_m : \mu(m) \> X \> \bar{X} \> \mu^w(m) \> Y \> \bar{Y}.$$
Notice that $\widehat{\>}_m \in S_m$. We will argue that $m$ matches with $\mu(m)$ under $\widehat{\>}_m$. Note that because of \Cref{prop:Permuting_Falsified_Lists}, the list $\>'_m$ is equivalent to the following list:
$$\succ'_m (\text{eqv.}): \bar{X} \> Y \> \boxed{\mu(m)} \> X \> \mu^w(m) \> \bar{Y}.$$
Then, the list $\widehat{\>}_m$ can be obtained by pushing down the sets $\bar{X}$ and $Y$ in the list $\>'_m (\text{eqv.})$ and placing these sets between $X$ and $\bar{Y}$. Therefore, by \Cref{prop:PushDown}, $m$ matches with $\mu(m)$ under $\widehat{\>}$ (see \Cref{fig:AccompliceMisreports_PairManipulation_Case-I}).

Since $\widehat{\mu}(m) = \mu(m)$, the list $\>'_m (\text{eqv.})$ can be created from the list $\widehat{\>}_m$ via a no-regret push up of the set $\bar{X} \cup Y$. If $\bar{X} \cup Y = \emptyset$, then $\widehat{\>}_m \, = \, \>'_m$ and the choice $\>''_m \coloneqq \widehat{\>}_m$ proves the lemma. Otherwise, if $\bar{X} \cup Y \neq \emptyset$, then it follows from \Cref{prop:PushUpOneWoman} that the same match for $w$ that is created by pushing up $\bar{X} \cup Y$ in $\widehat{\>}_m$, namely $\mu'(w) = \mu(w)$, can be created by pushing up exactly one woman in $\bar{X} \cup Y$, say $w'$. Additionally, from \Cref{lem:SingleAgentPushUpNoRegret}, it follows that $m$ would still match with $\widehat{\mu}_m$ (hence $\mu(m)$) after pushing up $w'$ in $\widehat{\>}_m$. Since $\widehat{\>}^{w' \uparrow}_m \, \in S_m$, we have that the choice $\>''_m \coloneqq \widehat{\>}^{w' \uparrow}_m$ proves the lemma.

\noindent
\newline
\textbf{Case II}: When $\mu^w(m) = \mu(m)$.\\

Following a similar line of reasoning as in Case I, we can once again say that the list $\>'_m$ is derived from $\>_m$ by pushing up the set $Y$ and pushing down the set $X$. Further, the list $\widehat{\>}_m$, which is obtained by pushing up $\mu(m)$ to the top position in $\>_m$, can be equivalently derived from $\>'_m$ by pushing down $\bar{X} \cup Y$. Thus, by \Cref{prop:PushDown}, we once again have that $\widehat{\mu}(m) = \mu(m)$ (see \Cref{fig:AccompliceMisreports_PairManipulation_Case-II}).

\begin{figure*}[th]
\centering
\begin{tikzcd}[cramped, column sep=tiny]
    \succ_m: & X & \bar{X} & \boxed{\mu(m)} & Y & \bar{Y}\\
    \succ'_m: & \bar{X} & Y & \boxed{\mu(m)} & X & \bar{Y}\\
    \widehat{\succ}_m: & \boxed{\mu(m)} & X & \bar{X} &  Y  & \bar{Y}\\
\end{tikzcd}
\caption{An illustration of man $m$'s preference lists under the profiles $\succ^w$, $\succ'$, and $\widehat{\succ}$ in the proof of \Cref{lem:AccompliceMisreports_PairManipulation} for Case II (when $\mu^w(m) = \mu(m)$). The boxed agents correspond to \DA{} partners of man $m$ under the respective profiles. Recall that it follows from \Cref{prop:Permuting_Falsified_Lists} that the order of women below these partners is not crucial.
}
\label{fig:AccompliceMisreports_PairManipulation_Case-II}
\end{figure*}

If $\bar{X} \cup Y = \emptyset$, then $\widehat{\>}_m \, = \, \>'_m$ and the choice $\>''_m \coloneqq \widehat{\>}_m$ proves the lemma since $\widehat{\>}_m \in S_m$. Otherwise, if $\bar{X} \cup Y \neq \emptyset$, then it follows from \Cref{prop:PushUpOneWoman} that the same match for $w$ that is created by pushing up $\bar{X} \cup Y$ in $\widehat{\>}_m$ to obtain $\>'_,$, namely $\mu'(w) = \mu(w)$, can be created by pushing up exactly one woman in $\bar{X} \cup Y$, say $w'$. Additionally, from \Cref{lem:SingleAgentPushUpNoRegret}, it follows that $m$ would still match with $\widehat{\mu}_m$ (hence $\mu(m)$) after pushing up $w'$ in $\widehat{\>}_m$. Since $\widehat{\>}^{w' \uparrow}_m \, \in S_m$, we have that the choice $\>''_m \coloneqq \widehat{\>}^{w' \uparrow}_m$ proves the lemma.
\end{proof}

\subsection{Proof of Theorem~\ref{thm:PairManipulation}}
\label{Proof_PairManipulation}

\paragraph{Why inconspicuousness is not straightforward to establish.}

Recall that optimal self manipulation and optimal accomplice manipulation are both known to be \emph{inconspicuous}, i.e., very similar to the true list~\citep{VG17manipulating,HUV21accomplice}. Prompted by the failure of individually optimal strategies in %
\Cref{eg:concatenate_pair}, we ask whether an inconspicuous strategy for the man could be combined with an inconspicuous strategy for the woman to yield an optimal pair manipulation.

More concretely, consider a manipulating pair $(m,w)$. %
Let $\>^*_m$ and $\>^*_w$ denote the respective lists of $m$ and $w$ under an optimal pair manipulation, and let $\>^* \coloneqq (\>_{-\{m,w\}},\>^*_m,\>^*_w)$ denote the corresponding preference profile.

\begin{figure}[H]
\centering
\begin{tikzpicture}
    \node[] at (0, 0) (l) {$\>$};
    \node[] at (4, 0.6)   (t) {$\>^m = \{\>_{-m},\>^*_m\}$};
    \node[] at (8, 0)   (r) {$\>^*$};
    \node[] at (4,-0.6)   (b) {$\>^w= \{\>_{-w},\>^*_w\}$};
    \draw[->] (l)--(t) node [midway,above=2pt,sloped,fill=white] {{$\>_m \rightarrow \>^*_m$}};
    \draw[->] (l)--(b) node [midway,below=2pt,sloped,fill=white] {{$\>_w \rightarrow \>^*_w$}};
    \draw[->] (t)--(r) node [midway,above=2pt,sloped,fill=white] {{$\>_w \rightarrow \>^*_w$}};
    \draw[->] (b)--(r) node [midway,below=2pt,sloped,fill=white] {{$\>_m \rightarrow \>^*_m$}};
\end{tikzpicture}
\caption{Preference profiles under pair manipulation.}
\label{fig:PairManipulation-Appendix}
\end{figure}
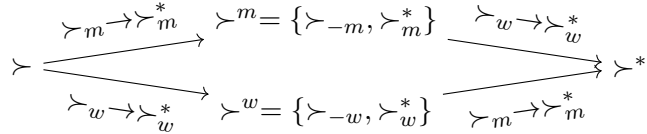

Since it is easier to think about single-agent misreports, let us break down the transition from the true profile $\>$ to the pair manipulation profile $\>^*$ in two steps: First, swap $m$'s list in $\>$ to obtain the intermediate profile $\>^m \coloneqq \{\>_{-m},\>^*_m\}$, and then swap $w$'s list in $\>^m$ to get $\>^*$; see top path in \Cref{fig:PairManipulation-Appendix}. Alternatively, by reversing the order of swaps and going via the profile $\>^w \coloneqq \{\>_{-w},\>^*_w\}$, we obtain the bottom path in \Cref{fig:PairManipulation-Appendix}.

Let us first consider the top path in \Cref{fig:PairManipulation-Appendix}, and specifically the transition $\>^m \rightarrow \, \>^*$ wherein woman $w$ is the only misreporting agent. Thus, the optimal pair manipulation strategy must involve $w$ performing an optimal self manipulation with respect to the profile $\>^m$ \emph{subject to $m$ being matched with $\mu(m)$}. In other words, the strategy $\>^*_w$ is an optimal solution to a \emph{constrained self manipulation} problem. While the unconstrained self manipulation problem has been explored previously~\citep{TS01gale,VG17manipulating}, we are not aware of results on the structure of optimal solution of the aforementioned constrained problem.\footnote{The works of \citet{KM09successful,KM10cheating} are relevant here. They study the computational complexity of determining whether, given the preferences of men and non-strategic women as well as a partial matching in which the strategic women are matched, there is a set of misreports for the strategic women that could induce a \DA{} matching that extends the given partial matching.} In particular, we do not know if the optimal constrained self manipulation problem also admits an inconspicuous solution. Recall that a with-regret misreport by a woman is \textit{not} guaranteed to be inconspicuous (see \Cref{eg:with-regret-woman-manipulation-not-inconspicuous}). %

Now consider the transition $\>^w \rightarrow \, \>^*$ in the bottom path in \Cref{fig:PairManipulation-Appendix}. For this step, man $m$ is the only misreporting agent. Note that $m$'s list under the profile $\>^w$ is his true list $\>_m$. Further, under $\>^w$, man $m$ must be matched with either $\mu(m)$ or some agent he ranks above $\mu(m)$; note that $m$'s partner under $\>^w$ cannot be someone he ranks below $\mu(m)$ since strategyproofness of \DA{} for proposing side would forbid $m$ from being matched with $\mu(m)$ under $\>^*$.

If the optimal pair strategy $\>^*_w$ of woman $w$ is such that her optimal partner under $\>^*$, namely $\mu^*(w)$, is ranked above her partner under $\>^w$, namely $\mu^w(w)$, then the transition $\>^w \rightarrow \, \>^*$ is an optimal no-regret accomplice manipulation for man $m$, which is known to be inconspicuous~\citep{HUV21accomplice}. On the other hand, if $\mu^w(m) \neq \mu(m)$, then the same transition is an optimal solution to a \emph{constrained with-regret manipulation problem}. Whether such strategies are inconspicuous remains an open problem.

We will now present the proof of \Cref{thm:PairManipulation}.

\PairManipulation*

\begin{proof}
(sketch) \Cref{alg:PairManipulation} can be used to compute an optimal pair manipulation. The \emph{correctness} of the algorithm follows from the discussion on how to compute optimal joint strategies in the main text. Here, we will argue that \Cref{alg:PairManipulation} runs in $\O(n^5)$ time.

The sets $S_m$ and $S_w$ are of size $\O(n)$ and $\O(n^2)$, respectively. The algorithm checks all possible pairs $(\>'_m,\>'_w) \in S_m \times S_w$. The total number of checks is in $\O(n^3)$, and for each check, running the \DA{} algorithm takes $\O(n^2)$ time.
\end{proof}

\subsection{Stability-Preserving Pair Manipulation is Suboptimal~(\Cref{rem:Unstable_Pair})}

In \Cref{eg:pair_manipulation_unstable}, we show that the stability of the manipulated outcome (with respect to true preferences) can be in conflict with optimality of manipulation. Specifically, our example demonstrates that the matched partner of the woman under \emph{unrestricted} pair manipulation (i.e., when the manipulated matching is not required to be stable with respect to true preferences) can be strictly better than under optimal \emph{stability-preserving} pair manipulation.

\begin{example}[\textbf{Stable pair manipulation strategy is strictly suboptimal}] \label{eg:pair_manipulation_unstable}
Consider the following preference profile where the \DA{} outcome is underlined.

\begin{table}[H]
    \centering
    \begin{tabularx}{0.7\linewidth}{XXXXXXXXXXXXXXX}
            ${m_1}\colon$ & $\underline{w_1}$ & $w_4^*$ & $w_5$ & $w_3$ & $w_2$ && $\boldsymbol{\textcolor{blue}{{w_1}}}\colon$ & $m_2^*$ & $m_4$ & $\underline{m_1}$ & $m_3$ & $m_5$\\
            ${m_2}\colon$ & $\underline{w_4}$ & $w_3$ & $w_1^*$ & $w_5$ & $w_2$ && ${w_2}\colon$ & $\underline{m_5}$ & $m_2$ & $m_4$ & $m_3^*$ & $m_1$\\
            ${m_3}\colon$ & $w_1$ & $w_2^*$ & $\underline{w_5}$ & $w_4$ & $w_3$ && ${w_3}\colon$ & $m_5$ & $m_1$ & $\underline{m_4^*}$ & $m_2$ & $m_3$\\
            $\boldsymbol{\textcolor{blue}{{m_4}}}\colon$ & $w_5$ & $\underline{w_3^*}$ & $w_4$ & $w_1$ & $w_2$ && ${w_4}\colon$ & $m_3$ & $m_5$ & $m_4$ & $m_1^*$ & $\underline{m_2}$\\
            ${m_5}\colon$ & $w_5^*$ & $\underline{w_2}$ & $w_3$ & $w_1$ & $w_4$ && ${w_5}\colon$ & $m_1$ & $\underline{m_3}$ & $m_4$ & $m_5^*$ & $m_2$
        \end{tabularx}
\end{table}

Suppose the manipulating pair is $(m_4, w_1)$. The \DA{} matching after $m_4$ and $w_1$ submit their respective optimal pair manipulated lists $\>'_{m_4} \coloneqq w_3 \> w_5 \> w_4 \> w_1 \> w_2$ and $\>'_{w_1} \coloneqq m_4 \> m_2 \> m_3 \> m_1 \> m_5$ is marked by ``$*$''. Notice that the beneficiary $w_1$ is matched with her top choice partner $m_2$, which means that the pair manipulation is indeed optimal.

Although $m_4$ does not incur regret, the manipulated matching admits a blocking pair $(m_4, w_5)$ with respect to the true preferences. In this instance, there is no stability-preserving pair manipulation that is strictly better than truthful reporting.\footnote{This was checked using an exhaustive search of all possible manipulations by $m_4$ and $w_1$.}\qed
\end{example}

Despite the negative result that we have just presented, we are able to show that any pair that blocks a ``two-for-one'' manipulation must contain the accomplice or the strategic woman (\Cref{lem:m-w-Stability}). This property is similar to the concept of \emph{$m$-stable} matchings~\citep{BH19partners,HUV21accomplice}, which are defined as matchings such that any pair that blocks them with respect to true preferences must contain the accomplice $m$.

\begin{restatable}{lemma}{MWStability}
Given a preference profile $\>$ and a manipulating pair $(m, w)$, let $\>'$
denote the profile after $m$ and $w$ misreport their preferences. For any matching $\mu' \in S_{\>'}$ and for any pair $(m', w')$ that blocks $\mu'$ with respect to $\>$, it must be that $m' = m$ or $w' = w$. 
\label{lem:m-w-Stability}
\end{restatable}

\begin{proof}
Suppose, for contradiction, there exists a man-woman pair $(m', w')$ that blocks $\mu'$ with respect to $\>$ (i.e., $w' \>_{m'} \mu'(m')$ and $m' \>_{w'} \mu'(w')$) such that $m' \neq m$ and $w' \neq w$. Since $m$ and $w$ are the only agents whose preferences differ between $\>$ and $\>'$, we have that $\>_{m'} \, = \, \>'_{m'}$ and $\>_{w'} \, = \, \>'_{w'}$. Thus, $w' \>'_{m'} \mu'(m')$ and $m' \>'_{w'} \mu'(w')$, implying that the pair $(m',w')$ blocks $\mu'$ with respect to $\>'$, which contradicts the assumption that $\mu' \in S_{\>'}$. Therefore, it must be that $m' = m$ or $w' = w$.
\end{proof}

We note that although there do indeed exist instances where the accomplice is in a pair that blocks an optimal ``two-for-one'' manipulated matching (as shown in \Cref{eg:pair_manipulation_unstable}), we have not yet found any instances where a pair blocking the optimal ``two-for-one'' manipulated matching contains the strategic woman. We leave the question of whether optimal pair manipulations are strictly $m$-stable as an open problem.

\section{Omitted Material from \Cref{sec:one-for-all}}
\subsection{Proof of Proposition~\ref{lem:CombiningPushUpPushDown_AllWomen}}

\CombiningPushUpPushDownAllWomen*

\begin{proof}
Our proof %
closely follows that of Lemma 2 in \citet{HUV21accomplice}, which, under the same set of assumptions, showed that if $\mu''(w) \>_w \mu(w)$ for any fixed woman $w$, then $\mu'(w) \succeq_w \mu''(w)$. For completeness, we recall their argument below with the necessary modifications to establish \Cref{lem:CombiningPushUpPushDown_AllWomen}.

We will use case analysis based on whether or not $\mu'(m) = \mu(m)$, where $\mu \coloneqq \DA(\>)$.\\

\textbf{\underline{Case I}} (when $\mu'(m) = \mu(m)$): The list $\>''_m$ can be considered as being derived from $\>'_m$ via a push down operation on the set $Y$ (see \Cref{fig:combining-push-up-push-down}). From \Cref{lem:PushDown_Worse_For_Women}, we know that a push down operation is weakly worse for all women; thus, in particular, we get $\mu' \succeq_W \mu''$, as desired. Note that the relative ordering of $X$ and $Y$ above $\mu'(m)$ in the list $\>'_m$ is not important in light of \Cref{prop:Permuting_Falsified_Lists}.

\begin{figure}[H]
\centering
\begin{tikzcd}[cramped, column sep=tiny]
    \succ_m: & \cdots & Y & \cdots & \mu(m) & \cdots & X & \cdots\\
    \succ'_m: & \cdots & X & \cdots & Y & \cdots & \mu(m) & \cdots\\
    \succ''_m: & \cdots & X & \cdots & \mu(m) & \cdots & Y & \cdots
\end{tikzcd}
\caption{An illustration of man $m$'s preference lists under the profiles $\succ$,  %
$\succ'$, and $\succ''$ in the proof of \Cref{lem:CombiningPushUpPushDown_AllWomen}. Note that the relative positions of the women in sets $X$ and $Y$ are unimportant due to \Cref{prop:Permuting_Falsified_Lists}.}
\label{fig:combining-push-up-push-down}
\end{figure}

\textbf{\underline{Case II}} (when $\mu'(m) \neq \mu(m)$):
Suppose $\mu'(m) \in X$. Then, the list $\>''_m$ can be considered as being derived from $\>'_m$ via a permutation of the women below $\mu'(m)$ (see \Cref{fig:combining-push-up-push-down}). By \Cref{prop:Permuting_Falsified_Lists}, this implies that $\mu' = \mu''$, as desired.

Therefore, for the remainder of the proof, let us assume that $\mu'(m) \notin X$. Since $\>'_m$ is derived from $\>_m$ via a push up operation on the set $X$, and since $\mu'(m) \neq \mu(m)$ by assumption, we have that $\mu(m) \>_m \mu'(m)$. By \Cref{prop:Permuting_Falsified_Lists}, we can assume, without loss of generality, that $\mu'(m)$ is positioned immediately below $\mu(m)$ in the list $\>_m$. By construction, the same property also holds for the lists $\>'_m$ and $\>''_m$. Thus, $\>''_m$ can be considered as being obtained from $\>'_m$ via a push down operation on the set $Y$ (note that this operation is defined with respect to $\mu'(m)$). By \Cref{lem:PushDown_Worse_For_Women}, we have $\mu' \succeq_W \mu''$, as desired.
\end{proof}

\subsection{Proof of Lemma~\ref{lem:NoRegretSubset}}

\NoRegretSubset*
\begin{proof}
Suppose, for contradiction, that $\>^Y$ is a with-regret profile. Then, from \Cref{lem:accomplice-match-after-regret-manipulation}, we have that $m$ matches with some $w_y \in Y$ under $\>^Y$, i.e., $\mu^Y(m) = w_y$. By definition of the no-regret set $W^\textsc{NR}$, we know that the profile obtained by promoting only $w_y$, namely $\>^y \coloneqq \{ \succ_{-m}, \succ_m^{y\uparrow} \}$, is a no-regret profile, i.e., $\mu^y(m) = \mu(m)$, where $\mu^y \coloneqq \DA(\>^y)$. 

Now consider the profile $\>^y$ wherein $m$ is matched with $\mu(m)$. If $m$ were to submit the preference list $\>^Y_m$ instead (obtained by pushing up the set $Y \setminus \{w_y\}$ in $\>^y_m$), then we will obtain the profile $\>^Y$ in which $m$ is matched with $w_y$. This, however, implies that $m$ is unilaterally able to improve its match from $\mu(m)$ to $w_y$ (with respect to the new list $\>^y_m$ wherein $w_y \>^y_m \mu(m)$), which contradicts the strategyproofness of the \DA{} algorithm. Hence, $\>^Y$ must be a no-regret profile.
\end{proof}

\subsection{Proof of Lemma~\ref{lem:Combining_WithRegret_NoRegret_PushUp}}

\CombiningWithRegretNoRegretPushUp*

\begin{proof}
Suppose, for contradiction, that $\>^Y$ is a no-regret profile.
\Cref{lem:SingleAgentPushUpNoRegret} shows that the accomplice $m$ does not incur regret under $\>^y \coloneqq \{ \succ_{-m}, \succ_m^{y\uparrow} \}$ for all $w_y \in Y$. However, it is given that $\>^y$ is a with-regret profile, thus posing a contradiction. 
\end{proof}

\subsection{Proof of Corollary~\ref{cor:PolyAlgorithm_OptimalStrategy}}

\PolyAlgorithmOptimalStrategy*
\begin{proof}
(sketch) The algorithm promotes each woman who is below $\mu(m)$ in the accomplice $m$’s true preference list to some position above $\mu(m)$ and checks the \DA{} outcome.
Recall that \Cref{prop:Permuting_Falsified_Lists} gives us the flexibility to place the promoted woman anywhere above $\mu(m)$ in the accomplice's list.
The set of women who, when individually pushed up, do not cause regret are returned. The total number of checks is $\O(n)$, and for each check, running the \DA{} algorithm takes $\O(n^2)$ time.
\end{proof}

\subsection{Computing the Smallest-Size Optimal Strategy}
\label{sec:Smallest-size_One-for-all}

Although pushing up the entire no-regret set $W^\textsc{NR}$ is optimal~(\Cref{thm:OptimalStrategy_MultipleWomen}), the accomplice might want to displace as few women as possible in order to remain close to his true preference list. Thus, it is of interest to find an optimal push up set of the \emph{smallest size}.

Formally, consider any set $Y \subseteq W^\textsc{NR}$. Let $\mu^Y$ and $\mu^\textsc{NR}$ denote the \DA{} matchings obtained by pushing up the sets $Y$ and $W^\textsc{NR}$ in the true list of the accomplice, respectively. We call $Y$ a \emph{minimum} no-regret push up set if $\mu^Y = \mu^{\textsc{NR}}$ and there exists no other set $Z \subseteq W^\textsc{NR}$ such that $\mu^Z = \mu^{\textsc{NR}}$ and $|Z| < |Y|$. In other words, $Y$ is the smallest set of women that, when pushed up in $\>^m$, results in an optimal no-regret strategy for the accomplice. Similarly, we call $Y \subseteq W^\textsc{NR}$ a \emph{minimal} no-regret push up set if $\mu^Y = \mu^{\textsc{NR}}$ and there exists no other set $Z$ such that $\mu^Z = \mu^{\textsc{NR}}$ and $Z \subset Y$. 

A minimum no-regret push up set is clearly also minimal. Our next %
result %
(\Cref{thm:No-regret_Minimal_Is_Minimum}) establishes a converse.

\begin{restatable}{theorem}{NoRegretMinimalIsMinimum}
Let $Y \subseteq W^\textsc{NR}$ be a minimal no-regret push up set. Then, $Y$ is also a minimum no-regret push up set.
\label{thm:No-regret_Minimal_Is_Minimum}
\end{restatable}

\Cref{thm:No-regret_Minimal_Is_Minimum} naturally gives rise to a greedy algorithm for finding a minimum no-regret push up set (Algorithm \ref{alg:MinimizePushUpSet}). Starting with the no-regret set $W^\textsc{NR}$ as its current guess, the algorithm iteratively removes an agent from the current no-regret set and checks if the resulting matching changes. If the matching does not change, the guess is updated by removing the said agent and the check is repeated for the reduced set.

\begin{algorithm}[t]
\small
\caption{Computing a minimum no-regret push up set}
\label{alg:MinimizePushUpSet}
\begin{algorithmic}[1]
    \Require{Preference profile $\succ$, accomplice $m$}
    \Ensure{Minimum no-regret push up set $Y$}
    \State Compute $W^\textsc{NR} \gets \{w \in W \, : \, \>' \coloneqq \{ \>_{-m}, \>_m^{w \uparrow} \}$  is a no-regret profile\}
    \State Initialize $Y \gets W^\textsc{NR}$ and $\mu^* \gets \DA(\>_{-m}, \>_m^{Y \uparrow})$
    \For{each $w_y \in Y$\label{algline:forloop}}
        \State {$Y' \gets Y \setminus \{w_y\}$}
        \State {$\mu' \gets \DA(\>_{-m}, \>_m^{Y' \uparrow})$}
        \If{$\mu' = \mu^*$}
            \State {Restart for-loop in Line~\ref{algline:forloop} with $Y \gets Y'$}
        \EndIf
    \EndFor
    \State \Return {$Y$}
\end{algorithmic}
\end{algorithm}

\begin{restatable}{corollary}{MinNoRegretPushUpPolyTime}
A minimum optimal one-for-all strategy can be computed by Algorithm \ref{alg:MinimizePushUpSet} in $\O(n^4)$ time.
\label{cor:PolyAlgorithm_MinimumPushUpSet}
\end{restatable}

The proofs of \Cref{thm:No-regret_Minimal_Is_Minimum} and \Cref{cor:PolyAlgorithm_MinimumPushUpSet} are presented in the forthcoming sections.

\begin{remark}
In \Cref{sec:Tight_Bounds_Appendix}, we show that the size of a minimum no-regret push up set is always at most $\lfloor \frac{n-1}{2} \rfloor$~(\Cref{prop:No-regret_MinimumSet_UpperBound}), and that this bound is tight for a family of instances~(\Cref{eg:TightBound}).
\label{rem:Push_Up_Set_Bounds}
\end{remark}

Since the optimal strategy only involves a no-regret push up operation, \Cref{prop:PushUp} implies that the resulting matching is stable with respect to the true preferences.

\begin{restatable}{corollary}{MinNoRegretPushUpStabilityPreserving}
An optimal one-for-all strategy for the accomplice is stability-preserving.
\label{cor:Min_NoRegret_PushUp_StabilityPreserving}
\end{restatable}

\subsection{Proof of Theorem~\ref{thm:No-regret_Minimal_Is_Minimum}}

The proof of Theorem~\ref{thm:No-regret_Minimal_Is_Minimum} relies on a number of intermediate results (\Cref{lem:PairwiseDisjointSets,lem:SingleProposalSetCover,lem:PushUp_Meet}). Before introducing these results, let us define some notation.

Given any profile $\>$, let $P_{\>}$ denotes the set of all proposals that occur in the execution of the \DA{} algorithm on the profile $\>$. Formally, for any man $m_i \in M$ and woman $w_j \in W$, the ordered pair $(m_i,w_j)$ belongs to the set $P_{\>}$ if $m_i$ proposes to $w_j$ during the execution of \DA{} algorithm on the profile~$\>$. Additionally, let $P_{\> \setminus \>'} \coloneqq P_{\>} \setminus P_{\>'}$ denote the set of proposals that occur under the profile $\>$ but not under~$\>'$.

Suppose $\>$ is a true preference profile and $\mu \coloneqq \DA(\>)$. For a fixed accomplice $m$ and any woman $w_x$ such that $\mu(m) \>_m w_x$, we use the shorthand notation $\>^x_m$, $\>^x$, and $\mu^x$ to denote the list $\>_m^{w_x \uparrow}$, the profile $\>^{w_x \uparrow} \coloneqq \{ \>_{-m}, \>^{w_x \uparrow}_m \}$, and the matching $\mu^{w_x \uparrow} \coloneqq \DA(\>^{w_x \uparrow})$, respectively. Similarly, given a set of women $Y \subseteq W$, we use $\>^Y_m$, $\>^Y$, and $\mu^Y$ to denote the list $\>_m^{Y \uparrow}$, the profile $\>^{Y \uparrow} \coloneqq \{ \>_{-m}, \>^{Y \uparrow}_m \}$, and the matching $\mu^{Y \uparrow} \coloneqq \DA(\>^{Y \uparrow})$, respectively.

Given a woman $w_x$ and a corresponding profile $\>^x$, let $P'_x \coloneqq \{ (m_i, w_j) \in P_{\>^x \setminus \>} : m_i \neq  m \}$ denote the set of proposals that occur under $\>^x$ but not $\>$ made by any man other than the accomplice. Similarly, given a set of women $Y$ and a corresponding profile $\>^Y$, let $P'_Y \coloneqq \{ (m_i, w_j) \in P_{\>^Y \setminus \>} : m_i \neq  m \}$.

\begin{restatable}{lemma}{PushUpMeet}
Let $Y \coloneqq \{ w_{y_1}, w_{y_2}, w_{y_3}, \dots \}$ be such that $Y \subseteq W^\textsc{NR}$. Then $\mu^Y = \mu^{y_1} \wedge \mu^{y_2} \wedge \mu^{y_3} \wedge \dots$.
\label{lem:PushUp_Meet}
\end{restatable}

\begin{restatable}{lemma}{SingleProposalSetCover}
Let $Y \coloneqq \{ w_{y_1}, w_{y_2}, w_{y_3}, \dots \}$ be such that $Y \subseteq W^\textsc{NR}$. If for some $w_x \in W^\textsc{NR}$, $P'_x \subseteq P'_{y_1} \cup P'_{y_2} \cup P'_{y_3} \cup \dots$, then $P'_x \subseteq P'_{y_k}$ for some $w_{y_k} \in Y$.
\label{lem:SingleProposalSetCover}
\end{restatable}

\begin{restatable}{lemma}{PairwiseDisjointSets}
Let $Y \coloneqq \{ w_{y_1}, w_{y_2}, w_{y_3}, \hdots \}$ be such that $Y \subseteq W^\textsc{NR}$. If $Y$ is a minimal no-regret push up set, then $P'_{y_1}, P'_{y_2}, P'_{y_3}, \hdots$ are pairwise disjoint.
\label{lem:PairwiseDisjointSets}
\end{restatable}

The proofs of \Cref{lem:PairwiseDisjointSets,lem:SingleProposalSetCover,lem:PushUp_Meet} are presented in the following sections. Below, we will use these results to prove~\Cref{thm:No-regret_Minimal_Is_Minimum}.

\NoRegretMinimalIsMinimum*

\begin{proof}
Suppose, for contradiction, that there exists a minimum set $Z \subseteq W^\textsc{NR}$ such that $|Z| < |Y|$. Let $Y \coloneqq \{ w_{y_1}, w_{y_2}, w_{y_3}, \dots \}$ and $Z \coloneqq \{ w_{z_1}, w_{z_2}, w_{z_3}, \dots \}$.  By our definition of minimal and minimum sets, we have that $\mu^{W^\textsc{NR}} = \mu^Y = \mu^Z$.

Let $P' \coloneqq P'_{y_1} \cup P'_{y_2} \cup P'_{y_3} \cup \dots$. We will now argue that $P'_Y = P'$. \footnote{Note that this claim does not immediately follow from \Cref{prop:PushUpPropsalsContainment}, as $P'$ is the set of \emph{new} proposals in relation to the original preference profile that do not involve the accomplice.} Suppose there exists a proposal $(m', w') \in P'_Y \setminus P'$. By definition of sets $P'_Y$ and $P'_{y_k}$ for all $w_{y_k} \in Y$, we know that $m' \neq m$ (i.e., $m'$ is truthful). Since men propose in decreasing order of their preference and $\>^Y_{m'} \ = \ \>^{y_k}_{m'}$ for all $w_{y_k} \in Y$, we infer that  $\mu^Y(m') \notin \{ \mu^{y_1}(m'), \mu^{y_2}(m'), \mu^{y_3}(m'), \dots \}$, which poses a contradiction to \Cref{lem:PushUp_Meet}. 

Now, suppose there exists a proposal $(m', w') \in P' \setminus P'_Y$. Then, it must be that for some $w_{y_k} \in Y$, $\mu^Y(m')  \>_{m'} \mu^{y_k}(m')$, because men propose in decreasing order of their preference and all men other than the accomplice are truthful. However, this means that under the matching $\mu^Y$, man $m'$ is not matched to his least-preferred partner among $\mu^{y_1}(m'), \mu^{y_2}(m'), \mu^{y_3}(m'), \dots$, which contradicts~\Cref{lem:PushUp_Meet}.
Thus, we obtain the desired property, $P'_Y = P'$. Note that a similar line of reasoning may be used to claim that $P'_Z = P''$, where $P'' \coloneqq P'_{z_1} \cup P'_{z_2} \cup P'_{z_3} \cup \dots$.

We will now argue that $|Z| = |Y|$ to derive the desired contradiction. Notice that if $P'_Y \setminus P'_Z \neq \emptyset$, then $\mu^Y(m') \neq \mu^Z(m')$, because men propose in decreasing order of their preference and $\>^Y_{m'} \ = \ \>^Z_{m'}$. However, this poses a contradiction to our claim that $\mu^Y = \mu^Z$. Therefore, $P'_Y \setminus P'_Z = \emptyset$, and for the same reason, $P'_Z \setminus P'_Y = \emptyset$. This implies that $P'_Y = P'_Z$, and thus $P' = P''$. We infer from this result that $P'_{z_k} \subseteq P'$ for all $w_{z_k} \in Z$. \Cref{lem:SingleProposalSetCover} then shows that for all $w_{z_k} \in Z$, there exists some $w_{y_\ell} \in Y$ such that $P'_{z_k} \subseteq P'_{y_\ell}$. By a similar argument, for every $w_{y_\ell} \in Y$, there exists some $w_{z_k} \in Z$ such that $P'_{y_\ell} \subseteq P'_{z_k}$. Additionally, by \Cref{lem:PairwiseDisjointSets}, we know that $P'_{y_1}, P'_{y_2}, P'_{y_3}, \dots$ are pairwise disjoint sets, as are $P'_{z_1}, P'_{z_2}, P'_{z_3}, \dots$. Then, given $P' = P''$, it must be the case that $|Z| = |Y|$, giving us the desired contradiction. Hence, a minimal no-regret push up set must also be a minimum no-regret push up set.
\end{proof}

\subsection{Proof of Lemma~\ref{lem:PushUp_Meet}}

\PushUpMeet*
\begin{proof}
Suppose, for contradiction, that $\mu^Y(w) \neq \mu^{y_k}(w)$ for some $w \in W$, where $\mu^{y_k}(w) \coloneqq \mu^{y_1}(w) \wedge \mu^{y_2}(w) \wedge \mu^{y_3}(w) \wedge \hdots$, i.e., $\mu^{y_k}(w)$ is $w$'s favorite partner in $\{ \mu^{y_1}(w), \mu^{y_2}(w), \mu^{y_3}(w), \hdots \}$. 

From \Cref{lem:NoRegretSubset}, we know that $m$ matches with $\mu(m)$ under $\>^Y$. Since $w_{y_k} \in Y \subseteq W^\textsc{NR}$, we know that $m$ matches with $\mu(m)$ under $\>^{y_k}$ as well. Since, $\>^Y$ is obtained via a no-regret push up of set $Y \setminus \{ w_{y_k} \}$ in $\>^{y_k}_m$, from \Cref{prop:PushUp}, we infer that $\mu^Y(w) \succeq_w \mu^{y_k}(w)$. Given our assumption that $\mu^Y(w) \neq \mu^{y_k}(w)$, we get that $\mu^Y(w) \>_w \mu^{y_k}(w)$. Due to this, the woman $w$ must receive a proposal under $\>^Y$ that she does not under $\>^{y_1},\>^{y_2},\dots,\>^{y_k}$. However, this poses a contradiction to \Cref{prop:PushUpPropsalsContainment}, which guarantees that any proposal under $\>^Y$ must occur under at least one of the profiles $\>^{y_1},\>^{y_2},\dots,\>^{y_k}$.
\end{proof}

\subsection{Proof of Lemma~\ref{lem:SingleProposalSetCover}}

Given two women $w_x, w_y \in W^\textsc{NR}$ and some man $m_x$ such that $m_x \>_{w_x} m$,  %
we show that if $m_x$ proposes to $w_x$ under the profile $\>^y$, then $P'_x$ is contained in $P'_y$~(\Cref{lem:Proposals_Set_Subset}). This result is leveraged in the proof of \Cref{lem:SingleProposalSetCover}. %

\begin{restatable}{lemma}{ProposalsSetSubset}
Let $w_x, w_y \in W^\textsc{NR}$. If there is a proposal $(m_x, w_x) \in P_{\>^y}$ such that $m_x \>_{w_x} m$,  
then $P'_x \subseteq P'_y$.
\label{lem:Proposals_Set_Subset}
\end{restatable}

\begin{proof}
Let $\>^{xy} \coloneqq \{\>_{-m},\>_m^{\{w_x,w_y\} \uparrow}\}$ denote the profile obtained by pushing up both $w_x$ and $w_y$ in the true list $\>_m$ of the accomplice. From \Cref{lem:NoRegretSubset}, we get that $\>^{xy}$ is a no-regret profile, and can be obtained via a no-regret push up of $w_x$ in $\>^y_m$. 

Given $(m_x, w_x) \in P_{\>^y}$ where $m_x \>_{w_x} m$, we also have that $\mu^y(w_x) \>_{w_x} m$, because women match with their favorite proposers and it is assumed that $\>^y_{w_x} = \ \>_{w_x}$ (since women report truthfully). Then, by \Cref{prop:WeakPushUp}, we have that $\mu^{xy} = \mu^y$. Furthermore, since men propose in decreasing order of their lists and all men except for the accomplice are truthful, we get that $P'_{xy} \setminus P'_y = \emptyset$ and $P'_y \setminus P'_{xy} = \emptyset$, and thus $P'_{xy} = P'_y$.

It follows from \Cref{prop:PushUpPropsalsContainment} that $P_{\>^{xy}} \subseteq P_{\>^x} \cup P_{\>^y}$. Combining this with \Cref{lem:ProposalSetEquivalence}, we get
\[P'_{xy} \cup P_{\>} \cup \{(m, w_x), (m, w_y) \} \subseteq \]
\[ P'_x \cup P'_y \cup P_{\>} \cup \{(m, w_x), (m, w_y)\}.\]
Since the sets on each side of the above containment relation are disjoint, we get that $P'_{xy} \subseteq P'_x \cup P'_y$. Furthermore, we already have that $P'_{xy} = P'_y$, which means that $P'_{xy} \subseteq P'_x \cup P'_{xy}$. Therefore, it must be that $P'_x \subseteq P'_{xy}$. Since $P'_{xy} = P'_y$, we obtain $P'_x \subseteq P'_{y}$, as desired.
\end{proof}

\SingleProposalSetCover*
\begin{proof}
Since $\>^x$ is a no-regret profile, the accomplice $m$ proposes to and is rejected by the woman $w_x$. Thus, there must be a proposal $(m_x, w_x) \in P_{\>^x}$ such that $m_x \>_{w_x} m$.

Suppose $\mu(w_x) \>_{w_x} m$. Then, by \Cref{prop:WeakPushUp}, $\mu^x = \mu$. Since men propose in decreasing order of their preference and all men other than the accomplice are truthful, we have that $P'_x = \emptyset$ and the lemma trivially holds.

Otherwise, we must have that $m \succeq_{w_x} \mu(w_x)$. Combining this with $m_x \>_{w_x} m$, we get $m_x \>_{w_x} \mu(w_x)$. Given this and $m_x \neq m$, we have that $(m_x, w_x) \in P'_x$. Due to our assumption that $P'_x \subseteq P'_{y_1} \cup P'_{y_2} \cup P'_{y_3} \cup \dots$, we infer $(m_x, w_x) \in P'_{y_k}$ for some $w_{y_k} \in Y$. Then, \Cref{lem:Proposals_Set_Subset} 
shows that $P'_x \subseteq P'_{y_k}$, thus completing the proof. 
\end{proof}

\subsection{Proof of Lemma~\ref{lem:PairwiseDisjointSets}}
\label{sec:Proof_Of_PairwiseDisjointSets}

The proof of \Cref{lem:PairwiseDisjointSets} uses a number of intermediate results (\Cref{lem:ProposalSetEquivalence,lem:chain_equivalence,lem:Proposals_Set_Subset,lem:Proposals_Subset_or_Disjoint}). We will start by showing that the set of proposals under the profile $\>^Y$, namely $P_{\>^Y}$, is equivalent to the union of $P'_Y$, $P_{\>}$, and the set of proposals $m$ makes to the women in $Y$ (\Cref{lem:ProposalSetEquivalence}).

\begin{restatable}{lemma}{ProposalSetEquivalence}
Let $Y \subseteq W^\textsc{NR}$. Then, $P_{\>^Y} = P'_Y \cup P_{\>} \cup \{(m, w_j) : w_j \in Y\}$.
\label{lem:ProposalSetEquivalence}
\end{restatable}

\begin{proof}
By definition of $P'_Y$, we know that $P'_Y = P_{\>^Y} \setminus (P_{\>} \cup \{(m, w_j) : w_j \in Y\})$ which implies that $P_{\>^Y} \subseteq P'_Y \cup P_{\>} \cup \{(m, w_j) : w_j \in Y\}$. It is also easy to see that $P'_Y \cup \{(m, w_j) : w_j \in Y\} \subseteq P_{\>^Y}$. From \Cref{lem:NoRegretSubset}, we know that $\>^Y$ is a no-regret profile, and thus, by \Cref{prop:PushUpProposals}, $P_{\>} \subseteq P_{\>^Y}$. Combining this with the above observation, we get $P'_Y \cup P_{\>} \cup \{(m, w_j) : w_j \in Y\} \subseteq P_{\>^Y}$. Since we have shown containment in both directions, the two sets must be the same, i.e., $P_{\>^Y} = P'_Y \cup P_{\>} \cup \{(m, w_j) : w_j \in Y\}$.
\end{proof}

Before stating the next intermediate result (\Cref{lem:chain_equivalence}), it would be helpful to define a combinatorial structure that we will refer to as ``chain of proposals.'' We note that \citet{A15susceptibility} study a similar construct which they call ``chains of rejections'' in the context of a many-to-one matching model.

\paragraph{Chains of proposals} Given an arbitrary woman $w_x \in W^\textsc{NR}$, let $\>^x$ be the corresponding profile where $m$ pushes up $w_x$ in his preference list without incurring regret. Additionally, let $\sigma_{\>^x}$ be a permutation of the sets of proposals that occur under profile $\>^x$. It is guaranteed by \Cref{prop:PushUpProposals} that $P_{\>} \subseteq P_{\>^x}$. Due to this, we suppose all proposals in $P_{\>}$ are at the top of $\sigma_{\>^x}$, followed by $p^0_x \coloneqq (m, w_x)$. We define $C_x$, consisting of the proposals $p_x^1, p_x^2, \dots, p_x^k$ (in that order), as a \emph{chain of proposals} that follows $p^0_x$ in $\sigma_{\>^x}$. 

Let $p_x^i$ be the $i$-th proposal in the chain, and let $p_x^{i-1} \coloneqq (m_x^{i-1}, w_x^{i-1})$ be the proposal that comes directly before $p_x^i$ in $\sigma_{\>^{x}}$. Since $P_{\>}$ yields a complete matching and we have assumed that all proposals in $P_{\>}$ take place before $C_x$ according to $\sigma_{\>^{x}}$, we infer that $w_x^{i-1}$ must be tentatively matched to some man, say $m'_x$, prior to $p_x^{i-1}$.
In other words, proposal $(m'_x, w_x^{i-1})$ is above $p_x^{i-1}$ in $\sigma_{\>^{x}}$, and $m'_x$ is $w_x^{i-1}$'s favorite proposer before $p_x^{i-1}$. Depending on who $w_x^{i-1}$ prefers, there are two cases we must consider.

\textbf{\underline{Case A}} (when $m'_x \>_{w_x^{i-1}} m_x^{i-1}$): Then, $w_x^{i-1}$ rejects $m_x^{i-1}$ in favor of $m'_x$, and we define $p_x^i$ to be $m_x^{i-1}$'s proposal to the next woman in his preference list, say $w_x^{i}$ (i.e., $p_x^i \coloneqq (m_x^{i-1}, w_x^i)$).

\textbf{\underline{Case B}} (when $m_x^{i-1} \>_{w_x^{i-1}} m'_x$): Then, $w_x^{i-1}$ rejects $m'_x$ in favor of $m_x^{i-1}$, and we define $p_x^i$ to be $m'_x$'s proposal to the next woman in his preference list, say $w'_{i}$ (i.e., $p_x^i \coloneqq (m'_x, w'_{i})$).

Regardless of who $w_x^{i-1}$ prefers, we later show in \Cref{lem:chain_equivalence} that $p_x^i \in P'_x$, and thus $p_x^i$ does indeed occur during the \DA{} execution on $\>^x$.

We define $p_x^k$ to be the first proposal after $(m, w_x)$ in $\sigma_{\>^{x}}$ such that some man $m_x \>_{w_x} m$ proposes to $w_x$. As stated above, $p_x^k$ is the last proposal in $C_x$. We use \Cref{eg:chains-of-proposals} to help illustrate the structure of these chains.
  
\begin{example}[\textbf{Chains of proposals}]\label{eg:chains-of-proposals}
Consider the following preference profile where the \DA{} outcome is underlined. 

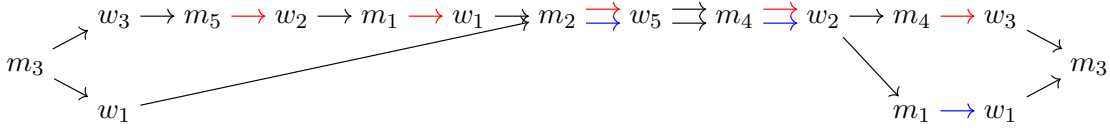
\begin{figure*}[t]
    \centering
    \begin{tikzcd}[cramped, row sep=tiny, column sep=small]
        & 
        w_3 \arrow[r] & 
        m_5 \arrow[r, red] & 
        w_2 \arrow[r] & 
        |[alias=m1]| m_1 \arrow[r, red] & 
        w_1 \arrow[r] & 
        |[alias=m2]| m_2 \arrow[r, shift left, red] \arrow[r, shift right, blue] & 
        w_5 \arrow[r, shift left] \arrow[r, shift right] & 
        m_4 \arrow[r, shift left, red] \arrow[r, shift right, blue] & 
        w_2 \arrow[r] \arrow[to=m11] & 
        m_4 \arrow[r, red] &
        w_3 \arrow[dr] & & [1.5em]\\
        m_3 \arrow[ur] \arrow[dr] & & & & & & & & & & & &
        |[alias=m3]| m_3\\
        & w_1 \arrow[to=m2] & & & & & & & & & 
        |[alias=m11]| m_1 \arrow[r, blue] &
        w_1 \arrow[to=m3]
    \end{tikzcd}
    \caption{Illustration of the sequence of events surrounding the chains of proposals, $C_1$ and $C_3$, as described in \Cref{eg:chains-of-proposals}. An edge from $m \in M$ to $w \in W$ represents $m$ proposing to $w$; an edge from $w \in W$ to $m \in M$ represents $w$ rejecting $m$. The bottom sequence represents the chain $C_1$, and the blue edges indicate proposals that belong to $C_1$. Similarly, the top sequence represents the chain $C_3$, and the red edges indicate proposals that belong to $C_3$ (e.g., $(m_5, w_2) \in C_3$).}
    \label{fig:chains-of-proposals}
\end{figure*}

\begin{table}[H]
    \centering
    \begin{tabularx}{0.7\linewidth}{XXXXXXXXXXXXXXX}
            ${m_1}\colon$ & $w_5$ & $\underline{w_2}$ & $w_1^{*,\dagger}$ & $w_4$ & $w_3$ && ${w_1}\colon$ & $m_1^{*,\dagger}$ & $m_5$ & $m_3$ & $m_4$ & $\underline{m_2}$\\
            ${m_2}\colon$ & $\underline{w_1}$ & $w_5^{*,\dagger}$ & $w_3$ & $w_2$ & $w_4$ && ${w_2}\colon$ & $m_5^*$ & $m_4^{\dagger}$ & $\underline{m_1}$ & $m_2$ & $m_3$\\
            $\boldsymbol{\textcolor{blue}{{m_3}}}\colon$ & $w_2$ & $\underline{w_4^{*}}^{,\dagger}$ & $w_3$ & $w_1$ & $w_5$ && ${w_3}\colon$ & $m_4^*$ & $m_3$ & $m_1$ & $m_2$ & $\underline{m_5}$\\
            ${m_4}\colon$ & $\underline{w_5}$ & $w_2^{\dagger}$ & $w_3^*$ & $w_1$ & $w_4$ && ${w_4}\colon$ & $m_5$ & $m_4$ & $\underline{m_3^{*}}^{,\dagger}$ & $m_1$ & $m_2$\\
            ${m_5}\colon$ & $\underline{w_3^{\dagger}}$ & $w_2^*$ & $w_4$ & $w_1$ & $w_5$ && ${w_5}\colon$ & $m_3$ & $m_2^{*,\dagger}$ & $\underline{m_4}$ & $m_1$ & $m_5$
        \end{tabularx}
\end{table}

Suppose the accomplice is $m_3$. The \DA{} matchings after $m_3$ submits the lists $\succ^1_{m_3} = w_1 \> w_2 \> w_4 \> w_3 \> w_5$ and $\>^3_{m_3} \coloneqq w_3 \> w_2 \> w_4 \> w_1 \> w_5$ are marked by $*$ and $\dagger$, respectively. Notice that $\>^1_{m_3}$ and $\>^3_{m_3}$ are derived from $\>_{m_3}$ by pushing up $w_1$ and $w_3$, respectively. 

Under both $\>^1 \coloneqq \{ \>_{-m_3}, \>^1_{m_3} \}$ and $\>^3 \coloneqq \{ \>_{-m_3}, \>^3_{m_3} \}$, there are corresponding chains of proposals, $C_1$ and $C_3$. Figure \ref{fig:chains-of-proposals} illustrates the sequence of events that surround these two chains.

Let us walk through chain $C_1$ as an example. Recall that the set $P_{\>}$ given by
$$P_{\>} = \{(m_1, w_5), (m_2, w_1), (m_3, w_2), (m_4, w_5),\\(m_5, w_3), (m_1, w_2), (m_3, w_4) \}$$ and proposal $(m_3, w_1)$ occur prior to the chain. Upon receiving $m_3$'s proposal, $w_1$ rejects $m_2$, who then proposes to the next woman in his preference list, $w_5$. Thus, $(m_2, w_5)$ is the first proposal in $C_1$. Because $w_5$ prefers $m_2$ over her original \DA{} match, $m_4$, she rejects $m_4$ who, as a result, must propose to the next woman in his preference list, $w_2$. Then, $(m_4, w_2)$ is the second proposal in $C_1$. Next, $w_2$ rejects $m_1$ in favor of $m_4$, and $m_1$ proposes to $w_1$. Since $(m_1, w_1)$ is the first proposal such that some man above $m_3$, the accomplice, in $w_1$'s preference list proposes to $w_1$. Thus, $(m_1, w_1)$ is the third and last proposal in $C_1$.\qed
\end{example}

Now that we have clearly defined what a chain of proposals $C_x$ is, we show that the set of proposals belonging to the chain is equivalent to the set $P'_x$ (\Cref{lem:chain_equivalence}).

\begin{lemma}
The set of proposals in $C_x$ is equivalent to the set $P'_x$.
\label{lem:chain_equivalence}
\end{lemma}

\begin{proof}
We first show that $P'_x$ is contained in $\{p_x^1, p_x^2, \dots, p_x^k\}$, which is the set of proposals that belong to $C'_x$. Suppose, for contradiction, that $p \coloneqq (m', w')$ is the \emph{first} proposal during the \DA{} execution on $\>^x$ that belongs to $P'_x \setminus \{p_x^1, p_x^2, \dots, p_x^k\}$. Note that for this proof we are not assuming the order defined by $\sigma_{\>^x}$ and are instead assuming that proposals occur in the order in which the \DA{} algorithm traditionally executes.

Since men propose in decreasing order of their preference and it is assumed that $p \in P_{\>^x}$, it must be that $m'$ was rejected by some woman, say $w''$, under $\>^x$ before proposing to $w'$. Notice that $w'$ cannot be the first woman in $m'$'s list, otherwise it must be that $p \in P_{\>} \cup \{(m, w_x)\}$, which would contradict our assumption that $p \in P'_x$.
Therefore, without loss of generality, we assume $w''$ is the woman directly above $w'$ in $\>^x_{m'}$. It must also be that $w''$ receives a proposal from some other man, say $m''$, such that $m'' \>_{w''} m'$. Clearly, $p' \coloneqq (m', w'')$ and $p'' \coloneqq (m'', w'')$ both occur before $p$ during the \DA{} execution on $\>^x$.

Because we assumed $p$ to be the first proposal during the $\DA$ execution on $\>^x$ to belong to $P'_x \setminus \{p_x^1, p_x^2, \dots, p_x^k\}$, we must have that $p' \notin P'_x \setminus \{p_x^1, p_x^2, \dots, p_x^k\}$. This implies that $p' \notin P'_x$ or $p' \in \{p_x^1, p_x^2, \dots, p_x^k\}$ (or both). The same can be said for proposal $p''$. 
Given these possibilities, we use case analysis to proceed with the proof.\\

\textbf{\underline{Case I}} (when $p' \notin P'_x$ and $p'' \notin P'_x$): Then, $p', p'' \in P_{\>} \cup \{(m, w_x)\}$ due to \Cref{lem:ProposalSetEquivalence}. If $p' = (m, w_x)$, then $p = (m, \mu(m))$, because the woman directly below $w_x$ in $\>^x_m$ is, without loss of generality, assumed to be $\mu(m)$. Recall that due to \Cref{prop:Permuting_Falsified_Lists}, we are able to assume any order of the women above $\mu(m)$ in $\>^x_m$. However, this would mean that $p \in P_{\>}$, posing a contradiction to our assumption that $p \in P'_x$. Therefore, we have that $p' \neq (m, w_x)$, and thus $p' \in P_{\>}$.

Now, we show that $p'' \neq (m, w_x)$. Suppose the contrary (i.e., $p'' = (m, w_x)$) is true. Because $m' \neq m''$, it must be that $m' \neq m$, and thus $p' \in P_{\>}$ due to our previous claim that $p', p'' \in P_{\>} \cup \{(m, w_x)\}$. If $m' = \mu(w_x)$ (i.e., $p' = (\mu(w_x), w_x)$), then by definition of $C_x$, we have that $p = (\mu(w_x), w') = p_x^1$, where $w'$ is the woman directly below $w_x$ in $\mu(w_x)$'s preference list. However, this poses a contradiction to our assumption that $p \in P'_x \setminus \{p_x^1, p_x^2, \dots, p_x^k\}$. If $m' \neq \mu(w_x)$, on the other hand, then it must be that $w'$ is below $w_x$ in $\>_{m'}^x$, because $p' \in P_{\>}$ and men propose in decreasing order of their preference under the \DA{} algorithm. This would imply that $p \in P_{\>}$, and therefore $p$ does not belong to the set $P'_x$, which again poses a contradiction. Thus, we have shown that $p'' \neq (m, w_x)$, which implies that $p'' \in P_{\>}$. 

Given that $p', p'' \in P_{\>}$, it can be shown that $p \in P_{\>}$. Indeed, if both $m'$ and $m''$ propose to $w''$ under $\>$, then $w''$ would reject $m'$ in favor of $m''$ under $\>$, because $m'' \>_{w''} m'$. In response, $m'$ would have to propose to the next woman in $\>_{m'}$, say $\bar{w}$. If $m'$ is truthful (i.e., $\>^x_{m'} \ = \ \>_{m'}$), then clearly $\bar{w} = w'$, and thus $p \in P_{\>}$. On the other hand, if $m'$ is the accomplice and $\bar{w} \neq w'$, then it must be that $w' = w_x$ or $w'' = w_x$, because we assume that the relative ordering of all women in $W \setminus \{w_x\}$ is the same in $\>^x_m$ and $\>_m$. Above, we showed that $p' \neq (m, w_x)$, which means $w'' \neq w'$. If $w' = w_x$, then it is easy to see that $p = (m, w_x)$, which would contradict our assumption that $p \in P'_x$. Thus, $\bar{w} = w'$ even if $m'$ is not truthful, which indicates that $p = (m', w')$ occurs under $P_{\>}$. However, our assumption was that $p' \notin P'_x$, giving us a contradiction, as desired.\\

\textbf{\underline{Case II}} (when $p' \notin P'_x$ and $p'' \in \{p_x^1, p_x^2, \dots, p_x^k\}$): Then, $p' \in P_{\>} \cup \{(m, w_x)\}$ due to \Cref{lem:ProposalSetEquivalence}. The proposals in $P_{\>} \cup \{(m, w_x)\}$ are above the set $\{p_x^1, p_x^2, \dots, p_x^k\}$ in the ordered list $\sigma_{\>^x}$, as previously defined. Therefore, $p'$ is above $p''$ in $\sigma_{\>^x}$. Without loss of generality, we say that $m'$ is the man who $w''$ is (temporarily) matched with when $m''$ proposes under $\sigma_{\>^x}$. 
Additionally, note that $p'' \neq p^k_x$, otherwise it can be shown that $p' = (m', w_x)$. It follows from the same arguments for Case I that $p' \neq (m', w_x)$. 
Then, according to Case B in our definition of the chain $C_x$, proposal $p$ must come directly after $p''$ in $C_x$.
However, this poses a contradiction, because we assumed $p \in P'_x \setminus \{p_x^1, p_x^2, \dots, p_x^k\}$.\\

\textbf{\underline{Case III}} (when $p' \in \{p_x^1, p_x^2, \dots, p_x^k\}$ and $p'' \notin P'_x$): Then, $p'' \in P_{\>} \cup \{(m, w_x)\}$ due to \Cref{lem:ProposalSetEquivalence}. The proposals in $P_{\>} \cup \{(m, w_x)\}$ are above the set $\{p_x^1, p_x^2, \dots, p_x^k\}$ in the ordered list $\sigma_{\>^x}$, as previously defined. Therefore, $p''$ is above $p'$ in $\sigma_{\>^x}$. Without loss of generality, we say that $m''$ is the man who $w''$ is (temporarily) matched with when $m'$ proposes under $\sigma_{\>^x}$. 
Additionally, note that $p' \neq p^k_x$, otherwise it can be shown that $p' = (m, \mu(m))$, which is not possible given our assumption that $p \in P_{\>}$.
Then, according to Case A in our definition of the chain $C_x$, proposal $p$ must come directly after $p'$ in $C_x$.
However, this poses a contradiction, because we assumed $p \in P'_x \setminus \{p_x^1, p_x^2, \dots, p_x^k\}$.\\

\textbf{\underline{Case IV}} (when $p' \in \{p_x^1, p_x^2, \dots, p_x^k\}$ and $p'' \in \{p_x^1, p_x^2, \dots, p_x^k\}$): If $p'$ is above $p''$ (respectively, $p''$ is above $p'$) in $\sigma_{\>^x}$, then the same arguments from Case II (respectively, Case III) imply that proposal $p$ belongs to the set $\{p_x^1, p_x^2, \dots, p_x^k\}$, posing a contradiction, as desired.\\

The four cases above cover all possible situations. Therefore, we have shown that $P'_x$ is contained in $\{p_x^1, p_x^2, \dots, p_x^k\}$. Now, we must show that the reverse is also true (i.e., $\{p_x^1, p_x^2, \dots, p_x^k\} \subseteq P'_x$). From \Cref{lem:ProposalSetEquivalence}, we know that $P_{\>^x} = P'_x \cup P_{\>} \cup \{(m, w_x)\}$. If $\{p_x^1, p_x^2, \dots, p_x^k\} \setminus P'_x \neq \emptyset$, then there must be a proposal in $\{p_x^1, p_x^2, \dots, p_x^k\}$ that does not belong to $P_{\>^x}$, since $P'_x \subseteq \{p_x^1, p_x^2, \dots, p_x^k\}$. Without loss of generality, suppose, for contradiction, that $p_x^i \coloneqq (m_x^i, w_x^i)$ is the \emph{first} proposal in $\{p_x^1, p_x^2, \dots, p_x^k\}$ to not belong to $P_{\>^x}$. Recall that $P_{\>}$ and $(m, w_x)$ are the only proposals before the set $\{p_x^1, p_x^2, \dots, p_x^k\}$ in $\sigma_{\>^x}$. Because $P_{\>} \cup \{(m, w_x)\} \subseteq P_{\>^x}$, it must be that $p_x^i$ is the first proposal in $\sigma_{\>^x}$ to not occur under $\>^x$.

Let $p^{i-1}_x \coloneqq (m_x^{i-1}, w_x^{i-1})$ be the proposal directly above $p^i_x$ in $\sigma_{\>^x}$. Note that it must be that $w_x^i$ is the woman directly below $w_x^{i-1}$ in $m_x^{i-1}$'s preference list given our definition of $\sigma_{\>^x}$ and $C_x$. Additionally, let proposal $\bar{p} \coloneqq (\bar{m}, w_x^{i-1})$ be such that $\bar{m}$ is the man who $w_x^{i-1}$ is (temporarily) matched with when $m_x^{i-1}$ proposes under $\sigma_{\>^x}$. If $m_x^{i-1} = m_x^i$, then, according to Case A in our definition of the chain $C_x$, it must be that $\bar{m} \>_{w_x^{i-1}} m_x^{i-1}$. Since $\bar{p}$ and $p^{i-1}_x$ are both above $p^{i}_x$ in $\sigma_{\>^x}$, they must belong to $P_{\>^x}$. However, this would imply that $w_x^{i-1}$ rejects $m_x^{i-1} = m_x^i$ in favor of $\bar{m}$ under $\>^x$, and thus $m_x^{i-1}$ would have to propose to $w_x^{i}$, giving us a contradiction, as desired.
If $m_x^{i-1} \neq m_x^i$, on the other hand, then, according to Case B in our definition of the chain $C_x$, it must be that $m_x^{i-1} \>_{w_x^{i-1}} \bar{m}$. Since $\bar{p}$ and $p^{i-1}_x$ are both above $p^{i}_x$ in $\sigma_{\>^x}$, they must belong to $P_{\>^x}$. However, this would imply that $w_x^{i-1}$ rejects $\bar{m} = m_x^{i}$ in favor of $m_x^{i-1}$ under $\>^x$, and thus $m_x^{i-1}$ would have to propose to $w_x^{i}$, again giving us a contradiction, as desired. Thus, we have shown that the set $\{p_x^1, p_x^2, \dots, p_x^k\}$ is contained in $P'_x$. 

Now that we have shown containment of $P'_x$ and $\{p_x^1, p_x^2, \dots, p_x^k\}$ in both directions, our proof is complete.
\end{proof}

We use the previous result to prove that for any pair of women, $w_x$ and $w_y$, who belong to $W^\textsc{NR}$, it must be that $P'_x$ and $P'_y$ are disjoint, or one is a subset of the other (\Cref{lem:Proposals_Subset_or_Disjoint}). More specifically, we show that if $P'_x$ and $P'_y$ are not disjoint (i.e., $P'_x \cap P'_y \neq \emptyset$), then there is either a proposal $(m_x, w_x) \in P_{\>^y}$ such that $m_x \>_{w_x} m$ or there is a proposal $(m_y, w_y) \in P_{\>^x}$ such that $m_y \>_{w_y} m$, which would indicate that $P'_x \subseteq P'_y$ or $P'_y \subseteq P'_x$.

\begin{restatable}{lemma}{ProposalsSubsetOrDisjoint}
For any pair $w_x, w_y \in W^\textsc{NR}$, at least one of the following must be true: 1) $P'_x \subseteq P'_y$, 2) $P'_y \subseteq P'_x$, or 3) $P'_x \cap P'_y = \emptyset$.
 \label{lem:Proposals_Subset_or_Disjoint}
\end{restatable}

\begin{proof}
Suppose, for contradiction, that $P'_x \setminus P'_y \neq \emptyset$, $P'_y \setminus P'_x \neq \emptyset$, and $P'_x \cap P'_y \neq \emptyset$.
Let $\sigma_{\>^x}$ and $\sigma_{\>^y}$ be permutations of the sets of proposals that occur under profiles $\>^x$ and $\>^y$, respectively, and let $C_x$ and $C_y$ be the respective chains of proposals, as defined earlier.

Our previous result (\Cref{lem:chain_equivalence}), combined with our contradiction assumption that $P'_x \cap P'_y \neq \emptyset$, implies there is at least some proposal, say $p' \coloneqq (m', w')$, in both $C_x$ and $C_y$ that belongs to $P'_x \cap P'_y$. 
Note that $m' \neq m$ (i.e., $m'$ is truthful), because $m$ does not make any proposals during $C_x$ and $C_y$. Indeed, all of his proposals are above the chains $C_x$ and $C_y$ in $\sigma_{\>^x}$ and $\sigma_{\>^y}$, respectively.

Since it is assumed that the set $P_{\>}$ occurs before the chains $C_x$ and $C_y$ and we know that $P_{\>}$ yields a complete matching, it must be that $w'$ has a proposal from some man, say $m''$ above $p'$ in both $\sigma_{\>^x}$ and $\sigma_{\>^y}$. Without loss of generality, suppose $m''$ is $w'$'s favorite proposer such that $(m'', w')$ is above $p'$ in both chains $C_x$ and $C_y$.

If $m'' \>_{w'} m'$, then $w'$ must reject $m'$ under both chains. Since $m' \neq m$, we know that $m'$'s preference list is the same under both instances. The woman he proposes to after $w'$, say $\bar{w}$, is therefore also the same. Thus, the proposal $(m', w'')$ also occurs under both $C_x$ and $C_y$. 

On the other hand, if $m' \succ_{w'} m''$, then $m'$ must reject $m''$ under both chains. If $m'' = m$, then it must be that $w' \in \{w_x, w_y\}$, and by \Cref{lem:Proposals_Set_Subset}, $P'_x \subseteq P'_y$ or $P'_y \subseteq P'_x$. However, this would contradict our assumption that $P'_x \setminus P'_y \neq \emptyset$ and $P'_y \setminus P'_x \neq \emptyset$. Therefore, $m''$ is truthful; his preference list is the same under both instances and the woman he proposes to after $w'$, say $w''$, is also the same. Thus, the proposal $(m'', w'')$ also occurs under both $C_x$ and $C_y$. 

For both cases, we have identified another proposal that occurs under both $C_x$ and $C_y$. We can repeat our analysis for this new proposal and so on. It follows that after each proposal, there must exist another unique proposal that occurs under both $C_x$ and $C_y$. However, we know that there are only a finite number of proposals that can occur under $C_x$ and $C_y$, thus posing a contradiction.
\end{proof}

We are now ready to prove \Cref{lem:PairwiseDisjointSets}.

\PairwiseDisjointSets*

\begin{proof}
Suppose, for contradiction, that $P'_{y_i} \cap P'_{y_j} \neq \emptyset$ for some pair $w_{y_i}, w_{y_j} \in Y$. Then, by \Cref{lem:Proposals_Subset_or_Disjoint}, either $P'_{y_i} \subseteq P'_{y_j}$ or $P'_{y_j} \subseteq P'_{y_i}$. Without loss of generality, suppose $P'_{y_i} \subseteq P'_{y_j}$. Let $P' \coloneqq P'_{y_1} \cup P'_{y_2} \cup P'_{y_3} \cup \dots$ and $P'' \coloneqq P' \setminus P'_{y_i}$. Then, it is clear that $P'' = P'$.

The profile $\>^Y$ can be created through no-regret push up operations on $\>^{y_1}$, $\>^{y_2}$, $\>^{y_3}, \dots$. Thus, \Cref{prop:PushUpProposals} shows that $P_{\>^{y_1}} \cup P_{\>^{y_2}} \cup P_{\>^{y_3}} \cup \dots \subseteq P_{\>^Y}$, and \Cref{prop:PushUpPropsalsContainment} implies that $P_{\>^Y} \subseteq P_{\>^{y_1}} \cup P_{\>^{y_2}} \cup P_{\>^{y_3}} \cup \dots$. From these observations, we observe that $P_{\>^Y} = P_{\>^{y_1}} \cup P_{\>^{y_2}} \cup P_{\>^{y_3}} \cup \dots$. Further, from \Cref{lem:ProposalSetEquivalence}, we have that $P_{\>^{y_k}} = P'_{y_k} \cup P_{\>} \cup \{(m,w_{y_k})\}$ for all $w_{y_k} \in Y$. This implies that $P_{\>^Y} = P' \cup P_{\>} \cup \{ (m, w_j) : w_j \in Y \}$. Further, if we denote $Y' \coloneqq Y \setminus \{w_{y_i}\}$, then a similar line of reasoning shows that $P_{\>^{Y'}} = P'' \cup P_{\>} \cup \{ (m, w_j) : w_j \in Y' \}$. Since $P' = P''$ and $Y' = Y \setminus \{w_{y_i}\}$, we infer that $P_{\>^Y \setminus \>^{Y'}} = \{(m, w_{y_i})\}$ and $P_{\>^{Y'} \setminus \>^Y} = \emptyset$.

Thus, the only difference between the proposal structure under $\>^Y$ and $\>^{Y'}$ is the proposal $(m, w_{y_i})$ by the accomplice. From \Cref{lem:NoRegretSubset}, we know that $m$ matches with $\mu(m)$ under $\>^Y$ and $\>^{Y'}$, which implies that no other agent's match is affected by this additional proposal, and therefore $\mu^Y = \mu^{Y'}$. However, this poses a contradiction to our assumption that $Y$ is minimal. Thus, the sets $P'_{y_1}, P'_{y_2}, P'_{y_3}, \dots$ must be pairwise disjoint.
\end{proof}

\subsection{Proof of Corollary~\ref{cor:PolyAlgorithm_MinimumPushUpSet}}

\MinNoRegretPushUpPolyTime*
\begin{proof}
(sketch) Recall that $W^\textsc{NR} \coloneqq \{w \in W \, : \, \>' \coloneqq \{ \>_{-m}, \>_m^{w \uparrow} \} \text{ is a no-regret profile}\}$ is the set of all women who do not cause $m$ to incur regret when pushed up individually. 
By \cref{thm:OptimalStrategy_MultipleWomen}, the optimal one-for-all strategy of the accomplice can be computed by promoting $W^\textsc{NR}$, which can be computed in $\O(n^3)$ according to \cref{cor:PolyAlgorithm_OptimalStrategy}.
It follows from \Cref{thm:No-regret_Minimal_Is_Minimum} that we can compute a minimum no-regret push up set by finding a minimal push up set.
Algorithm \ref{alg:MinimizePushUpSet} starts with the minimal no-regret set $W^\textsc{NR}$ as its current guess. In each iteration, it removes an agent from the current no-regret set and checks if the resulting matching changes. If the matching does not change, the guess is updated by removing the said agent and the check is repeated for the reduced set.
The maximum number of iterations is $\O(n)$ and the number of checks in each iteration (i.e., the size of the set $W^\textsc{NR}$) is also at most $\O(n)$.
For each check, running the \DA{} algorithm takes $\O(n^2)$ time, resulting in the total running time of $\O(n^4)$.
\end{proof}

\subsection{Tight Bound on the Size of Minimum No-Regret Push Up Sets} %
\label{sec:Tight_Bounds_Appendix}

In this section, we will provide an upper bound on the cardinality of a minimum no-regret push up set~(\Cref{prop:No-regret_MinimumSet_UpperBound}) and subsequently show that this bound is, in fact, tight (\Cref{eg:TightBound}).

First, we recall some notation that was introduced earlier for the proof of \Cref{thm:No-regret_Minimal_Is_Minimum}.
Given any profile $\>$, let $P_{\>}$ denotes the set of all proposals that occur in the execution of the \DA{} algorithm on the profile $\>$. Formally, for any man $m_i \in M$ and woman $w_j \in W$, the ordered pair $(m_i,w_j)$ belongs to the set $P_{\>}$ if $m_i$ proposes to $w_j$ during the execution of \DA{} algorithm on the profile~$\>$. Additionally, let $P_{\> \setminus \>'} \coloneqq P_{\>} \setminus P_{\>'}$ denote the set of proposals that occur under the profile $\>$ but not under~$\>'$.

Given a woman $w_x$ and a profile $\>^x \coloneqq \{ \>_{-m}, \>^{w_x \uparrow}_m \}$, let $P'_x \coloneqq \{ (m_i, w_j) \in P_{\>^x \setminus \>} : m_i \neq  m \}$ denote the set of proposals that occur under $\>^x$ but not $\>$ made by any man other than the accomplice. Similarly, given a set of women $Y$ and a corresponding profile $\>^Y \coloneqq \{ \>_{-m}, \>^{Y \uparrow}_m \}$, let $P'_Y \coloneqq \{ (m_i, w_j) \in P_{\>^Y \setminus \>} : m_i \neq  m \}$. 

\begin{restatable}[\textbf{Upper bound}]{proposition}{NoRegretMinimumSetUpperBound}
Let $Y \subseteq W^\textsc{NR}$ be a minimum no-regret push up set. Then, $|Y| \leq \lfloor \frac{n-1}{2} \rfloor$.
\label{prop:No-regret_MinimumSet_UpperBound}
\end{restatable}

\begin{proof}
Suppose, by way of contradiction, that there exists some instance where $|Y| > \lfloor \frac{n-1}{2} \rfloor$. Let $Y \coloneqq \{ w_{y_1}, w_{y_2}, w_{y_3}, \dots \}$, and let $P^*_Y \coloneqq \{(m_i, w_j) \in P_{\>^Y} : w_j = \mu^Y(m_i) \neq \mu(m_i) \}$ denote the man-woman pairs who are matched under $\mu^Y$ but not under $\mu$ (i.e., the pairs in $\mu^Y \setminus \mu$).
Similarly, for all $w_{y_k} \in Y$, let $P^*_{y_k} \coloneqq \{(m_i, w_j) \in P_{\>^{y_k}} : w_j = \mu^{y_k}(m_i) \neq \mu(m_i) \}$, and let $P^* \coloneqq P^*_{y_1} \cup P^*_{y_2} \cup P^*_{y_3} \cup \dots$.

We first show that $P^*_Y = P^*$. Suppose there exists a pair $(m', w') \in P^*_Y \setminus P^*$. Then, $\mu^Y(m') \notin \{ \mu^{y_1}(m'), \mu^{y_2}(m'), \mu^{y_3}(m'), \dots \}$, posing a contradiction to \Cref{lem:PushUp_Meet}, which states that $\mu^Y = \mu^{y_1} \wedge \mu^{y_2} \wedge \dots$. Now, suppose there exists a proposal $(m', w') \in P^* \setminus P^*_Y$. Then, it must be the case that $\mu^{y_k}(m') \neq \mu^Y(m')$ for some $w_{y_k} \in Y$. From \Cref{lem:PushUp_Meet}, we then infer that $\mu^{y_k}(m') \>_{m'} \mu^{y_j}(m')$ for some $w_{y_j} \in Y \setminus \{w_{y_k}\}$. Note that $m' \neq m$ (i.e., $m'$ is truthful), because the accomplice's partner is unchanged between $\mu^Y$ and $\mu^{y_k}$ due to \Cref{lem:NoRegretSubset} and the no-regret assumption. Additionally, recall that $P'_{y_k} \coloneqq \{ (m_i, w_j) \in P_{\>^{y_k} \setminus \>} : m_i \neq m \}$ denotes the set of proposals that occur under $\>^{y_k}$ but not $\>$ made by any man other than the accomplice. Since men propose in decreasing order of their preference and $\>^{y_k}_{m'} \ = \ \>^{y_j}_{m'} \ = \ \>_{m'}$,
we infer that proposal $(m', \mu^{y_k}(m')) \in P'_{y_k} \cap P'_{y_j}$ 
However, this poses a contradiction to \Cref{lem:PairwiseDisjointSets}, which requires $P'_{y_k}$ and $P'_{y_j}$ to be pairwise disjoint. Thus, we must have that $P^*_Y = P^* = P^*_{y_1} \cup P^*_{y_2} \cup P^*_{y_3} \cup \dots$.

From \Cref{prop:PushUp}, it follows that $\mu(m) \succeq_{m} \mu^{y_k}(m)$ for all $m \in M$ and all $w_{y_k} \in Y$.
Recall that for every man $m' \in M$ such that $(m', w') \in P^*_{y_k}$ for some woman $w' \in W$, we have $\mu(m') \neq \mu^{y_k}(m')$, and thus $\mu(m') \>_{m'} \mu^{y_k}(m')$.
Additionally, we have that $m' \neq m$ (i.e., $m'$ is truthful) because $m$ matches with $\mu(m)$ under $\>^{y_k}$ due to the no-regret assumption of $w^{y_k} \in Y$. Since men propose in decreasing order of their preference and $\>^{y_k}_{m'} \ = \ \>_{m'}$, we infer that for all $w_{y_k} \in Y$, $P^*_{y_k} \subseteq P_{\>^{y_k} \setminus \>}$, where the man-woman pairs in $P^*_{y_k}$ translate to proposals in $P_{y_k \setminus \>}$. We also know that $P^*_{y_k}$ does not contain any pairs made by the accomplice $m$, which means that for all $w_{y_k} \in Y$, $P^*_{y_k} \subseteq P'_{y_k}$. This, combined with \Cref{lem:PairwiseDisjointSets}, which says that the sets $P'_{y_1}, P'_{y_2}, P'_{y_3}, \dots$ are pairwise disjoint, shows that the sets $P^*_{y_1}, P^*_{y_2}, P^*_{y_3}, \dots$ are also pairwise disjoint. Further, since $P^*_Y = P^*_{y_1} \cup P^*_{y_2} \cup P^*_{y_3} \dots$, we get that the sets $P^*_{y_1}, P^*_{y_2}, P^*_{y_3}, \dots$ constitute a partition of the set $P^*_Y$.

Since the accomplice $m$ matches with $\mu(m)$ under $\>^{y_k}$ for all $w_{y_k} \in Y$, it follows from \Cref{prop:StrictPushUp} that $\mu(m') \>_{m'} \mu^{y_k}(m')$ for at least two distinct men $m' \in M \setminus \{m\}$. This implies that for all $w_{y_k} \in Y$, $|P^*_{y_k}| \geq 2$. The number of $P^*_{y_1}, P^*_{y_2}, P^*_{y_3}, \dots$ sets is clearly $|Y|$, which we assumed, for contradiction, to be greater than $\lfloor \frac{n-1}{2} \rfloor$. Note that this equivalent to saying $|Y| > \frac{n-1}{2}$, because $|Y|$ must be a non-negative integer.
Since $P^*_{y_1}, P^*_{y_2}, P^*_{y_3}, \dots$ are pairwise disjoint sets, this implies that $|P^*_Y| = |P^*| > \frac{n-1}{2} \times 2 = n-1$. 
However, from \Cref{lem:NoRegretSubset}, we know that $\>^Y$ is a no-regret profile for the accomplice $m$, which implies that $P^*_Y$ can contain at most $n-1$ pairs. Thus, we obtain a contradiction.
\end{proof}

In \Cref{eg:TightBound}, we show that the bound in \Cref{prop:No-regret_MinimumSet_UpperBound} is tight by presenting a family of instances for which the cardinality of the minimum no-regret set is $\lfloor \frac{n-1}{2} \rfloor$.

\begin{example}[\textbf{Lower bound}] 
Consider the following family of preference profiles where the \DA{} outcome is underlined. The notation ``\dots'' denotes that the order of that portion of the preference list may be arbitrary. 

\begin{table}[H]
    \centering
    \begin{tabularx}{0.75\linewidth}{XXXXXXXXXXXXXXX}
            $\boldsymbol{\textcolor{blue}{{m_1}}}\colon$ & $\underline{w_1^*}$ & $\dots$ & && ${w_1}\colon$ & $\underline{m_1^*}$ & $\dots$\\
            ${m_2}\colon$ & $\underline{w_2}$ & $w_3^*$ & $\dots$ && ${w_2}\colon$ & $m_3^*$ & $m_1$ & $\underline{m_2}$ & $\dots$\\
            ${m_3}\colon$ & $\underline{w_3}$ & $w_2^*$ & $\dots$ && ${w_3}\colon$ & $m_2^*$ & $\underline{m_3}$ & $\dots$\\
            $\vdots$ & & & && $\vdots$\\
            ${m_i}\colon$ & $\underline{w_i}$ & $w_{i+1}^*$ & $\dots$ && ${w_i}\colon$ & $m_{i+1}^*$ & $m_1$ & $\underline{m_i}$ & $\dots$\\
            ${m_{i+1}}\colon{}$ & $\underline{w_{i+1}}$ & $w_i^*$ & $\dots$ && ${w_{i+1}}\colon$ & $m_i^*$ & $\underline{m_{i+1}}$ & $\dots$\\
            $\vdots$ & & & && $\vdots$\\
            ${m_{n-1}}\colon$ & $\underline{w_{n-1}}$ & $w_n^*$ & $\dots$ && ${w_{n-1}}\colon$ & $m_n^*$ & $m_1$ & $\underline{m_{n-1}}$ & $\dots$\\
            ${m_n}\colon$ & $\underline{w_n}$ & $w_{n-1}^*$ & $\dots$ && ${w_n}\colon$ & $m_{n-1}^*$ & $\underline{w_n}$ & $\dots$\\
        \end{tabularx}
\end{table}

Suppose the accomplice is $m_1$. We first consider the case where $n$ is odd. The \DA{} matching after $m_1$ pushes up the set $Y \coloneqq \{w_i \in W : i \text{ is even}\}$ is marked by $*$. Clearly, the resulting matching is optimal for $W$, since all women are matched with their top choice partners. Further, we can show that $Y$ is a \emph{minimum} no-regret push up set.

First, we show that no woman in the set $\bar{Y} \coloneqq \{w_i \in W : i \neq 1 \text { and $i$ is odd}\}$ belongs in a minimum push up set. Suppose, for contradiction, that this is not true. Then, let $Z$ be a minimum push up set that contains a woman, say $w_k$, in $\bar{Y}$, and let $Z' \coloneqq Z \setminus {w_k}$. Additionally, let $\>^Z$ and $\>^{Z'}$ be the profiles after $m_1$ pushes up the sets $Z$ and $Z'$, respectively. Since $Z \subseteq W^\textsc{NR}$ (due to our assumption that $Z$ is a minimum push up set and thus by definition is also a no-regret set) and $Z' \subseteq Z$, \Cref{lem:NoRegretSubset} implies that $\>^{Z}$ and $\>^{Z'}$ are both no-regret profiles. Therefore, it follows from \Cref{prop:PushUp} that $\mu^{Z'} \succeq_W \mu$, where $\mu^{Z'} \coloneqq \DA(\>^{Z'})$ and $\mu$ is the underlined matching in the profile above. Then, it must be that $\mu^{Z'}(w_k) \succ_{w_k} m_1$, because $\mu(w_k) \>_{w_k} m_1$ due to our assumption that $w_k \in \bar{Y}$. Since profile $\>^{Z}$ can be obtained from $\>^{Z'}$ via a no-regret push up of $w_k$, \Cref{prop:WeakPushUp} indicates that $\mu^{Z} = \mu^{Z'}$, which poses a contradiction to our assumption that $Z$ is a minimum push up set. Thus, we have shown no woman in the set $\bar{Y} \coloneqq \{w_i \in W : i \neq 1 \text { and $i$ is odd}\}$ belongs in a minimum push up set. Given this and that $w_1$ cannot belong in a push up set (since $\mu(m_1) = w_1$), if $Y$ is not a minimum push up set, then it must be that $Y$ is not a minimal push up set. However, it can easily be seen that removing any woman from $Y$ would be suboptimal. Indeed, if $m_1$ were to push up $Y \setminus \{w_{i}\}$, for instance, then $w_i$ and $w_{i+1}$ would both be strictly worse off than if he were push up the entire set $Y$. Thus, we have shown that $Y$ is a minimum push up set.

We also consider the case where $n$ is even. If we add ``dummy'' agents $m_0$ and $w_0$ such that $\>_{m_0} \coloneqq w_0 \dots$ and $\>_{w_0} \coloneqq m_0 \dots$ to the profile instance above, then the same arguments from the case where $n$ is odd hold.\qed
\label{eg:TightBound}
\end{example}

\section{Omitted Material from Section~\ref{sec:experiments}}
\label{sec:Additional_Experiments}

\paragraph{Comparing the quality of partners.}
In this section, we will compare the quality of partners under the one-sided and two-sided manipulation strategies for both the one-for-all and two-for-one models.

Recall that we generated $1\,000$ preference profiles uniformly at random  for each value of $n \in \{2, 4, \dots, 20\}$ (where $n$ is the number of men/women). For the one-for-all model, we choose the two-sided and one-sided strategies that maximize the sum of improvement in ranks of the matched partners of all women for each instance. For the two-for-one model, on the other hand, we only consider the instances where one-sided (respectively, two-sided) manipulation is better than truth telling.

\Cref{fig:two_sided_quality_partners_detailed} presents the results of this experiment along with the additional information about outliers.
Note that we count a value as an outlier if it is more than 1.5 times the interquartile range from either the lower or upper quartile.

\begin{figure*}%
    \centering
    \begin{minipage}{0.45\linewidth}
        \centering
        \includegraphics[width=\textwidth]{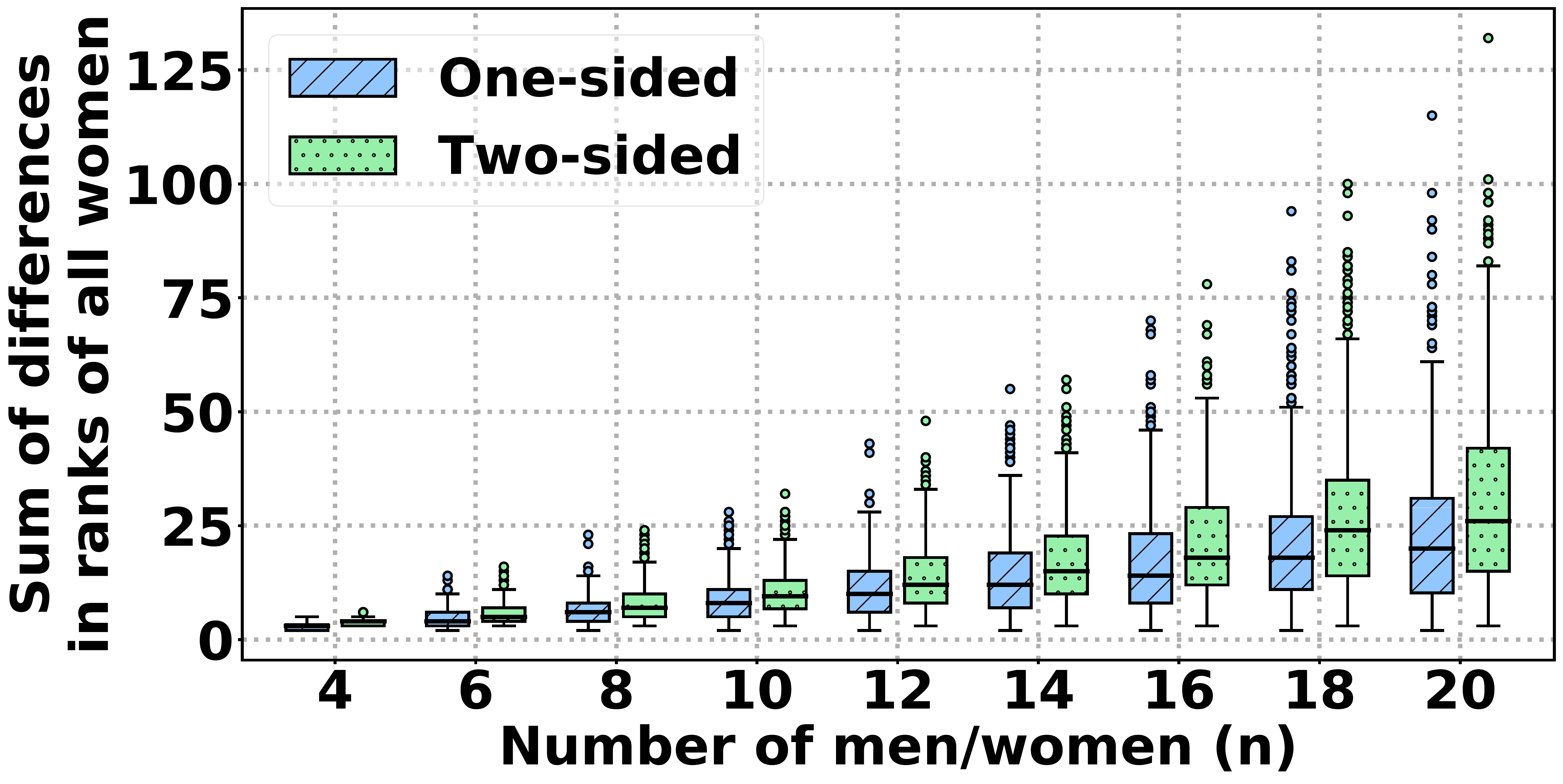}
    \end{minipage}
    \hfill
    \begin{minipage}{0.45\linewidth}
        \centering
        \includegraphics[width=\textwidth]{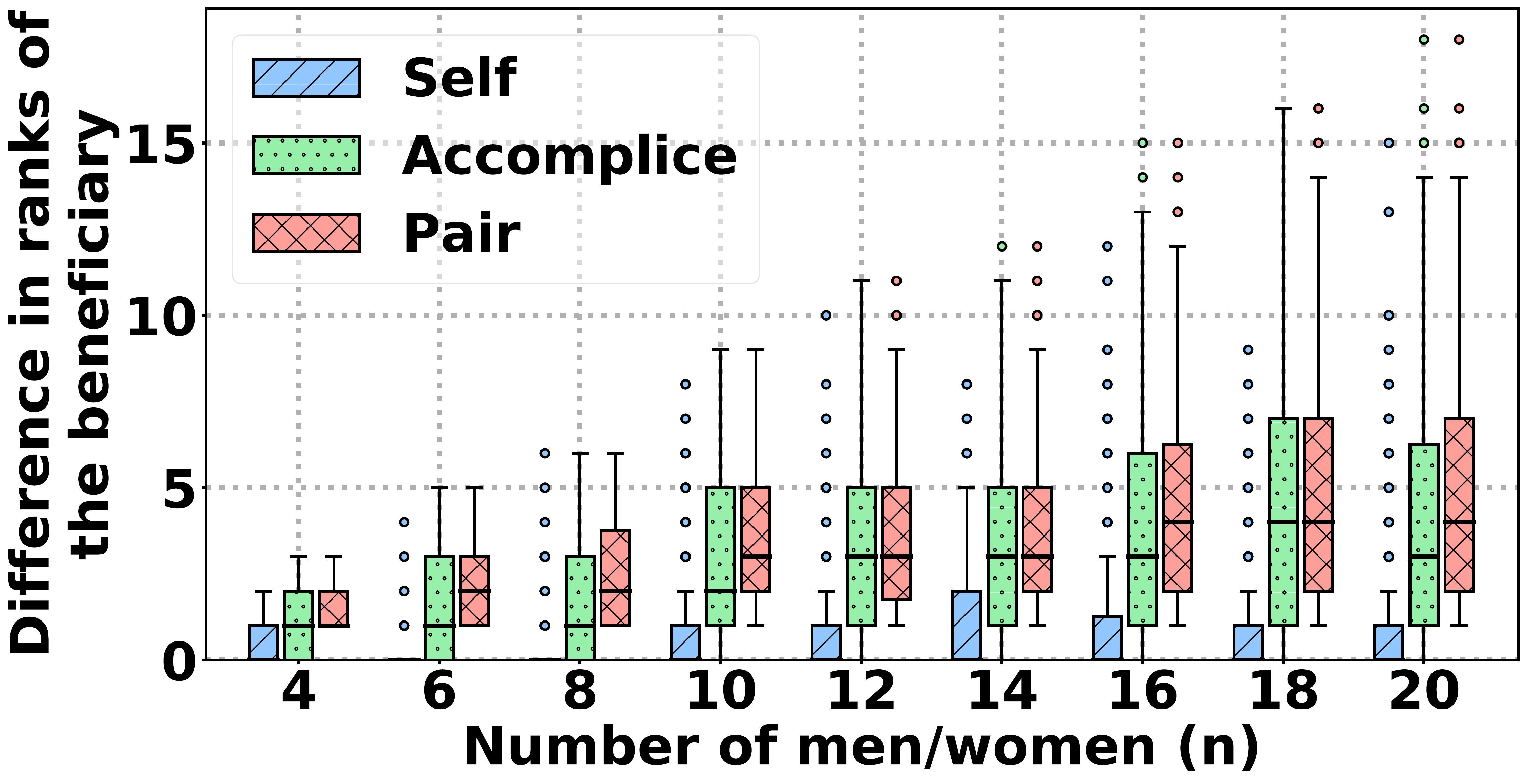}
    \end{minipage}
    \caption{Comparing one-sided and two-sided strategies against truthful reporting when \emph{all} women are beneficiaries (left) and comparing pair, self, and accomplice manipulation strategies for a \emph{single} beneficiary $w$ (right) in terms of improvement in ranks of their matched partners. The solid bars, whiskers, and dots denote the interquartile range, range excluding outliers, and outliers, respectively.}
     \label{fig:two_sided_quality_partners_detailed}
\end{figure*}

\Cref{fig:two_sided_quality_partners_detailed} shows that two-sided strategies result in \emph{better matched partners} than one-sided ones under both coalitional generalizations. The fact that two-sided outperforms one-sided under the two-for-one setting is not surprising (since the strategy space is strictly larger); however, our results highlight that the gains are non-trivial.

\paragraph{Comparing one-for-all strategies for larger numbers of agents.}
We recreate the experiments that compare the frequency of availability as well as the quality of partners of one-sided and two-sided strategies in the one-for-all model for larger values of $n$.\footnote{Note that due to the larger run-time for the pair manipulation algorithm (\Cref{thm:PairManipulation}), we did not do the same for our two-for-one experiments.}
More specifically, we generate $1\,000$ uniformly random preference profiles for each value of $n \in \{10, 20, \dots, 100\}$.
\Cref{fig:OneForAll_n=100} (left) illustrates the fractions of instances where some man can misreport in a way that is beneficial for women as well as the analogous fractions where some woman can help all women.%

We additionally show the distributions of improvement (in terms of the sum of rank differences for all women) for the one-sided (respectively, two-sided) strategy that maximizes rank improvement. The one-sided (respectively, two-sided) manipulation boxplots in \Cref{fig:OneForAll_n=100} (middle) only reflect the data for when one-sided (respectively, two-sided) manipulation is successful. 

Further, we show the distributions of women strictly better off after the one-sided (respectively, two-sided) strategy that maximizes the number of women who improve their match. Again, the one-sided (respectively, two-sided) manipulation boxplots in \Cref{fig:OneForAll_n=100} (right) only reflect the data for when one-sided (respectively, two-sided) manipulation is successful. 

It is evident that two-sided manipulation is not only more prevalent than one-sided manipulation but also, in expectation, is more likely to help more woman and get them matched with better quality partners.

\begin{figure*}%
    \centering
    \begin{minipage}{0.32\linewidth}
        \centering
        \includegraphics[width=\textwidth]{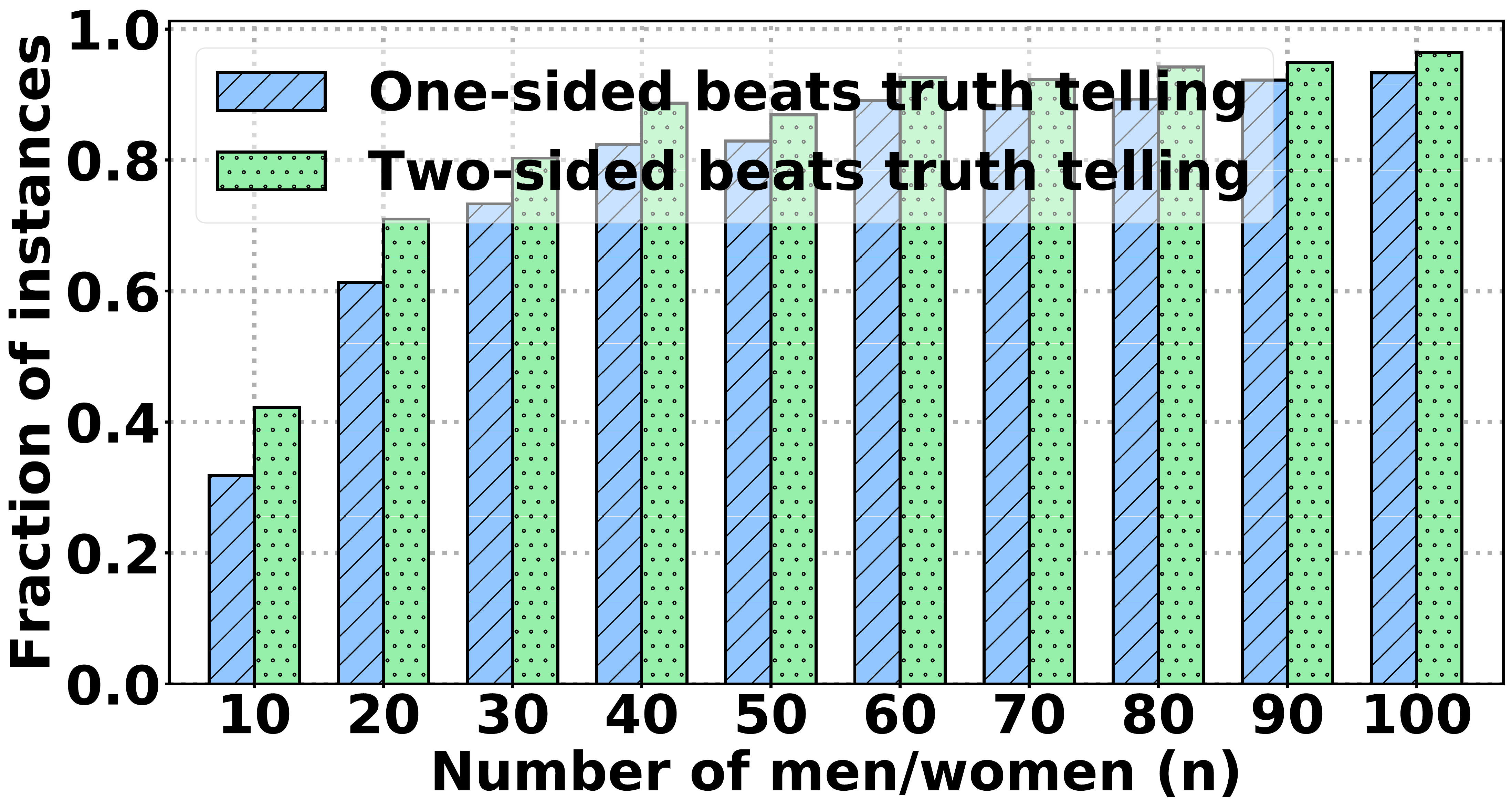}
    \end{minipage}
    \hfill
    \begin{minipage}{0.32\linewidth}
        \centering
        \includegraphics[width=\textwidth]{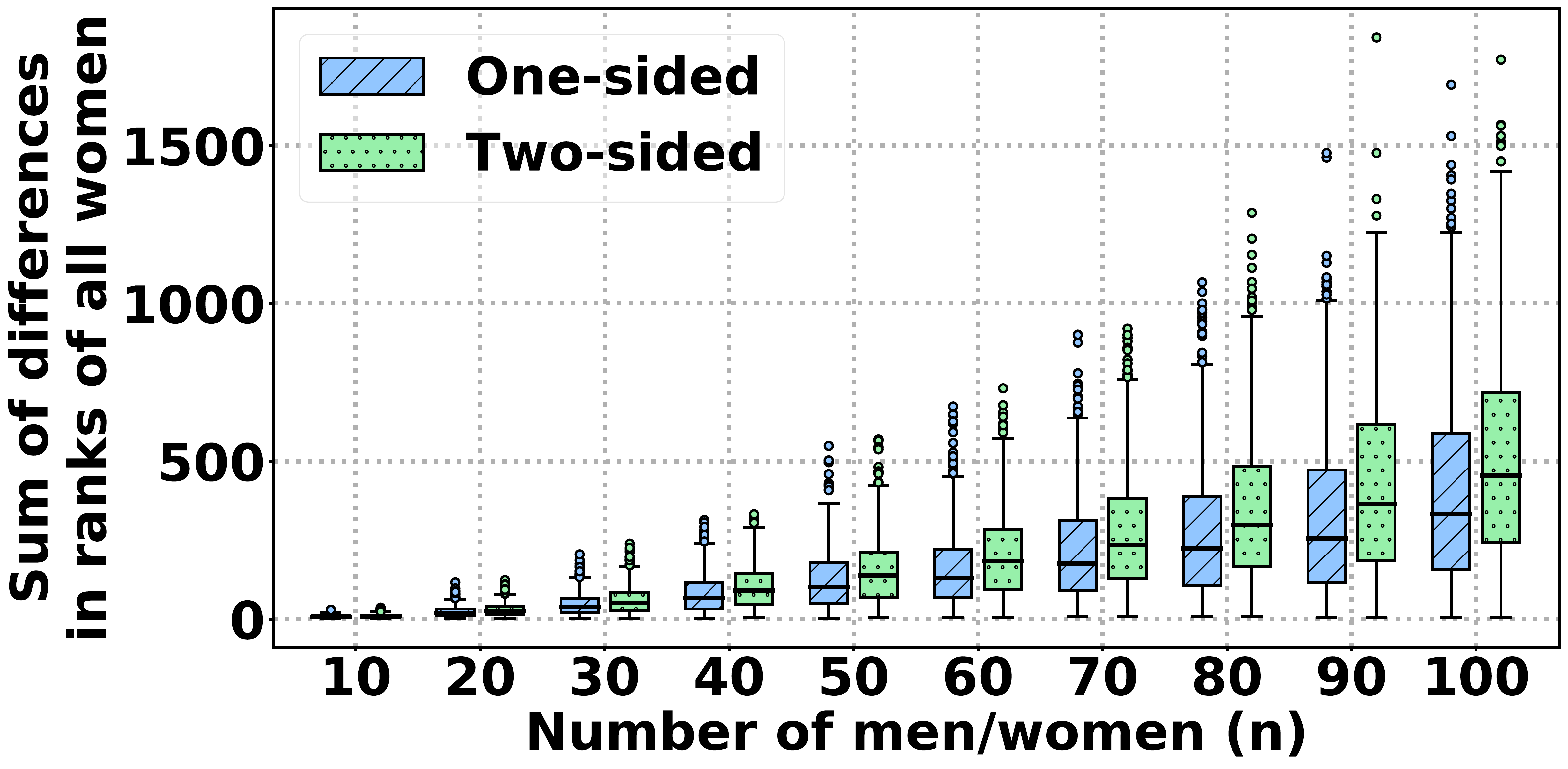}
    \end{minipage}
    \hfill
    \begin{minipage}{0.32\linewidth}
        \centering
        \includegraphics[width=\textwidth]{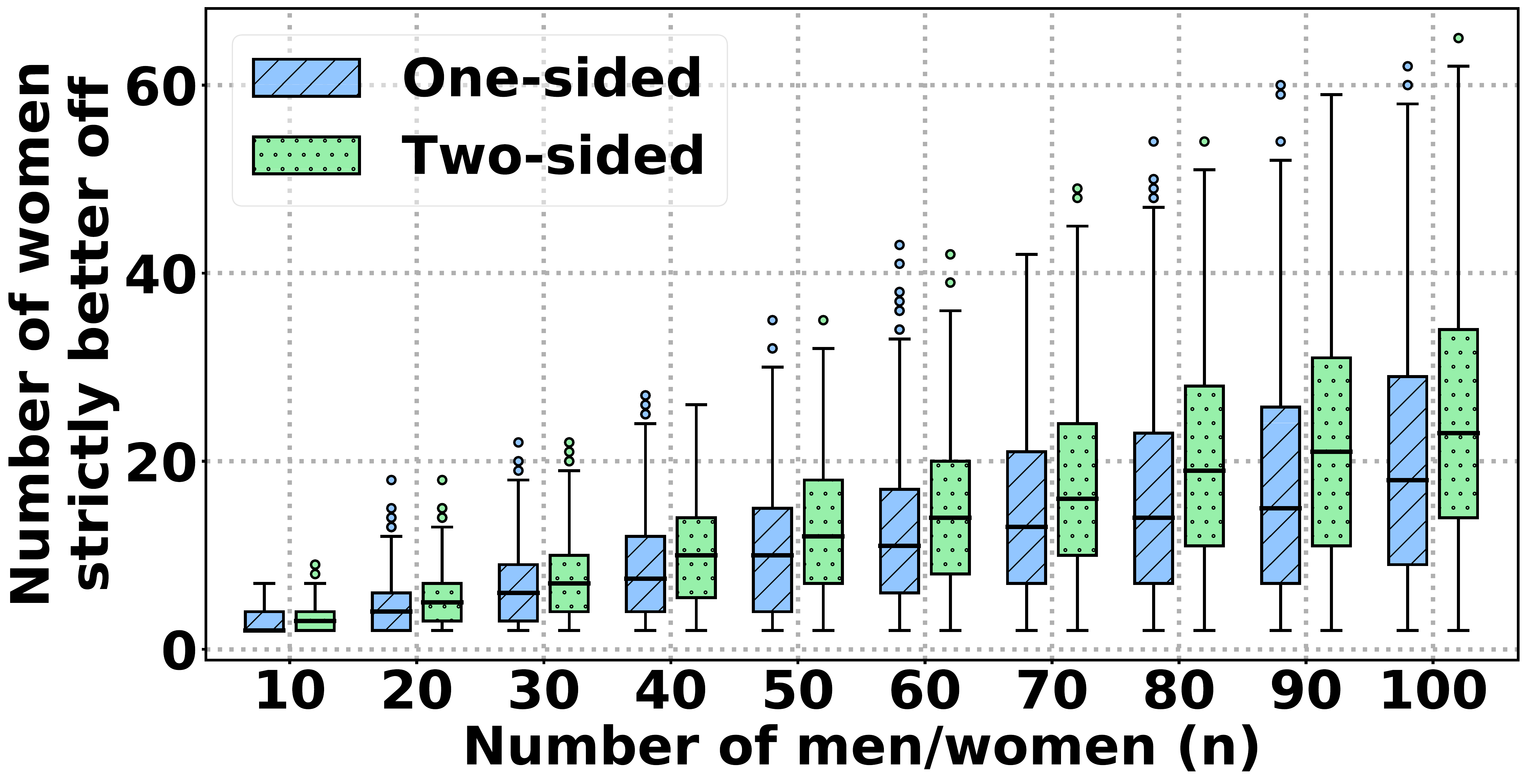}
    \end{minipage}
    \caption{Comparing one-sided and two-sided strategies against truthful reporting for $n = 100$ when \emph{all} women are beneficiaries in terms of 1) frequency of occurrence (left), 2) sum of improvement in the ranks of the new partners of all women (middle), and 3) number of women strictly better off after manipulation. The solid bars, whiskers, and dots denote the interquartile range, range excluding outliers, and outliers, respectively.}
     \label{fig:OneForAll_n=100}
\end{figure*}

\paragraph{How conspicuous is the accomplice's one-for-all list?}
Although optimal one-for-all accomplice manipulation is not guaranteed to be inconspicuous (\Cref{eg:inconspicuous_is_suboptimal}), we show that, in expectation, the manipulated list of the accomplice may not look very different from the true list. For $10\,000$ uniformly random preference profiles with 20 men and 20 women, we tabulate the minimum number of women who must be pushed up to create an optimal strategy (computed using \Cref{alg:MinimizePushUpSet}) for some fixed accomplice $m$ in \Cref{tab:WomenWhoImprove}.\footnote{Unlike for our previous experiments, we ran $10\,000$ instances here in order to identify a clearer pattern in the data.}

\begin{table}[h]
\centering
\begin{tabular}{l|cccc} 
    Size of minimum push up set & 0 & 1 & 2 & 3\\
    \hline%
    No. of instances & 7952 & 1947 & 100 & 1\\
\end{tabular}%
\caption{The number of instances (out of $10,000$) where there are varying numbers of women in the minimum-sized no-regret push up set when $n$ = 20.}
\label{tab:WomenWhoImprove}
\end{table}

We similarly computed the size of the minimum push up set for some fixed man $m$ for each of $1\,000$ preference profiles generated uniformly at random for each value of $n \in \{10, 20, \dots, 100\}$. \Cref{fig:PushUpSetSize} illustrates the distributions of the size of the sets out of only those instances where manipulation is strictly better than truthful reporting. The results suggest that the expected size of the push up set does not grow at the same rate as the number of agents.

\begin{figure}[ht]
    \centering
    \includegraphics[width=0.5\linewidth]{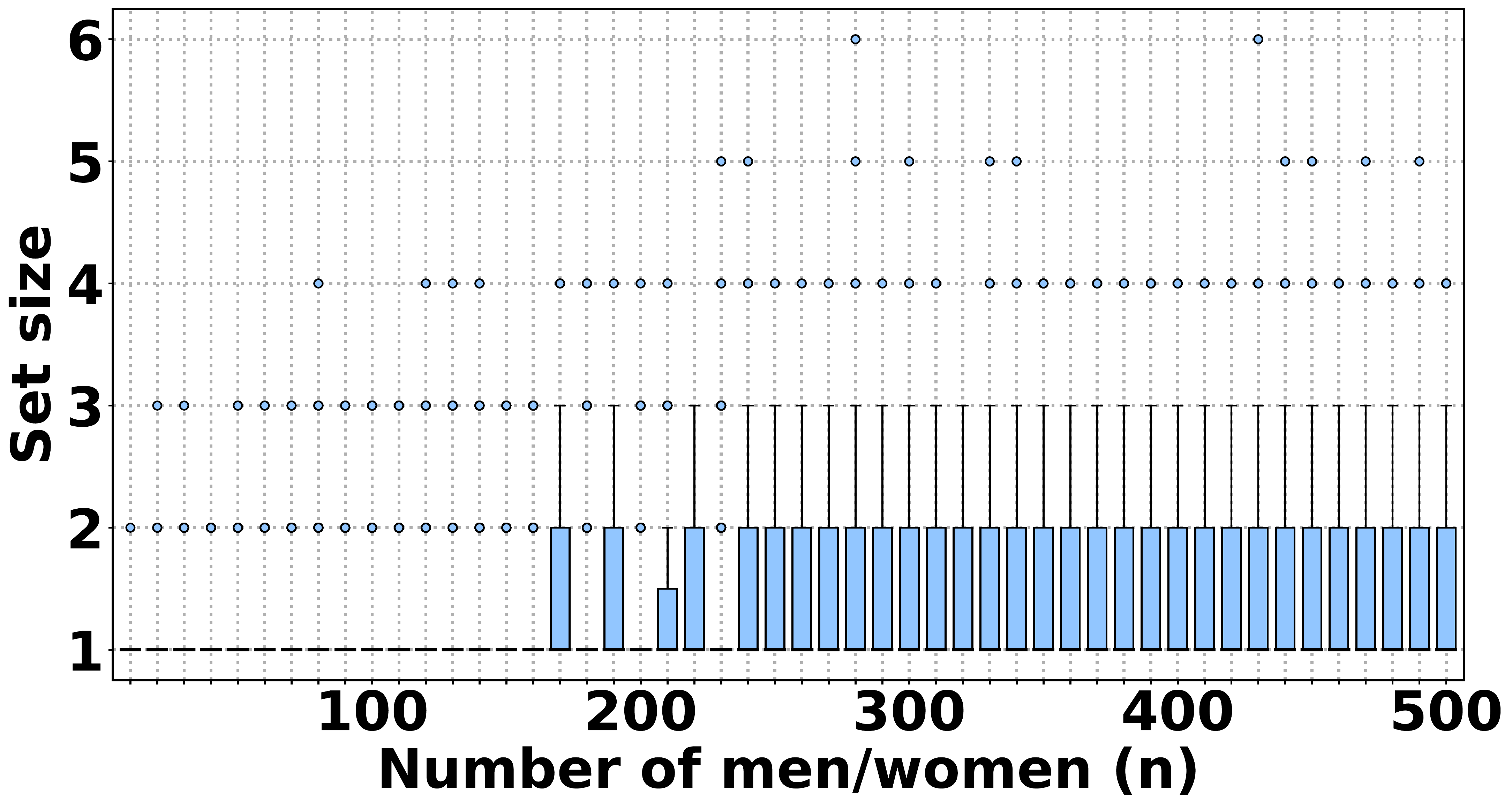}
    \caption{Distributions of the size of minimum no-regret push up sets. The solid bars, whiskers, and dots denote the interquartile range, excluding outliers, and outliers, respectively.}
\label{fig:PushUpSetSize}
\end{figure}

\end{document}